\documentclass[a4paper,11pt,reqno]{amsart}

\usepackage[utf8]{inputenc}
\usepackage[english]{babel}
\usepackage{amsmath}
\usepackage{amssymb}
\usepackage{amsfonts}
\usepackage{amsthm}
\usepackage{enumitem}
\usepackage[colorlinks=false]{hyperref}
\usepackage{pgf,tikz}
\usetikzlibrary{arrows}
\usetikzlibrary{decorations.markings,decorations.pathmorphing,patterns,positioning}

\newcommand{\wto}{\rightharpoonup}
\newcommand{\eps}{\varepsilon}
\newcommand{\epsj}{{\varepsilon_j}}
\newcommand{\mc}{\mathcal}
\newcommand{\xeps}{{x^\eps}}
\newcommand{\xepsd}{{\dot{x}^\eps}}
\newcommand{\xepsdd}{{\ddot{x}^\eps}}
\newcommand{\xepsj}{{x^\epsj}}
\newcommand{\xepsjd}{{\dot{x}^\epsj}}
\newcommand{\xepsjdd}{{\ddot{x}^\epsj}}
\newcommand{\NN}{\mathcal{N}}
\newcommand{\EE}{\mathcal{E}}
\newcommand{\RR}{\mathcal{R}}
\newcommand{\CC}{\mathcal{C}}
\newcommand{\KK}{\mathcal{K}}
\newcommand{\UU}{\mathcal{U}}
\newcommand{\BB}{\mathcal{B}}
\newcommand{\A}{\mathbb{A}}
\newcommand{\M}{\mathbb{M}}
\newcommand{\V}{\mathbb{V}}
\newcommand{\R}{\mathbb{R}}
\newcommand{\N}{\mathbb{N}}
\newcommand{\Q}{\mathbb{Q}}
\newcommand{\piz}{\pi_Z}
\renewcommand{\to}{\rightarrow}
\renewcommand{\d}{\,\mathrm{d}}
\newcommand{\abs}[1]{\left\lvert #1 \right\rvert}      
  
\newcommand{\scal}[2]{\left\langle #1,#2\right\rangle}
\DeclareMathOperator{\diag}{Diag}
\DeclareMathOperator{\Id}{I}

\numberwithin{equation}{section}
\newtheorem{thm}{Theorem}[section]
\newtheorem{defi}[thm]{Definition}
\newtheorem{prop}[thm]{Proposition}
\newtheorem{lemma}[thm]{Lemma}
\newtheorem{cor}[thm]{Corollary}

\theoremstyle{definition}
\newtheorem{rmk}[thm]{Remark}

\begin{document}
	
	\author{Paolo Gidoni and Filippo Riva}
	
	\title[Vanishing inertia analysis for finite dimensional rate-independent systems]{A vanishing inertia analysis for finite dimensional rate-independent systems with nonautonomous dissipation, and an application to soft crawlers} 
	
	\begin{abstract}
	We study the approximation of quasistatic evolutions, formulated as abstract finite-dimensional rate-independent systems, via a vanishing-inertia asymptotic analysis of dynamic evolutions. We prove the uniform convergence of dynamical solutions to the quasistatic one, employing the concept of energetic solution. Motivated by applications in soft locomotion, we allow time-dependence of the dissipation potential, and translation invariance of the potential energy. 
	\end{abstract}

	\maketitle
		
	{\small
		\par
		}

	\pagenumbering{arabic}
	
\medskip

\tableofcontents
\addtocontents{toc}{\setcounter{tocdepth}{1}}

\section{Introduction and motivation} \label{sec:intro}

The approximation and the selection of quasistatic evolutions, via the asymptotic analysis of richer and more natural viscous or dynamic problems, has been intensively and increasingly investigated from a rigorous mathematical perspective in the last two decades. If on one hand the vanishing viscosity approach, concerning the limit behaviour of a first order singular perturbation of the quasistatic model, has been widely studied and discussed (see for instance \cite{AgoRos,EfMielk} in finite dimension,  \cite{MielkRosSav09,MielkRosSav12,MielkRosSav16} for  abstract analyses in infinite dimension and the concept of balanced-viscosity solutions, and the recent comparison in \cite{KreMon} of different approaches), on the other hand the second order analysis dealing with inertial systems still offers open questions and hard challenges.\par 

In this latter direction we may identify two main lines of investigation.
The first family of results is inspired by physical models where the quasistatic evolution is defined by a driving potential and a rate-independent dissipation. We mention for instance: an approximation of perfect elastoplasticity by suitable dynamic viscoelasto-plastic problems \cite{DMSca}; a vanishing inertia limit in models of dynamic debonding \cite{LazNar,Rivquas}; a vanishing inertia and viscosity limits for a delamination model \cite{Sca}, or for damage in a thermo-viscoelastic material \cite{LRTT}; a realisation of fully rate-independent system for viscoelastic solids \cite{Roubvaninertia} or systems with hardening \cite{MielkPetr, MarMon} as inertia vanishes.

A second approach deals with vanishing inertia (and viscosity) approximation of quasistatic evolutions driven by a potential energy alone, aiming at a deeper comprehension  in an abstract but finite-dimensional setting. Starting from \cite{Ago, Nar15} and culminating with \cite{ScilSol}, a detailed description of the limiting evolution coming from a vanishing inertia and viscosity procedure has been given.\par

In this paper we contribute to the topic in a third, intermediate direction, introducing new features with respect to both approaches. More precisely, we derive abstract rate-independent systems of the form
\begin{equation}\label{eq:intro_stat}
\partial_v \mc R(t,\dot{x}(t))+D_x\mc E(t,x(t))\ni 0,
\end{equation}
as the limit, for $\eps\to 0^+$, of the dynamic problem
\begin{equation}\label{eq:intro_dyn}
\eps^2 \mathbb{M}\xepsdd(t)+\eps \mathbb{V}\xepsd(t)+\partial_v \mc R(t,\xepsd(t))+D_x\mc E(t,\xeps(t))\ni 0.
\end{equation}
Here $\mc E$ is a driving potential energy, $\RR$ a time-dependent dissipation potential (one homogeneous in space to ensure rate-independence of \eqref{eq:intro_stat}), while $\mathbb{M}$ is a symmetric positive-definite operator representing masses, and $\mathbb{V}$ a positive-semidefinite (hence, possibly $\mathbb{V}=0$) operator describing the possible presence of viscosity in the model. Systems of the form \eqref{eq:intro_stat} are usually referred as \emph{rate-independent systems} \cite{MielkRoubbook}.

Our framework is motivated by an emergent application in soft locomotion \cite{Gid18}, which we discuss later in this section. There are however several elements of novelty with respect to the other applications cited above. Firstly, usually in such models  rate-independent dissipation and inertia act on two disjoint variables: more precisely, the mass operator $\M$ is null on a subspace, and the rate-independent dissipation depends only on the kernel of $\M$ (this is very clear for instance in \cite{MielkPetr}). The opposite occurs instead in our dynamic problem \eqref{eq:intro_dyn}, since the matrix $\M$ is nondegenerate and we will assume a positive dissipation for each change in the state, namely $\RR(t,v)> 0$ for every $v\neq0$.
Also, up to our knowledge, our paper is the first to study a nonautonomous dissipation in a vanishing inertia limit, while in all the references above time-dependence is assumed only in the potential energy. Indeed, even in the quasistatic setting, the case of a nonautonomous functional $\RR$ has been considered only very recently within the theory of rate-independent systems \cite{HeiMielk} (see also \cite{Gid18, Mielk18} for applications), even if it was well discussed in the special framework of sweeping processes \cite{KunMon}.

Our finite dimensional setting with invertible $\M$ may therefore seem closer to the approach of \cite{Ago,Nar15,ScilSol}. We notice however that adding a rate-independent potential to the quasistatic evolution highly affects the structure of the problem. The most evident consequence is that we may neglect viscosity from our analysis, while it is crucial in \cite{ScilSol}, where $\mathbb{V}$ has to be positive definite. Indeed, the key point seems to be that at least one kind of dissipation must be included in the model, otherwise kinetic effects persist in the limit precluding the resulting evolution to be rate-independent. Compare for instance \cite{LazNar} with \cite{Rivquas}, or \cite{Nar15} with \cite{ScilSol}, where the addition of viscous terms makes dynamic solutions converge to the quasistatic ones, in contrast with the undamped case where counterexamples are shown.

Notice also that, with respect to \cite{ScilSol}, we are able to weaken the regularity assumptions on the energy $\EE$, even if we require convexity in order to complete our argument (it will however not be needed in the first part of the investigation).

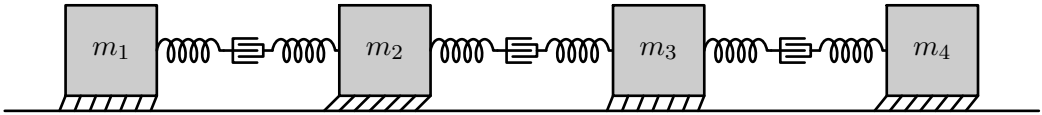
\begin{figure}
    \centering
    	\centering
	\begin{tikzpicture}[line cap=round,line join=round,x=4mm,y=4mm, line width=1pt]
	\clip(-3,-2) rectangle (33,7);
	\draw [line width=1pt, fill=gray!40] (0,0.5)-- (3,0.5)--(3,3.5)-- (0,3.5)-- (0,0.5);
	\draw[decoration={aspect=0.5, segment length=1.5mm, amplitude=1.5mm,coil},decorate] (3,2)-- (5.2,2.);
	\draw[decoration={aspect=0.5, segment length=1.5mm, amplitude=1.5mm,coil},decorate] (6.8,2)-- (9.,2.);
	\draw (5.2,2)--(6.3,2);
	\draw (6.5,2)--(6.8,2);
	\draw (6.3,1.6)-- (5.5,1.6)--(5.5,2.4)--(6.3,2.4);
	\draw (5.7,1.8)-- (6.5,1.8)--(6.5,2.2)--(5.7,2.2);
	\draw [line width=1pt, fill=gray!40] (9,3.5)-- (9,0.5)-- (12,0.5)-- (12,3.5)-- (9,3.5);
		\draw[decoration={aspect=0.5, segment length=1.5mm, amplitude=1.5mm,coil},decorate] (12,2)-- (14.2,2.);
	\draw[decoration={aspect=0.5, segment length=1.5mm, amplitude=1.5mm,coil},decorate] (15.8,2)-- (18.,2.);
	\draw (14.2,2)--(15.3,2);
	\draw (15.5,2)--(15.8,2);
	\draw (15.3,1.6)-- (14.5,1.6)--(14.5,2.4)--(15.3,2.4);
	\draw (14.7,1.8)-- (15.5,1.8)--(15.5,2.2)--(14.7,2.2);
	\draw [line width=1pt, fill=gray!40] (18,3.5)-- (18,0.5)-- (21,0.5)-- (21,3.5)-- (18,3.5);
	\draw[decoration={aspect=0.5, segment length=1.5mm, amplitude=1.5mm,coil},decorate] (21,2)-- (23.2,2.);
	\draw[decoration={aspect=0.5, segment length=1.5mm, amplitude=1.5mm,coil},decorate] (24.8,2)-- (27.,2.);
	\draw (23.2,2)--(24.3,2);
	\draw (24.5,2)--(24.8,2);
	\draw (24.3,1.6)-- (23.5,1.6)--(23.5,2.4)--(24.3,2.4);
	\draw (23.7,1.8)-- (24.5,1.8)--(24.5,2.2)--(23.7,2.2);
	\draw [line width=1pt, fill=gray!40] (27,3.5)-- (27,0.5)-- (30,0.5)-- (30,3.5)-- (27,3.5);
	\draw [] (-2,0)-- (32,0);	
	\draw [] (0,0.5)-- (-0.2,0);
	\draw [] (0.5,0.5)-- (0.3,0);
	\draw [] (1,0.5)-- (0.8,0);
	\draw [] (1.5,0.5)-- (1.3,0);
	\draw [] (2,0.5)-- (1.8,0);
	\draw [] (2.5,0.5)-- (2.3,0);
	\draw [] (3,0.5)-- (2.8,0);
	\draw [] (9,0.5)-- (8.5,0);
	\draw [] (9.5,0.5)-- (9.,0);
	\draw [] (10,0.5)-- (9.5,0);
	\draw [] (10.5,0.5)-- (10.,0);
	\draw [] (11,0.5)-- (10.5,0);
	\draw [] (11.5,0.5)-- (11.,0);
	\draw [] (12,0.5)-- (11.5,0);
	\draw [] (18,0.5)-- (17.8,0);
	\draw [] (18.5,0.5)-- (18.3,0);
	\draw [] (19,0.5)-- (18.8,0);
	\draw [] (19.5,0.5)-- (19.3,0);
	\draw [] (20,0.5)-- (19.8,0);
	\draw [] (20.5,0.5)-- (20.3,0);
	\draw [] (21,0.5)-- (20.8,0);
	\draw [] (27,0.5)-- (26.6,0);
	\draw [] (27.5,0.5)-- (27.1,0);
	\draw [] (28,0.5)-- (27.6,0);
	\draw [] (28.5,0.5)-- (28.1,0);
	\draw [] (29,0.5)-- (28.6,0);
	\draw [] (29.5,0.5)-- (29.1,0);
	\draw [] (30,0.5)-- (29.6,0);
	\draw (1.5,2) node[] {$m_1$};
	\draw (10.5,2) node[] {$m_2$};
	\draw (19.5,2) node[] {$m_3$};
	\draw (28.5,2) node[] {$m_4$};
	\end{tikzpicture}
    \caption{A model of soft crawler, discussed in Subsection \ref{subsec:crawler}.}
    \label{fig:crawler4}
\end{figure}

\subsubsection*{The motivating model}
Our work is motivated by an application to a discrete model of soft crawler \cite{Gid18,ColGid}. Crawling encompasses the motility strategies employed by several animals, such as earthworms and leeches, and by biomimetic robots. Usually, a crawler can be effectively modelled as a chain of material points on a line, each subject to dry friction. The case $N=4$ is portrayed in Figure \ref{fig:crawler4}. The attribute \emph{soft} is due to the fact that each couple of adjacent masses is joined by an elastic, actuated link. By a mathematical point of view, this means that the actual shape of the locomotor is not directly prescribed, but undergoes to hysteresis. Soft actuation is widespread in Nature, where soft bodies and soft body parts, compliant joints and soft shells are the norm. This is even more evident for worm-like locomotion: for instance earthworms and leeches are entirely soft-bodied, while no lever action on the skeleton is employed by snakes during rectilinear locomotion. The properties of compliance and adaptability to a continuously variable and unstructured environment, observed in (soft) animal locomotion, have caught the attention of engineers, leading, in the last two decades, to the design of bio-inspired robotic locomotors -- crawling, but also swimming, running, etc. --  and the development of the novel field of \emph{soft robotics} \cite{Cal,Kim}. In robotic crawlers, soft actuation may be implemented in several ways. One common approach is to couple an elastic structure with a motor-tendon actuator \cite{Vikas}, or a coil made of shape-memory alloy \cite{Seok}, or a pneumatic actuator \cite{Raf18}. Alternatively, also links made of nematic elastomers have been successfully employed to provide both elasticity and actuation \cite{Jung}. 

In addition to the soft actuation on the links, a second active control is sometimes available to crawlers: the ability to change the friction coefficients in time.
The most remarkable example is inching, i.e.~the locomotion strategy of leeches and inchworms, which has been also reproduced in soft robotic devices \cite{Vikas,Gamus}. In inching locomotion the crawler can be modelled as a single link, periodically elongating and contracting, with the two extremities alternately increasing the friction coefficient  (anchoring): during elongation the backward extremity has more grip, so it remains steady while the forward extremity advances, and vice versa during contraction.
Other examples of active control of the friction coefficients can be observed  in crawlers using anisotropic friction: changing the tilt angle of bristles -- such as \emph{setae} and \emph{chaetae} in anellids \cite{AccCasDar04,Quillin} -- or scales -- such as in snakes \cite{HNSS09} -- and analogous mechanisms in robotic replicas \cite{Manw,Mar12,Raf18} produces a change in the friction coefficients \cite{GidDeS17}, that is used to facilitate sliding or gripping.

In the modelling of crawling locomotion, it is quite customary to work in the quasistatic setting, which is physically justified by their slow motion and observed stick-slip behaviour. Indeed, as discussed in \cite{WagLau}, crawling strategies that rely on inertial effects are possible, but would require an inefficient continuous sliding. We therefore propose to corroborate this choice, providing a mathematically rigorous derivation of the quasistatic limit for such models. 

To explain the occurence of system \eqref{eq:intro_dyn}, let us take a reference input $\tau\mapsto \left(\EE(\tau,\cdot),\RR(\tau,\cdot)\right)$, for $\tau\in[0,T]$, and suppose that it can be applied at an arbitrarily slow rate $\eps>0$, so that the characteristic time of the systems is proportional to $1/\eps$.
The evolution of the system is described by the following differential inclusion
\begin{equation}\label{eq:timeresc}
\mathbb{M}\ddot{x}_\eps(\tau)+\mathbb{V}\dot{x}_\eps(\tau)+\partial_{v} \mc R(\eps \tau,\dot{x}_\eps(\tau))+D_x\mc E(\eps \tau,x_\eps(\tau))\ni 0,
\end{equation}
on a time interval $\tau\in[0,T/\eps]$.
In the specific example of the locomotion model of Figure \ref{fig:crawler4}, the components $(x_\eps)_i$ of the solution will represent the position of the $i$-th block. The term $\M \ddot{x_\eps}$ describes the inertial forces, hence  $\M:=\diag\{m_1,\dots,m_N\}$ is the mass distribution. In this case the matrix $\mathbb V$ (possibly $\mathbb V=0$) could describe for instance viscous resistances to length changes in the links, or the linear component of a Bingham type friction on the blocks, caused by lubrication with a non-Newtonian fluid \cite{DeSGNT}. The term $\RR$ will in general have the form
\begin{equation*}
\RR(t,v)=\chi_K(v)+\RR_\mathrm{finite}(t,v),
\end{equation*}
where $\chi_K$ is the characteristic function of a closed convex cone $K$ and $\RR_\mathrm{finite}$ is positively homogeneous of degree one (in space) with values in $[0,+\infty)$. 
The dissipation potential $\RR_\mathrm{finite}$ accounts for dry friction forces, which may change in time. The term $\chi_K$ represents a constraint on velocities and may be used to describe situations in which hooks  or hard scales \cite{Men06} are used to create an extreme anisotropy in the interaction with the surface, so that motion ``against the hair'' may be considered impossible. The mathematical difference between a velocity constraint and a large but finite dry friction becomes extremely relevant for planar models, cf.~Subsection~\ref{subsec:plan}. 

The term $\EE$ describes the elastic energy of the system. We emphasize that, since we are dealing with a locomotion problem, rigid translations must be included in the space of admissible configurations. This implies that the elastic energy $\EE$ takes the form
\begin{equation*}
\mc E(t,x)=\mc E_\mathrm{sh}(t,\piz(x)),
\end{equation*}
where $\mc E_\mathrm{sh}$ is defined on a smaller subspace $Z\subseteq X$, on which it assumes the usual properties of coercivity/uniform convexity. The linear operator $\piz\colon X\to Z$ assigns to each configuration $x\in X$ the corresponding shape of the locomotor; in the example of Figure \ref{fig:crawler4} with $N=4$ a natural choice could be $\piz(x)=(x_2-x_1,x_3-x_2,x_4-x_3)$. We remark however that our results hold also outside locomotion, in the more classical framework with $Z=X$ and the energy $\EE$ coercive on the whole space.

We are therefore interested in the behaviour of the solutions $x_\eps\colon[0,T/\eps]\to X$ of \eqref{eq:timeresc} as $\eps\to 0^+$. In order to properly compare such solution it is necessary to rescale them in time, so that they are all defined on the same domain and to each instant $t$ corresponds the same stage of the input for all solutions.
Hence we consider the rescaled solutions $\xeps(t):=x_\eps(t/\eps )\colon[0,T]\to X$. It is easy to verify that $x_\eps$ is a solution of \eqref{eq:timeresc} if and only if $\xeps$ is a solution of \eqref{eq:intro_dyn}. Let us remark how velocity-independent forces (as elastic forces), rate-independent forces (as dry friction), and autonomous constraints are preserved by time-rescaling, whereas viscous and inertial forces are rescaled.

A more thorough interpretation of the dynamic problem \eqref{eq:intro_dyn} is the following:  we may  assume that $\M$ and $\RR$ have been normalized and their ratio has been absorbed in the parameter $\eps^2$, so that equation \eqref{eq:intro_dyn} can be seen as the result of a nondimensionalization of the system, and $\eps^2$ can be interpreted as a parameter expressing the ratio of the magnitude of inertial forces to that of dry friction forces. Measuring the weight of inertial forces compared to the other relevant forces in the system is a pivotal concept in the analysis of gaits; in terrestrial locomotion such ratio is often referred as \emph{Froude's number}, from an analogy to its namesake in fluidodynamics \cite{Ale,VauOMa05}. For instance, in legged locomotion Froude's number, together with the hip-height/stride-length ratio, plays a key role in characterizing gaits, and have led to the first estimates of the speeds achieved by dinosaurs \cite{SelMan07}.
Let us also remark that very low Froude's numbers, corresponding to quasistaticity, are not uncommon in locomotion: indeed a key challenge in the design of walking robots has been the transition from quasistatic to dynamic gaits.

For crawling locomotion, following \cite{WagLau}, Froude's number can be defined as the ratio 
\begin{equation}\label{eq:froude}
\text{Froude's number}:=\frac{\text{inertial forces}}{\text{dry friction forces}}=\frac{m_\mathrm{char}\, L_\mathrm{char}}{T_\mathrm{char}^2\, F_\mathrm{char}} 
\end{equation}
where $m_\mathrm{char},L_\mathrm{char},T_\mathrm{char},F_\mathrm{char}$ are respectively the characteristic mass (e.g.~the total mass), length (e.g. the distance covered in one iteration of the gait), time (e.g.~the period of the gait) and friction force (e.g.~the average friction force during sliding) of the locomotor.\footnote{Sometimes Froude's number for crawlers is defined as 
\begin{equation*}
\text{Froude's number}:=\frac{\text{inertial forces}}{\text{gravitational forces}}=\frac{v_\mathrm{char}^2}{gL_\mathrm{char}}\quad \left(=\mu_\mathrm{char}\, \frac{m_\mathrm{char}\, L_\mathrm{char}}{T_\mathrm{char}^2\, F_\mathrm{char}} \right),
\end{equation*}	
 which is the same expression used in legged locomotion. The validity of this second definition is based on the assumption that the normal load proportional to dry friction forces is  caused by gravity, so that  $F_\mathrm{char}= m_\mathrm{char}g\mu_\mathrm{char}$.
  The two notions are thus related by setting the characteristic speed as $v_\mathrm{char}=L_\mathrm{char}/T_\mathrm{char}$.
 We prefer the definition \eqref{eq:froude} for two reasons.
 Firstly, it provides a direct measure of the relevance of inertia in the gait, without the need to compare it with the characteristic friction coefficient $\mu_\mathrm{char}$. Secondly, not necessarily the normal load is produced by gravity: consider for instance a crawler underground or in a pipe.}  It is therefore possible to compute the relevance of inertial forces for  specific locomotor and gait.
For example, Froude's number can be estimated in the order of $10^{-3}$ both for an earthworm on the ground \cite[Chapter 6]{Ale}, and for rectilinear locomotion in boas \cite{MarBriHu}.

We finally remark that our results are not limited to soft locomotion. For instance, within the same formalism it is possible to describe finite-dimensional models made of elastic, viscous or plastic elements; these are often studied as rheological models, see e.g.~\cite{BBLbook, BroTan,Krejbook} and references therein. We briefly present some simple examples in Section~\ref{sec:models}, where, on a more theoretical perspective, we also recall how the play operator and the sweeping process are related with the quasistatic problem \eqref{eq:intro_stat} and briefly discuss the corresponding dynamic approximation.

\subsubsection*{Summary} The paper is structured as follows. In Section~\ref{sec:Setting} we present in detail our assumptions and state the main result of the paper. Section~\ref{sec:dynamicproblem} is dedicated to the dynamic problem \eqref{eq:intro_dyn}, studying existence, uniqueness and useful bounds on the solutions. The time-dependence of the dissipation functional $\RR$ requires a time-dependent generalization of BV functions, which we study in Section~\ref{secACBV}; since the arguments are the same, these results are presented in the more general framework of an arbitrary Banach space. 
The quasistatic problem \eqref{eq:intro_stat} is analysed in Section~\ref{secenergetic}, and the vanishing inertia limit is performed in Section~\ref{seclimit}. Finally, we present some applications and examples in Section~\ref{sec:models}.

\section{Setting of the problem and main result}\label{sec:Setting}

	Let $X$ be a finite dimensional vector space endowed with the norm $|\cdot|$. The same symbol will be also adopted for the modulus in $\R$; however, its meaning will be always clear from the context. We denote by $X^*$ the topological dual of $X$, and by $\langle x^*,\,x\rangle$ the duality product between $x^*\in X^*$ and $x\in X$. The operator norm in $X^*$ will be denoted by $|\cdot|_*$. Given $R>0$, by $\BB^X_R$ we denote the open ball in $X$ of radius $R$ and centered at the origin, and with $\overline{\BB^X_R}$ its closure.\par 
	Let us also recall some basic notions on set-valued maps. 	Given two topological spaces $A_1,A_2$, we denote with $F\colon A_1\rightrightarrows A_2$ a map from $A_1$ having as values subsets of $A_2$. We say that such a set-valued map is upper continuous in a point $a\in A_1$ if for every neighbourhood $U\subseteq A_2$ of $F(a)$ there exists a neighbourhood $V\subseteq A_1$ of $a$ such that $F(\tilde a)\subset U$ for every $\tilde a\in V$. We say that a map is upper semicontinuous if it is so for every point of its domain. We recall that if a set-valued map has compact values, then it is upper semicontinuous if and only if its graph is closed (cf.~e.g.~\cite{AubCel}). 
	
	Given a convex, lower semicontinuous map $\phi\colon X\to[0,+\infty]$, we define its subdifferential $\partial \phi(x_0)\subseteq X^*$ at each point $x_0\in X$ as
	\begin{equation*}
		\partial \phi(x_0)=\{\xi\in X^*\mid \phi(x_0)+\langle\xi,x-x_0\rangle \leq \phi(x) \quad\text{for every $x\in X$}\}.
	\end{equation*}
	Notice that $\partial \phi$ has closed convex values. Moreover, if $\phi(x_0)=+\infty$ and $\phi$ is finite in at least one point, then $\partial\phi(x_0)=\emptyset$.
	Given a subset $\KK\subset X$, we denote with $\chi_\KK\colon X\to [0,+\infty]$ its characteristic function:
	\begin{equation*}
	\chi_\KK(x):=\begin{cases}
	0,&\text{if $x\in \KK$},\\
	+\infty, &\text{if $x\notin \KK$}.
	\end{cases}
	\end{equation*}

	Let us now present in detail our assumptions on the mechanical problems which will be the subject of our investigation.
	\subsection*{Mass and viscosity}
	 Let $\mathbb{M}\colon X\to X^*$ be a symmetric positive-definite linear operator, which will represent mass distribution. Since $X$ has finite dimension, we observe that there exist two constants $M\geq m>0$ such that
	\begin{equation}\label{boundsassumption1}
m|x|^2\leq |x|_{\mathbb{M}}^2:=\langle\mathbb{M}x,x\rangle\leq M|x|^2,\qquad\text{ for every $x\in X$}.
	\end{equation}
We want to stress that the requirement on $\mathbb{M}$ of being positive definite, crucial for our analysis, fits well with the finite dimensional setting in which we are working; in particular, all the applications we have in mind, see Section~\ref{sec:models}, fulfil this assumption. On the contrary, in infinite dimensional models usually the mass operator is null on a subspace (see for instance \cite{MielkPetr}), thus in that case $\mathbb{M}$ turns out to be only positive-semidefinite.\par 
 We consider also the (possible) presence of viscous dissipation, by introducing the positive-semidefinite linear operator $\mathbb{V}\colon X\to X^*$ (symmetry is not needed here). As before, we notice that there
 exists a nonnegative constant $V\geq 0$ such that
 \begin{equation}\label{boundsassumption2}
 0\le|x|_\mathbb V^2:=\langle\mathbb{V}x,x\rangle\le V|x|^2,\qquad\text{ for every $x\in X$}.
 \end{equation}
 We point out that we include also the case $\mathbb{V}\equiv 0$, corresponding to the absence of viscous friction forces in the dynamic problem \eqref{dynprob}. Indeed, in this paper we are mostly interested in the presence of a different type of dissipation, which will be introduced in the following, and which actually overwhelms the effects of viscosity for the purposes of the vanishing inertia analysis.
 
\subsection*{The elastic energy}
Before introducing our assumptions on the elastic energy $\EE$, we recall that our main application concerns a locomotion problem. This implies that the space of admissible states $X$ must include translations, for which the elastic energy is invariant. Hence the elastic energy will be coercive only on a subspace, intuitively corresponding to the shape of the locomotor.

Let us therefore consider a linear subspace $Z\subseteq X$, which is often convenient to endow with its own norm $|\cdot|_Z$, cf.~the examples in  \cite{Gid18}. We assume that the elastic energy $\mc E\colon [0,T]\times X\to [0,+\infty)$ has the form $\mc E(t,x)=\mc E_\mathrm{sh}(t,\piz(x))$, where $\piz\colon X\to Z$ is a linear and surjective operator and $\mc E_\mathrm{sh}\colon [0,T]\times Z\to [0,+\infty)$ satisfies:
	\begin{enumerate}[label=\textup{(E\arabic*)}]
		\item \label{hyp:E1} $\mc E_\mathrm{sh} (\cdot,z)$ is absolutely continuous in $[0,T]$ for every $z\in Z$; 
		\item \label{hyp:E2} $\mc E_\mathrm{sh}(t,\cdot)$ is $\mu$-uniformly convex for some $\mu>0$ for every $t\in[0,T]$, namely for every $\theta\in[0,1]$, $z_1,z_2\in Z$:
		\begin{equation*}
			\mc E_\mathrm{sh}(t,\theta z_1+(1-\theta)z_2)\le\theta\mc E_\mathrm{sh}(t,z_1)+(1-\theta)\mc E_\mathrm{sh}(t,z_2)-\frac{\mu}{2}\theta(1-\theta)|z_1-z_2|^2_Z;
		\end{equation*}
		\item \label{hyp:E3} $\mc E_\mathrm{sh}(t,\cdot)$ is differentiable for every $t\in[0,T]$ and the differential $D_z \mc E_\mathrm{sh}$ is continuous in $[0,T]\times Z$;
		\item \label{hyp:E4} for a.e. $t\in[0,T]$ and for every $z\in Z$ it holds
		\begin{equation*}
			\left|\frac{\partial}{\partial t} \mc E_\mathrm{sh}(t,z)\right|\le\omega(\mc E_\mathrm{sh}(t,z))\gamma(t),
		\end{equation*}
		where $\omega\colon [0,+\infty)\to[0,+\infty)$ is nondecreasing and continuous, while $\gamma\in L^1(0,T)$ is nonnegative;
		\item \label{hyp:E5} for every $R>0$ there exists a nonnegative function $\eta_R\in L^1(0,T)$ such that for a.e. $t\in[0,T]$ and for every $z_1,z_2\in \overline{\BB_R^Z}$ it holds 
		\begin{equation*}
			\left|\frac{\partial}{\partial t}\mc E_\mathrm{sh}(t,z_2)-\frac{\partial}{\partial t}\mc E_\mathrm{sh}(t,z_1)\right|\le \eta_R(t)|z_2-z_1|_Z.
		\end{equation*}
	\end{enumerate}
Let us also introduce some additional assumptions on the energy $\EE$, which are in general not required, but provide sharper results.
\begin{enumerate}[label=\textup{(E\arabic*)},resume*]
	\item \label{hyp:E6} for every $\lambda>0$ and $R>0$ there exists $\delta=\delta(\lambda,R)>0$ such that if $|t-s|\le\delta$ and $z\in \BB^Z_R$, then $$\left|\frac{\partial}{\partial t} \mc E_\mathrm{sh}(t,z)-\frac{\partial}{\partial t} \mc E_\mathrm{sh}(s,z)\right|\le\lambda;$$
	\item \label{hyp:E7} for every $R>0$ there exists a nonnegative function $\varsigma_R\in L^1(0,T)$ such that for a.e. $t\in[0,T]$ and for every $z_1,z_2\in \overline{\BB_R^Z}$ it holds 
	\begin{equation*}	|D_z\mc E_\mathrm{sh}(t,z_2)-D_z\mc E_\mathrm{sh}(t,z_1)|_*\le \varsigma_R(t)|z_2-z_1|_Z.
			\end{equation*}
\end{enumerate}
We finally present the classical case of a quadratic energy:
	\begin{enumerate}[label=\textup{(QE)}]
\item \label{hyp:QE}  $\mc E_\mathrm{sh}(t,z)=\frac 12\langle\mathbb{A}_\mathrm{sh}(z-\ell_\mathrm{sh}(t)),z-\ell_\mathrm{sh}(t)\rangle_Z$, where $\mathbb{A}_\mathrm{sh}\colon Z\to Z^*$ is a symmetric, positive-definite linear operator  and $\ell_\mathrm{sh}\in AC([0,T];Z)$.
	\end{enumerate}
It can be easily verified that \ref{hyp:QE} implies conditions \ref{hyp:E1}--\ref{hyp:E5} and \ref{hyp:E7}, whereas it satisfies \ref{hyp:E6} if and only if $\ell_\mathrm{sh}$ has continuous derivative. However, for our purposes, the additional structure of \ref{hyp:QE} will alone provide a suitable alternative to \ref{hyp:E6}.

\begin{rmk}
We point out that the more common case $Z\equiv X$ is also included in our formulation. In such a case all the assumptions above on $\EE_\mathrm{sh}$ are taken directly on $\EE$.
\end{rmk}

\begin{rmk} 
Let us notice that, since $\piz$ is linear, if any of \ref{hyp:E1}, \ref{hyp:E3}--\ref{hyp:E7} holds, the same property enunciated for $\EE_\mathrm{sh}$ is satisfied also \lq\lq directly\rq\rq\ by the entire function $\EE$ on $[0,T]\times X$, with the only change of the addition of the multiplicative term $|\piz|_*$ in the bounds of \ref{hyp:E5}, \ref{hyp:E7}. 
The only caveat is with \ref{hyp:E2}, which implies that  $\mc E(t,\cdot)$ is convex, but in general not uniformly convex in the whole $X$. We however point out that convexity will not be necessary when dealing with the dynamic problem \eqref{eq:intro_dyn} and for the first part of the subsequent vanishing inertia analysis performed in Section~\ref{seclimit}, where also non convex energies are allowed.
\end{rmk}

Thanks to the above remark, we observe that by \ref{hyp:E1} and \ref{hyp:E3} we deduce that $\mc E$ is continuous in $[0,T]\times X$, while from \ref{hyp:E1} and \ref{hyp:E5} we get that $\frac{\partial}{\partial t} \mc E$ is a Caratheodory function. Thus for every $x\colon[0,T]\to X$ measurable, the function $t\mapsto \frac{\partial}{\partial t} \mc E(t,x(t))$ is measurable too. Moreover if $x$ is also bounded, namely $\sup\limits_{t\in[0,T]}|x(t)|\le R$, then \ref{hyp:E4} implies that $\frac{\partial}{\partial t}\mc E(\cdot,x(\cdot))$ is summable in $[0,T]$, indeed:
		\begin{equation*}
			\int_{0}^{T}\left|\frac{\partial}{\partial t}\mc E(\tau,x(\tau))\right|\d\tau\le\int_{0}^{T}\omega(\mc E(\tau,x(\tau)))\gamma(\tau)\d\tau\le\omega(M_R)\int_{0}^{T}\gamma(\tau)\d\tau<+\infty,
		\end{equation*}
		where $M_R$ denotes the maximum of $\mc E$ on the compact set $[0,T]\times \overline{\BB^X_R}$. If in addition $x$ is absolutely continuous from $[0,T]$ to $X$, by \ref{hyp:E1}, \ref{hyp:E3} and \ref{hyp:E4} we also deduce that $t\mapsto \mc E(t,x(t))$ is absolutely continuous in $[0,T]$ too, indeed for every $0\le s\le t\le T$ it holds:
		\begin{align*}
			|\mc E(t,x(t))-\mc E(s,x(s))|&\le |\mc E(t,x(t))-\mc E(t,x(s))|+|\mc E(t,x(s))-\mc E(s,x(s))|\\
			&\le C_R|x(t)-x(s)|+\int_{s}^{t}\left|\frac{\partial}{\partial t}\mc E(\tau,x(s))\right|\d\tau\\
			&\le C_R|x(t)-x(s)|+\omega(M_R)\int_{s}^{t} \gamma(\tau)\d\tau,
		\end{align*}
    where $C_R$ is the maximum of $|D_x\mc E|_*$ on $[0,T]\times \overline{\BB^X_R}$.

\subsection*{The dissipation potential} We introduce the main dissipative forces involved in the system, described by a time-dependent dissipation potential $\mc R\colon [0,T]\times X\to[0,+\infty]$ which takes into account both possible constraints on the velocity and the presence of dry friction. It originates from a function $\RR_\mathrm{finite}\colon [0,T]\times X\to[0,+\infty)$ with finite values on which we make the following assumptions:
\begin{enumerate}[label=\textup{(R\arabic*)}, start=1]
		\item \label{hyp:R1}  for every $t\in[0,T]$, the function $\mc R_\mathrm{finite}(t,\cdot)$ is convex, positively homogeneous of degree one, and satisfies $\RR_\mathrm{finite}(t,0)=0$;
		\item \label{hyp:R2} there exist two positive constants $\alpha^*\ge\alpha_*>0$ for which
		\begin{equation*}
		\alpha_*\abs{v}\le \mc R_\mathrm{finite}(t,v)\le \alpha^*\abs{v},\quad\text{ for every }(t,v)\in[0,T]\times X;
		\end{equation*}
		\item \label{hyp:R3}  there exists a nonnegative function $\rho\in L^1(0,T)$ for which
		\begin{equation*}
		|\mc R_\mathrm{finite}(t,v)-\mc R_\mathrm{finite}(s,v)|\le\abs{v}\int_{s}^{t}\rho(\tau)\d\tau,\text{ for every }0\le s\le t\le T \text{ and for every }v\in X.
		\end{equation*}
	\end{enumerate}
	\begin{rmk} 
We observe that the second inequality in \ref{hyp:R2} actually follows from \ref{hyp:R1} and \ref{hyp:R3}. Indeed, since we are in finite dimension, the convex function $\RR_\mathrm{finite}(t,\cdot)$ is automatically continuous on $X$; by \ref{hyp:R3} this easily implies $\RR_\mathrm{finite}$ is continuous on the whole $[0,T]\times X$, and hence by one-homogeneity we get $\RR_\mathrm{finite}(t,v)\leq C\abs{v}$ for some constant $C>0$ and every $(t,v)\in[0,T]\times X$.
\end{rmk}
	As regards $\RR$ we finally assume that:
	\begin{enumerate}[label=\textup{(R\arabic*)}, start=4]
	\item \label{hyp:RK} there exists a nonempty closed convex cone $K\subseteq X$, independent of time, and there exists a function $\RR_\mathrm{finite}\colon [0,T]\times X\to[0,+\infty)$ satisfying \ref{hyp:R1}--\ref{hyp:R3} such that for every $(t,v)\in[0,T]\times X$ it holds
		\begin{equation*}
		\RR(t,v)=\chi_K(v)+\RR_\mathrm{finite}(t,v).
		\end{equation*}
	\end{enumerate}
We will denote with $\partial_v \RR$ the subdifferential of $\RR$ with respect to its second variable. The choice of the letter $v$ when dealing with the dissipation potential reminds the fact that the second argument of $\RR$ is usually a velocity.

As an immediate consequence of condition \ref{hyp:RK} we can rephrase conditions \ref{hyp:R1}--\ref{hyp:R3} directly on $\mc R$:

\begin{cor}\label{propertiesR}
    Let $\RR$ be as in \ref{hyp:RK}. Then it holds:
    \begin{enumerate}[label=\textup{(\Roman*)}]
        \item \label{prop:regR1}  for every $t\in[0,T]$, the function $\mc R(t,\cdot)$ is convex, positively homogeneous of degree one, lower semicontinuous, and satisfies $\RR(t,0)=0$;
		\item \label{prop:regR2} for every $(t,v)\in[0,T]\times K$ one has
		\begin{equation*}
		\alpha_*\abs{v}\le \mc R(t,v)\le \alpha^*\abs{v},
		\end{equation*}
		with the same constants $\alpha^*$ and $\alpha_*$ of \ref{hyp:R2};
		\item \label{prop:regR3}  for every $0\le s\le t\le T$ and for every $v\in K$ one has
		\begin{equation*}
		|\mc R(t,v)-\mc R(s,v)|\le\abs{v}\int_{s}^{t}\rho(\tau)\d\tau,
		\end{equation*}
		with the same function $\rho$ of \ref{hyp:R3}.
    \end{enumerate}
    Moreover the following properties hold true:
    \begin{enumerate}[label=\textup{(\Roman*)},resume]
        \item \label{prop:R_IV} for every $(t,v)\in [0,T]\times X$ one has
        \begin{equation*}
            \partial_v\RR_\mathrm{finite}(t,v)\subseteq \overline{\BB^{X^*}_{\alpha^*}},
        \end{equation*}
        with $\alpha^*$ as in \ref{hyp:R2}. In particular $\partial_v\RR_\mathrm{finite}$ has compact, convex, non-empty values.
        \item \label{prop:R_V} the multivalued map  $\partial_v \RR_\mathrm{finite}$ is upper semicontinuous on $[0,T]\times X$;
    \end{enumerate}
\end{cor}
\begin{proof}
The first three points are a trivial consequence of \ref{hyp:R1}--\ref{hyp:R3}, respectively, due to the form of $\mc R$ given by \ref{hyp:RK}. We indeed notice that, since $K$ is a nonempty closed convex cone, its characteristic function $\chi_K$ is convex, positively homogeneous of degree one, lower semicontinuous, and vanishes at $v=0$. 
To prove \ref{prop:R_IV}, since $\RR_\mathrm{finite}$ has finite values, we deduce that $\xi\in \partial_v \RR_\mathrm{finite}(t,v)$ if and only if 
	\begin{equation*}
	    \scal{\xi}{\tilde v}\le \RR_\mathrm{finite}(t,\tilde v+v)-\RR_\mathrm{finite}(t,v),\quad\text{ for every }\tilde v\in X.
	\end{equation*}
	We now recall that convexity plus one-homogeneity easily yield subadditivity, thus we can continue the above inequality getting
	\begin{equation*}
	    \scal{\xi}{\tilde v}\le \RR_\mathrm{finite}(t,\tilde v),\quad\text{ for every }\tilde v\in X.
	\end{equation*}
	By means of \ref{hyp:R2} we thus deduce that $|\xi|_*\le \alpha^*$, and so \ref{prop:R_IV} is proved.

	To prove \ref{prop:R_V}, since $\partial_v \RR_\mathrm{finite}$ has compact values, it is sufficient to show that for every sequence $(t_k,v_k,\xi_k)$ in $[0,T]\times X\times X^*$ such that $\xi_k\in \partial_v\RR_\mathrm{finite}(t_k,v_k)$, if  $(t_k,v_k,\xi_k)\to(\bar t,\bar v,\bar \xi)\in [0,T]\times X\times X^*$ then $\bar \xi\in \partial_v\RR_\mathrm{finite}(\bar t,\bar v)$. 
	By definition of subdifferential, for every $k\in \N$ we have 
	\begin{equation*}
	\RR_\mathrm{finite}(t_k,v_k)+\langle\xi_k,v-v_k\rangle \leq \RR_\mathrm{finite}(t_k,v), \quad\text{for every $v\in X$}.
	\end{equation*}
	By the continuity of $\RR_\mathrm{finite}$ on $[0,T]\times X$ and of the dual coupling, passing to the limit in the above estimate gives
	\begin{equation*}
	\RR_\mathrm{finite}(\bar t, \bar v)+\langle\bar \xi,v-\bar v\rangle \leq \RR_\mathrm{finite}(\bar t,v), \quad\text{for every $v\in X$},
	\end{equation*}
	namely $\bar \xi\in \partial_v\RR_\mathrm{finite}(\bar t,\bar v)$, concluding the proof.
\end{proof}
\begin{rmk}[\textbf{Comparison with $\psi$-regularity \cite{HeiMielk}}] Let us remark that our assumptions on $\RR$ are very close to the notion of $\psi$-regularity introduced in \cite{HeiMielk} (see also Definition~\ref{psiregular}). Most of the differences between the two frameworks are due to the fact that \cite{HeiMielk} deals with functionals $\RR$ defined on a general Banach space $X$, but with finite values. For instance, if the functional $\RR$ has finite values, we observe that assumption \ref{hyp:RK} is automatically satisfied with $K=X$.
	
The are only two points in which our assumptions are actually slightly stricter than  \cite{HeiMielk}, and both are motivated. The first one is the left inequality in \ref{hyp:R2}, corresponding in the framework of \cite{HeiMielk} to the additional assumption $c\abs{v}\leq \psi(v)$. This is related to the fact that we have renounced to coercivity in the energy $\EE$, and such loss has to be compensated with a coercivity in the dissipation potential $\RR$, in order to recover some a priori estimates, such as \ref{prop:stima_i} in Corollary~\ref{unifbound}. We however point out that such a request is absolutely natural in the finite dimensional setting we are considering, as we will see in the examples of Section~\ref{sec:models}. On the contrary, it becomes very restrictive in infinite dimension: indeed, in standard models of elasticity where the simplest ambient space is $H^1_0(\Omega)$, a common choice of dissipation potential is $\int_\Omega |v(x)|\d x$, which of course lacks of coercivity.\par
The second stronger assumption is that the modulus of continuity appearing in \ref{hyp:R3} is of integral type. This is because we are interested in absolutely continuous solutions of the quasistatic problem \eqref{quasprob}, not just continuous ones, cf.~Proposition~\ref{regularity}. 
However, a general modulus of continuity (as the one used in \cite{HeiMielk}) would be enough to get all the results presented in Section~\ref{secACBV}.
\end{rmk}

Let us also introduce an optional assumption on $\RR$ (actually on the set $K$), which will be used to improve the regularity of the quasistatic solutions:

\begin{enumerate}[label=\textup{(R\arabic*)}, start=5]
	\item \label{hyp:R5} there exists a constant $C_K>0$ such that, for every $z\in Z$
	\begin{itemize}
	\item either $\piz(x)\neq z$ for every $x\in K;$
	\item or there exists $x\in K$ such that $\piz(x)=z$ and $\abs{x}\le C_K\abs{z}_Z$.
	\end{itemize}
\end{enumerate}
We remark that, by a physical point of view, assumption \ref{hyp:R5} is usually satisfied. Indeed, violating \ref{hyp:R5} would mean that the constraints allow a locomotor to achieve an arbitrarily large displacement with an arbitrarily small change in shape. All the concrete models we consider in Section~\ref{sec:models} satisfy \ref{hyp:R5}; we discuss a purely theoretical counterexample in Subsection \ref{subsec:notR5}. By a mathematical point of view, let us highlight some common situations where  \ref{hyp:R5} is true.

\begin{prop} \label{prop:R5char} Each of the following is a sufficient condition for \ref{hyp:R5}:
\begin{enumerate}[label=\textup{(\roman*)}]
	\item $K=X$ or $K=\{0\}$;
	\item $\dim Z= \dim X$;
	\item \label{cond:R5char3} $\dim X=1+\dim Z$ and $K$ is a polyhedral closed cone, i.e.~there exist $J$ covectors $f_1^K,\dots f_J^K\in X^*$ such that
	\begin{equation*}
		K=\{x\in X \mid \scal{f_j^K}{x}\geq 0 \quad \text{for every $j=1,\dots,J$}\};
	\end{equation*}
\end{enumerate}	

\end{prop}
\begin{proof}
The first two points are trivial. Let us therefore prove the third point. First of all we observe that for $z=0_Z$ the second alternative of \ref{hyp:R5} is satisfied by $x=0_X$. For $z\neq 0_Z$, by homogeneity, it is sufficient to consider the case $\abs{z}_Z=1$. Moreover, without loss of generality we can assume $\bigl\lvert f_j^K \bigr\rvert_*=1$.

  Let $i\colon Z\times \ker \piz\to X$ be the canonical identification. For every $z\in Z$ we write $\hat z:=i(z,0)\in X$; moreover, fixed any nonzero vector $y\in\ker\piz$, we set $\eta:=i(0,y)/\abs{i(0,y)}$. Since $\dim X=1+\dim Z$, we deduce that $\piz(x)=z$ if and only if $x=\hat z+\lambda\eta$ for some $\lambda\in \R$.
  
  Let us write $\mathcal S=\{\hat z=i(z,0)\in X\mid\abs{z}_Z=1\}$ and set
\begin{align*} 
&C_1:=\max_{j=1,\dots,J} \max_{\hat{z}\in\mathcal S} \abs{\scal{f^K_j}{\hat z}},\\
&C_2:=\min_{j=1,\dots,J}  \{\abs{\scal{f^K_j}{\eta}}\mid\scal{f^K_j}{\eta}\neq 0\},\\
&C_3:=\max_{\hat{z}\in\mathcal S} \abs{\hat z}.
\end{align*}
 We claim that we can take $C_K={C_3+(C_1/C_2)}$.
Fix $z$ with norm $1$, and consider the corresponding $\hat z\in \mathcal S$. Since $K$ is closed, we have two alternative possibilities:
\begin{itemize} 
	\item either $\piz(x)\neq z$ for every $x\in K$;
	\item or there exists $\bar \lambda\in R$ such that $\hat z+\bar \lambda \eta\in K$ and
	\begin{equation*}
		\abs{\hat z+\bar \lambda \eta}\leq \abs{\hat z+\lambda \eta}, \qquad \text{for every $\lambda\in \R $ such that $\hat z+\lambda \eta\in K$}.
	\end{equation*}
\end{itemize} 
 To prove \ref{hyp:R5} it is sufficient to show that, if the second option holds,  $\abs{\bar\lambda}\leq C_1/C_2$, so that $\abs{\hat z+\bar\lambda \eta}\le\abs{\hat{z}}+\abs{\bar \lambda}\leq C_3+(C_1/C_2)$. 
 To show this estimate on $\abs{\bar\lambda}$, let us observe that, in order to minimize the absolute value, either $\bar \lambda=0$ or there exists an index $j$ such that 
\begin{equation*}
	\scal{f_j^K}{\hat z}+\bar \lambda\scal{f_j^K}{\eta}=0, \qquad\text{and}\qquad \scal{f_j^K}{\eta}\neq 0,
\end{equation*}
which implies $\abs{\lambda}\leq C_1/C_2$.
\end{proof}
\medskip

We now present the dynamic and quasistatic problems we will study and state our main result.

\subsection*{The dynamic problem} Let $\M,\V$ be as above, and assume that \ref{hyp:E1}, \ref{hyp:E3}--\ref{hyp:E5} and \ref{hyp:RK} are satisfied. For $\eps>0$ we  refer as \emph{dynamic problem} to the differential inclusion
	\begin{equation}\label{dynprob}
\begin{cases}
\eps^2 \mathbb{M}\xepsdd(t)+\eps \mathbb{V}\xepsd(t)+\partial_v \mc R(t,\xepsd(t))+D_x\mc E(t,\xeps(t))\ni 0,\\
\xeps(0)=x^\eps_0,\quad\xepsd(0)=x^\eps_1,
\end{cases}
\end{equation}
where the initial velocity satisfy the admissibility condition
\begin{equation} \label{eq:dyn_admiss}
	x^\eps_1\in K,
\end{equation}
for $K$ as in \ref{hyp:RK}. 
\begin{defi}
    We say that a function $\xeps\in W^{2,1}(0,T; X)$ is a \emph{differential solution} of \eqref{dynprob} if the differential inclusion holds true in $X^*$ for a.e. $t\in[0,T]$ and initial position and velocity are attained.
\end{defi} 
We discuss existence and uniqueness of a differential solution for \eqref{dynprob} in Section~\ref{sec:dynamicproblem}, see Theorem~\ref{existencedyn}.

\subsection*{The quasistatic problem} Assume that \ref{hyp:E1}--\ref{hyp:E5} and \ref{hyp:RK} are satisfied. We refer as \emph{quasistatic problem} to the differential inclusion
	\begin{equation}\label{quasprob}
	\begin{cases}
	\partial_v \mc R(t,\dot{x}(t))+D_x\mc E(t,x(t))\ni 0,\\
	x(0)=x_0.
	\end{cases}
	\end{equation}
For the quasistatic problem we introduce two notions of solution. Conditions for existence of each type of solution are a direct consequence of the main result, although they could be derived separately (see for instance \cite{MielkRoubbook} for a general argument based on time-discretization).
 	\begin{defi}
We say that a function $x\in AC([0,T];X)$ is a \emph{differential solution} of \eqref{quasprob} if the differential inclusion holds true in $X^*$ for a.e. $t\in[0,T]$ and the initial position is attained. 
	\end{defi}
We observe that the existence of differential solutions for \eqref{quasprob}  requires the admissibility condition on the initial datum
\begin{equation}\label{eq:adm_quasistat}
-D_x\mc E(0,x_0)\in \partial_v \mc R(0,0).
\end{equation}

In order to introduce the second (weaker) notion of solution, let us first state a suitable generalization of functions of bounded variation, which we will discuss in detail in Section~\ref{secACBV}.
\begin{defi} \label{def:RBV}
	Given a function $f\colon [a,b]\to X$, we define its $\mc R$-variation in $[s,t]$, with $a\le s<t\le b$, as:
	\begin{equation}\label{rvar}
	V_{\mc R}(f;s,t):=\lim\limits_{n\to +\infty}\sum_{k=1}^{n}\mc R(t_{k-1}, f(t_k)-f(t_{k-1})),
	\end{equation}
	where $\{t_k\}_{k=1}^{n}$ is a \emph{fine sequence of partitions} of $[s,t]$, namely it is of the form $s=t_0<t_1<\dots<t_n=t$ and satisfies
	\begin{equation}\label{finezza}
	\lim\limits_{n\to +\infty}\sup\limits_{k=1,\dots, n}(t_k-t_{k-1})= 0.
	\end{equation}
	We also set $V_{\mc R}(f;t,t):=0$, for every $t\in[a,b]$.\par 
	We say that $f$ is a function of bounded $\mc R$-variation in $[a,b]$, and we write $f\in BV_{\mc R}([a,b];X)$, if its $\mc R$-variation in $[a,b]$ is finite, i.e. $V_{\mc R}(f;a,b)<+\infty$.
\end{defi}
\begin{defi}\label{defenergetic}
	We say that $x\in BV_{\mc R}([0,T];X)$ is an \emph{energetic solution} for the quasistatic problem \eqref{quasprob} if the initial position is attained and the following global stability condition and weak energy balance hold true:
	\begin{enumerate}[label=\textup{(GS)}]
		\item \label{GS} $\mc E(t,x(t))\le \mc E(t,v)+\mc R(t,v-x(t)),\quad \text{ for every }v\in X \text{ and for every }t\in[0,T];$
	\end{enumerate}
	\begin{enumerate}[label=\textup{(WEB)}]
		\item \label{WEB} $\displaystyle \mc E(t,x(t))+V_{\mc R}(x;0,t)=\mc E(0,x_0)+\int_{0}^{t}\frac{\partial}{\partial t} \mc E(\tau,x(\tau))\d\tau, \quad \text{for every }t\in[0,T].$
	\end{enumerate}
\end{defi}
The justification of this definition, together with the main properties of energetic solutions, will be given in Section~\ref{secenergetic}; see in particular Proposition~\ref{propdiffenerg}.
We remark that the notion of energetic solution is more flexible than the one of differential solution, since it does not involve derivatives and in general allows for discontinuous solutions. We refer to \cite{MielkRoubbook} for a wide and complete presentation on the topic.

\subsection*{Main result} We are now ready to state the main result of this paper, concerning the asymptotic behaviour as $\eps\to 0^+$ of differential solutions of the dynamic problem \eqref{dynprob}.

\begin{thm}
	Let $\M,\V$ be as above; assume that $\RR$ satisfies \ref{hyp:RK}, and that $\mc E(t,x)=\mc E_\mathrm{sh}(t,\piz (x))$ satisfies \ref{hyp:E1}--\ref{hyp:E6} or \ref{hyp:QE}.
	Let $\xeps$ be a differential solution of the dynamic problem \eqref{dynprob} related to the initial position $x^\eps_0\in X$ and the initial velocity $ x^\eps_1\in K$, and assume
	\begin{align}\label{eq:conv_initial}
		\lim_{\eps\to 0} x^\eps_0=x_0, && \lim_{\eps\to 0}\eps x^\eps_1=0,
	\end{align}
	 for some $x_0$ satisfying \eqref{eq:adm_quasistat}. Then there exist a subsequence $\epsj\searrow 0$ and a function $x\in BV_{\mc R}([0,T];X)\cap \CC^0([0,T];X)$ such that $x$ is an energetic solution for \eqref{quasprob} with initial position $x_0$ and:
	\begin{enumerate}[label=\textup{(\alph*)}]
		\item $\lim\limits_{j\to +\infty}\xepsj(t)=x(t)$ uniformly on $[0,T]$;
		\item$\displaystyle\lim\limits_{j\to +\infty}\int_{s}^{t}\mc R(\tau,\xepsjd(\tau))\d\tau=V_\mc R(x;s,t)$ for every $0\leq s\le t\le T$;
		\item $\lim\limits_{j\to +\infty}\epsj |\xepsjd(t)|_{\mathbb{M}}=0$ uniformly on $[0,T]$;
		\item $\displaystyle\lim\limits_{j\to +\infty}\epsj\int_{0}^{T}|\xepsjd(\tau)|^2_\mathbb{V}\d\tau=0$.
	\end{enumerate}
In particular, in case of uniqueness of energetic solutions to the quasistatic problem \eqref{quasprob}, cf. for instance Lemmata~\ref{lemma:uniqquas} and \ref{lemma:uniqquas2}, the result holds true for the whole sequence $x^\eps$.\par 
If, in addition, \ref{hyp:R5} holds or $\mc R$ does not depend on time, then the limit function $x$ is absolutely continuous and, in particular, it is a differential solution of \eqref{quasprob}.
\end{thm}

We remark that assumption \eqref{eq:conv_initial} can be relaxed to the boundedness of the sequences. In such a case, as we argue in Theorem~\ref{finalthm}, we obtain similar results with energetic solutions having a (possible) jump in $t=0$.



\section{Existence of solutions for the dynamic problem}\label{sec:dynamicproblem}

This section is devoted to the analysis of the dynamic problem \eqref{dynprob} and to the proof of an existence result under the main assumptions \ref{hyp:E1}, \ref{hyp:E3}--\ref{hyp:E5} and \ref{hyp:RK}. Convexity, i.e. \ref{hyp:E2}, here is not needed. Condition \ref{hyp:E7} will be also added to obtain uniqueness of differential solutions, see Theorem~\ref{existencedyn}. Of course in this section the parameter $\eps>0$ is fixed; however, since some results we obtain here will be useful also in the rest of the paper where $\eps$ is sent to $0$, for the sake of brevity we prefer to assume that the initial data are uniformly bounded in $\eps$. Namely we require there exists a positive constant $\Lambda>0$ for which
\begin{equation}\label{unifdata}
    \abs{x^\eps_0}\le\Lambda,\quad\text{ and }\quad\abs{\eps x^\eps_1}\le\Lambda,\qquad\text{for every }\eps>0.
\end{equation}
Before starting the analysis we recall the following Gr\"onwall-type estimate:
\begin{lemma}[\textbf{Gr\"onwall inequality}]\label{Gronwall}
	Let $f\colon[a,b]\to [0,+\infty)$ be a bounded measurable function such that
	\begin{equation}\label{esthyp}
	f(t)\le C+\int_{a}^{t}\omega(f(\tau))g(\tau)\d\tau,\quad\text{ for every }t\in[a,b],
	\end{equation}
	where $C>0$ is a positive constant, $\omega\colon[0,+\infty)\to[0,+\infty)$ is a nondecreasing continuous function such that $\omega(x)>0$ if $x>0$, and $g\in L^1(a,b)$ is nonnegative.\par 
	Then it holds:
	\begin{equation*}
	f(t)\le\varphi^{-1}\left(\varphi(C)+\int_{a}^{t}g(\tau)\d\tau\right),\quad\text{ for every }t\in[a,b],
	\end{equation*}
	where $\displaystyle\varphi(t):=\int_{1}^{t}\frac{1}{\omega(\tau)}\d\tau$.
\end{lemma}
\begin{proof}
	We consider the auxiliary function $F(t):=\int_{a}^{t}\omega(f(\tau))g(\tau)\d\tau$. Since $f$ is bounded, $F$ is absolutely continuous in $[a,b]$ and $F(a)=0$. Moreover by \eqref{esthyp} we deduce:
	\begin{equation*}
	\dot{F}(\tau)=\omega(f(\tau))g(\tau)\le \omega(C+F(\tau))g(\tau),\quad\text{ for a.e. }\tau\in[a,b].
	\end{equation*}
	From the above inequality we thus infer for every $t\in[a,b]$:
	\begin{align*}
	\int_{a}^{t}g(\tau)\d\tau&\ge\int_{a}^{t}\frac{\dot{F}(\tau)}{\omega(C+F(\tau))}\d\tau=\int_{C}^{C+F(t)}\frac{1}{\omega(\tau)}\d\tau=\varphi(C+F(t))-\varphi(C)\\
	&\ge \varphi(f(t))-\varphi(C),
	\end{align*}
	where in the last inequality we used again \eqref{esthyp} and exploited the monotonicity of $\varphi$. Hence we conclude.
\end{proof}

For a reason which will be clear later, to develop all the arguments of this section we need to introduce a truncated version of the elastic energy $\EE$. We argue as follows: for every $\rho\in(0,+\infty)$, let $\lambda^\rho\colon[0,+\infty)\to[0,\rho+1]$ be a $\CC^\infty$, monotone increasing, concave function such that $\lambda^\rho(r)=r$ for $r\leq \rho$ and let us consider the truncated energies
\begin{equation}
\EE^\rho(t,x)=\EE\left(t,\sigma_\rho(x)\right),
 \quad \text{where $\sigma_\rho(x):=\frac{\lambda^\rho(\abs x) x}{\abs{x}}$},
\end{equation}
setting in the limit case $\EE^{+\infty}\equiv\EE$. Notice that $\sigma_\rho$ is the identity on $\overline{\BB^X_\rho}$ and that the Jacobian of $\sigma_\rho$ at each point has (operator) norm less or equal to one.

We observe that the new functions $\EE^\rho$ cannot be expressed any longer as function of $(t,\piz(x))$. Yet they inherit many of the regularity properties of $\EE$ and $\EE_\mathrm{sh}$. Indeed we observe that, by \ref{hyp:E1} and \ref{hyp:E3}, the functions $\EE^\rho$ and $D_x \EE^\rho$ are continuous in $[0,T]\times X$, while from \ref{hyp:E1} and \ref{hyp:E5} we get that $\frac{\partial}{\partial t} \EE^\rho$ is a Caratheodory function.  Moreover, by \ref{hyp:E4} it holds
\begin{equation}\label{eq:Erho4}
\left|\frac{\partial}{\partial t} \EE^\rho(t,x)\right|\le\omega(\EE^\rho(t,x))\gamma(t), \qquad \text{for a.e.~$t\in[0,T]$ and for every $x\in X$},
\end{equation}
where $\omega$ and $\gamma$ are the same of \ref{hyp:E4} and in particular do not depend on $\rho$.
Furthermore, by compactness and the properties of $\sigma_\rho$, if $\rho\in(0,+\infty)$ then we get that $D_x\EE^\rho$ is bounded on the whole $[0,T]\times X$, namely there exists a constant $C_\rho>0$ such that
\begin{equation}\label{DxErho}
   \sup\limits_{(t,x)\in [0,T]\times X} |D_x\mc E^\rho(t,x)|_*\le C_\rho.
\end{equation}
The above estimate is the main reason why we introduced the truncated energy; it will be indeed crucial in the proof of Proposition~\ref{prop:trunc_exist}.\par 
If in addition also \ref{hyp:E7} holds, we deduce that there exists a function $\widetilde\varsigma_\rho\in L^1(0,T)$ such that
\begin{equation}\label{eq:Erho7}
|D_x\EE^\rho(t,x_1)-D_x\EE^\rho(t,x_2)|_*\le\widetilde\varsigma_\rho(t)\abs{x_1-x_2},
\end{equation}
for a.e. $t\in [0,T]$, and every $x_1,x_2\in \overline{\BB^X_\rho}$.

\medskip

Let us thus consider the approximated problems
	\begin{equation}\label{eq:trun_dynprob}
\begin{cases}
\eps^2 \mathbb{M}\xepsdd(t)+\eps \mathbb{V}\xepsd(t)+\partial_v \mc R(t,\xepsd(t))+D_x\mc E^\rho(t,\xeps(t))\ni 0,\\
\xeps(0)=x^\eps_0,\quad\xepsd(0)=x^\eps_1,
\end{cases}
\end{equation}
where for the sake of clarity we do not stress the dependence on $\rho$ of the solution. We recall that we are always assuming \ref{hyp:E1}, \ref{hyp:E3}--\ref{hyp:E5}, \ref{hyp:RK} and considering $\mathbb{M},\mathbb{V}$ as in Section~\ref{sec:Setting}, in particular satisfying \eqref{boundsassumption1} and \eqref{boundsassumption2}.\par
As a first step we present an alternative formulation of \eqref{eq:trun_dynprob}, based on the definition of subdifferential. We emphasize that the following results, where not otherwise explicitly stated, hold also for the original dynamic problem \eqref{dynprob}, corresponding to  $\rho=+\infty$. In particular, the uniform estimates with respect to the initial data of Corollary~\ref{unifbound} for the original dynamic problem will be employed later in the paper.

\begin{prop}\label{equivdyn}For every $\eps>0$ and $\rho\in(0,+\infty]$, a function $\xeps\in W^{2,1}(0,T; X)$ is a differential solution of \eqref{eq:trun_dynprob} if and only if initial data are attained and the following dynamic local stability condition and dynamic energy balance hold true:
\begin{enumerate}[label=\textup{(LS$^\eps$)}]
		\item \label{LSeps} for a.e.~time $t\in[0,T]$ and for every $v\in X$
		\begin{equation*}
		\mc R(t,v)+\langle D_x\mc E^\rho(t,\xeps(t))+\eps^2 \mathbb{M}\xepsdd(t)+\eps\mathbb{V}\xepsd(t),v\rangle\ge 0;
		\end{equation*}
\end{enumerate}
\begin{enumerate}[label=\textup{(EB$^\eps$)}]
		\item \label{EBeps} for every $t\in[0,T]$
		\begin{align*}
		&\quad\,\frac{\eps^2}{2}|\xepsd(t)|^2_{\mathbb{M}}+\mc E^\rho(t,\xeps(t))+\int_{0}^{t}\!\!\mc R(\tau,\xepsd(\tau))\d\tau+\eps\int_{0}^{t}|\xepsd(\tau)|^2_{\mathbb{V}}\d\tau\\
		&=\frac{\eps^2}{2}|x^\eps_1|^2_{\mathbb{M}}+\mc E^\rho(0,x^\eps_0)+\int_{0}^{t}\frac{\partial}{\partial t}\mc E^\rho(\tau,\xeps(\tau))\d\tau.
		\end{align*}
	\end{enumerate}
\end{prop}
\begin{proof}
	By definition of subdifferential we deduce that $\xeps\in W^{2,1}(0,T; X)$ is a differential solution of \eqref{eq:trun_dynprob} if and only if initial data are attained and for a.e. $t\in [0,T]$ and for every $\tilde{v}\in X$ it holds:
	\begin{equation}\label{equivsub}
	\begin{aligned}
	&\quad\,\mc R(t,\tilde v)+\langle D_x\mc E^\rho(t,\xeps(t))+\eps^2 \mathbb{M}\xepsdd(t)+\eps\mathbb{V}\xepsd(t),\tilde v\rangle\\
	&\ge\mc R(t,\xepsd(t))+\langle D_x\mc E^\rho(t,\xeps(t))+\eps^2\mathbb{M}\xepsdd(t)+\eps\mathbb{V}\xepsd(t),\xepsd(t)\rangle.
	\end{aligned}		 
	\end{equation}
	We thus conclude if we show that \eqref{equivsub} is equivalent to \ref{LSeps} and \ref{EBeps}. \par 
	We first assume that \eqref{equivsub} holds true. We fix $v\in X$ and we choose $\tilde{v}=nv$, with $n\in\N$; by means of the one homogeneity of $\mc R(t,\cdot)$ and letting $n\to +\infty$ we deduce the validity of \ref{LSeps}. Choosing $\tilde{v}=0$ and exploiting \ref{LSeps}, we instead get the following local energy balance (also called power balance):
\begin{enumerate}[label=\textup{(LEB$^\eps$)}]
		\item \label{LEBeps} for a.e.~time $t\in[0,T]$ it holds
		\begin{equation*}
		\mc R(t,\xepsd(t))+\langle D_x\mc E^\rho(t,\xeps(t))+\eps^2 \mathbb{M}\xepsdd(t)+\eps\mathbb{V}\xepsd(t),\xepsd(t)\rangle=0.
		\end{equation*}
	\end{enumerate}
	Integrating \ref{LEBeps} between $0$ and $t$ we finally get \ref{EBeps}. Indeed we recall that, since $\xeps$ is absolutely continuous, the map $\mc E^\rho(\cdot,\xeps(\cdot))$ is absolutely continuous too and $\frac{\partial}{\partial t}\mc E^\rho(\cdot,\xeps(\cdot))$ is summable in $[0,T]$.\par 
	We now assume that \ref{LSeps} and \ref{EBeps} hold true. By differentiating \ref{EBeps} we easily get \ref{LEBeps}; combining it with \ref{LSeps} we thus obtain \eqref{equivsub} and we conclude.
\end{proof}
Thanks to the energy balance \ref{EBeps} we are able to infer the following uniform bound of the involved energy along a differential solution. As we said before we assume that the initial data are uniformly bounded with respect to $\eps$ since this result will be useful also for the next sections.
\begin{prop}\label{energybound}
	Assume that the initial data satisfy \eqref{unifdata} and let $\xeps$ be a differential solution of \eqref{eq:trun_dynprob}. Then there exists a positive constant $\widetilde C_\Lambda>0$, independent of $\eps>0$ and of $\rho\in(0,+\infty]$, such that:
	\begin{equation}\label{unifest}
	\frac{\eps^2}{2}|\xepsd(t)|^2_{\mathbb{M}}+\mc E^\rho(t,\xeps(t))+\int_{0}^{t}\mc R(\tau,\xepsd(\tau))\d\tau+\eps \int_{0}^{t}|\xepsd(\tau)|^2_{\mathbb{V}}\d\tau\le \widetilde C_\Lambda,\quad\text{ for every }t\in[0,T].
	\end{equation}
\end{prop}
\begin{proof}
	We denote by $\mc F^\eps(t)$ the left-hand side of \eqref{unifest}. By means of the energy balance \ref{EBeps}, together with the estimates \eqref{boundsassumption1} and \eqref{eq:Erho4}, we deduce that the following inequality holds true for every $t\in[0,T]$:
	\begin{align*}
	\mc F^\eps (t)&=\frac{\eps^2}{2}|x^\eps_1|^2_{\mathbb{M}}+\mc E^\rho(0,x^\eps_0)+\int_{0}^{t}\frac{\partial}{\partial t}\mc E^\rho(\tau,\xeps(\tau))\d\tau\le  C+\int_{0}^{t}\omega(\mc E^\rho(\tau,\xeps(\tau)))\gamma(\tau)\d\tau\\
	&\le  C+\int_{0}^{t}\omega(\mc F^\eps(\tau))\gamma(\tau)\d\tau.
	\end{align*}
	We now conclude by means of Lemma~\ref{Gronwall}.
\end{proof}
As a simple corollary we deduce:
\begin{cor}\label{unifbound}
	Assume that  the initial data satisfy \eqref{unifdata} and let $\xeps$ be a differential solution of \eqref{eq:trun_dynprob}. Then there exists a positive constant $C_\Lambda>0$, independent of $\eps>0$ and of $\rho\in(0,+\infty]$, such that:
	\begin{enumerate}[label=\textup{(\roman*)}]
		\item \label{prop:stima_i} $\max\limits_{t\in[0,T]}|\xeps(t)|< C_\Lambda$;
		\item \label{prop:stima_ii} $\displaystyle \int_{0}^{T}\mc R(\tau,\xepsd(\tau))\d\tau< C_\Lambda$;
		\item \label{prop:stima_iii} $\max\limits_{t\in[0,T]}\eps|\xepsd(t)|_{\mathbb{M}}< C_\Lambda$.
	\end{enumerate}
\end{cor}
\begin{proof}
	The bounds in \ref{prop:stima_ii} and \ref{prop:stima_iii} simply follow from \eqref{unifest}. To get \ref{prop:stima_i} we recall that $\xeps$ belongs to $W^{2,1}([0,T];X)$, and hence by using  \ref{hyp:R2} we obtain:
	\begin{align*}
	|\xeps(t)|\le|x^\eps_0|+|\xeps(t)-x^\eps_0|\le \Lambda+\int_{0}^{t}|\xepsd(\tau)|\d\tau\le \Lambda+\frac{1}{\alpha_*}\int_{0}^{t}\mc R(\tau,\xepsd(\tau))\d\tau.
	\end{align*}
	We indeed notice that $\xepsd(t)$ is forced to live in $K$ for almost every time $t\in[0,T]$, otherwise $\partial_v\RR(t,\xepsd(t))$ would be empty or alternatively \ref{prop:stima_ii} could not be valid. Thus we conclude by \ref{prop:stima_ii}.
\end{proof}

Let us now recall a notion of normal cone suitable to our framework. For a convex subset  $\KK\subset X$ and a positive definite, symmetric linear operator $\A\colon X\to X^*$, we denote with $\NN^\A_\KK(x)$ the normal cone to the set $\KK$ in the point $x\in \KK$ with respect to the scalar product $\scal{\A \cdot}{\cdot}\colon X\times X\to \R$, namely
\begin{equation}
\NN^{\A}_\KK(x):=\{v\in X\mid\langle \A v, \tilde x-x\rangle\leq 0 \quad\text{for every $\tilde x\in \KK$}\}.
\end{equation}
If $x$ instead does not belong to $\KK$, for convention we set $\NN^{\A}_\KK(x):=\emptyset$. If finally the scalar product is the one endowed to the space, we simply write $\NN_\KK(x)$.

We also recall an existence and uniqueness result for the second order perturbed sweeping process, see \cite{Adly16}.
\begin{thm} \label{th:exist_2sweep}
	Let $E$ be an Euclidean space, $\KK\subseteq E$ a non-empty closed convex subset, and $F\colon[0,T]\times E\times \KK\rightrightarrows E$ an upper semicontinuous  set-valued map with non-empty compact convex values and satisfying for every $(t,\eta,\mu)\in [0,T]\times E\times \KK$ the bound
	\begin{equation*}
	F(t,\eta,\mu)\subseteq \beta(1+\abs{\eta}_E+\abs{\mu}_E)\BB_1^E,
	\end{equation*}
	where $\BB_1^E$ is the open unitary ball in $E$ centered at the origin.
 Then, for every $(\eta_0,\eta_1)\in E\times \KK$, the problem
	\begin{equation} \label{eq:sweep2order}
		\begin{cases}
		\ddot \eta(t)\in -\NN_\KK(\dot \eta(t))-F(t,\eta(t),\dot\eta(t)),\\
		\eta(0)=\eta_0, \qquad \dot{\eta}(0)=\eta_1,
		\end{cases}
	\end{equation}
	admits at least one differential solution, namely a function $\eta\in W^{2,1}(0,T; E)$ such that the differential inclusion holds true for a.e.~$t\in[0,T]$ and the initial data are attained. Moreover it actually holds  $\eta\in W^{2,\infty}(0,T; E)$.
\end{thm}
\begin{thm}\label{th:unique_2sweep}
	Under the assumptions of Theorem~\ref{th:exist_2sweep}, suppose in addition that there exists an open set $\UU\subseteq E$ such that 
	\begin{enumerate}[label=\textup{(j)}]
	\item\label{cond:j} every solution $\eta$ of \eqref{eq:sweep2order} satisfies $\eta(t)\in\UU$ for every $t\in[0,T]$;
	\end{enumerate}
\begin{enumerate}[label=\textup{(jj)}]
	\item\label{cond:jj} there exists a function $k\in L^1(0,T)$ such that
	\begin{equation*}
	\scal{f_1-f_2}{\mu_1-\mu_2}_E\geq -k(t)(\abs{\eta_1-\eta_2}_E^2+\abs{\mu_1-\mu_2}^2_E),
	\end{equation*}
	 for a.e.~$t\in[0,T]$ and for every $\eta_1,\eta_2\in \UU$, $\mu_1,\mu_2\in \KK$, $f_1\in F(t,\eta_1,\mu_1)$,  $f_2\in F(t,\eta_2,\mu_2)$. 
	\end{enumerate}
	 Then the solution of \eqref{eq:sweep2order} provided by Theorem~\ref{th:exist_2sweep} is unique.
\end{thm}

The existence Theorem~\ref{th:exist_2sweep} is a special case of \cite[Theorem 3.1]{Adly16}.
The uniqueness Theorem~\ref{th:exist_2sweep} is instead a straightforward corollary of \cite[Theorem 3.3]{Adly16}, noticing that once a uniform bound \ref{cond:j} on the solutions is available, it is sufficient to require \ref{cond:jj} in a region $\UU$ where the solutions are contained.\par 
In the next proposition we translate these results in our framework, obtaining existence (and uniqueness) of solutions to \eqref{eq:trun_dynprob}, but only for $\rho\in (0,+\infty)$.
\begin{prop} \label{prop:trunc_exist}
	Fix $\eps>0$. For every initial values $x^\eps_0\in X$ and $x^\eps_1\in K$, and for every $\rho\in (0,+\infty)$, there exists at least a differential solution $\xeps\in W^{2,\infty}(0,T;X)$ to problem  \eqref{eq:trun_dynprob}.

Moreover, let us assume that also \ref{hyp:E7} holds. We take $\Lambda:=\max\{\abs{x^\eps_0},\abs{\eps x^\eps_1}\}$ and consider $C_\Lambda$ to be as in Corollary~\ref{unifbound}. Then for every  $\rho\in(C_\Lambda,+\infty)$ the solution of \eqref{eq:trun_dynprob} is unique.	
\end{prop}
\begin{proof}
Let us recall that by \ref{hyp:RK} and the linearity of the subdifferential with respect to the sum of two convex functions, we can write
\begin{equation*}
	\partial_v\RR(t,v)=\partial \chi_K(v)+\partial_v\RR_\mathrm{finite}(t,v),\quad\text{for every }(t,v)\in [0,T]\times X.
\end{equation*}
 Hence we can rewrite  problem \eqref{eq:trun_dynprob} as
	\begin{equation}
\begin{cases}
\eps^2 \mathbb{M}\xepsdd(t) \in -\partial \chi_K(\xepsd(t))-\widetilde F(t,\xeps(t),\xepsd(t)),\\
\xeps(0)=x^\eps_0,\quad\xepsd(0)=x^\eps_1,
\end{cases}
\end{equation}
where
\begin{equation*}
	\widetilde F(t,u,v):=\eps \mathbb{V}v+\partial_v\RR_\mathrm{finite}(t,v)+D_x\mc E^\rho(t,u).
\end{equation*}
We now observe that, by \ref{prop:R_IV} and \ref{prop:R_V} in Corollary~\ref{propertiesR}, the map $\partial_v\RR_\mathrm{finite}\colon [0,T]\times K\rightrightarrows X^*$ has compact, convex, non-empty values and it is upper semicontinuous. Thus trivially also the map $\widetilde F\colon[0,T]\times X\times K\rightrightarrows X^*$ has compact, convex, non-empty values and it is upper semicontinuous on the whole domain. Moreover, by \eqref{boundsassumption2}, \ref{prop:R_V} in Corollary~\ref{propertiesR} and \eqref{DxErho}, for every $\rho\in(0,+\infty)$ there exists a constant $\tilde \beta_\rho>0$ such that
\begin{equation} \label{eq:preLineargrowth}
	\widetilde{F}(t,u,v)\subseteq \tilde\beta_\rho(1+\abs{v})\,\BB_1^{X^*}, \qquad\text{for every $(t,u,v)\in [0,T]\times X\times K$},
\end{equation}
where $\BB_1^{X^*}$ is the open unitary ball in $X^*$ centered at the origin.

Let us now set $\Q^\eps:=\eps^{-2}\M^{-1} \colon  X^*\to X$, so that $\Q^\eps$ is a positive definite, symmetric linear operator. Using also that $K$ is a closed, convex cone, for every $\eta\in X^*$ we have 
\begin{align*}
\partial\chi_K(\Q^\eps\eta)&=\{\xi\in X^*\mid \chi_K(\Q^\eps\eta)+\scal{\xi}{x}\leq \chi_K(\Q^\eps \eta+x)\quad\text{for every $x\in X$}\}\\
&=\{\xi\in X^*\mid \chi_K(\Q^\eps\eta)+\scal{\xi}{\Q^\eps \zeta}\leq \chi_K(\Q^\eps (\eta+\zeta)) \quad\text{for every $\zeta\in X^*$}\}\\
&=\{\xi\in X^*\mid \chi_{\M K}(\eta)+\scal{\xi}{\Q^\eps \zeta}\leq \chi_{\M K}(\eta+\zeta) \quad\text{for every $\zeta\in X^*$}\}\\
&=\{\xi\in X^*\mid \chi_{\M K}(\eta)+\scal{\xi}{\Q^\eps (\tilde \eta-\eta)}\leq \chi_{\M K}(\tilde \eta) \quad\text{for every $\tilde \eta\in X^*$}\}\\
&=\NN^{\Q^\eps}_{\M K}(\eta).
\end{align*}
In the third step we have used the fact that $K$ is a cone to neglect the factor $\eps^2$.
The last step follows by observing that both sets are empty if $\eta\notin \mathbb{M}K$, since the inequality would fail for $\tilde\eta \in \mathbb{M} K$. On the other hand, if $\eta\in \mathbb{M}K$, the inequality is always true for $\tilde\eta \notin \mathbb{M}K$, while it is equivalent to $\scal{\xi}{\Q^\eps(\tilde \eta-\eta)}\leq 0$ for $\tilde\eta \in \mathbb{M}K$.

Let us now introduce the Euclidean space $E$ as the vector space $X^*$ endowed with the scalar product $\scal{\cdot}{\Q^\eps\cdot}$ with $\Q^\eps$ as above. By \eqref{boundsassumption1} we observe that
\begin{equation}\label{eq:scal_equiv}
	\frac{1}{\eps \sqrt{M}}\abs{\eta}_*\leq \abs{\eta}_E\leq \frac{1}{\eps \sqrt{m}}\abs{\eta}_*,\quad\text{for every }\eta\in E.
\end{equation}

Then, $\xeps$ is a differential solution of \eqref{eq:trun_dynprob} if and only if $\eta^\eps:=\eps^2\M\xeps$ is  a differential solution of the following second order perturbed sweeping process on $E$:
	\begin{equation} \label{eq:sweep2form}
\begin{cases}
\ddot\eta^\eps(t) \in -\NN_{\M K}(\dot \eta^\eps(t))- F(t,\eta^\eps(t),\dot\eta^\eps(t)),\\
\eta^\eps(0)=\eps^2\M x^\eps_0,\quad\dot\eta^\eps(0)=\eps^2\M x^\eps_1,
\end{cases}
\end{equation}
where the function $F\colon[0,T]\times E\times \M K\rightrightarrows E$ is defined by
\begin{equation*}
	F(t,u,v):=\widetilde F(t, \Q^\eps u,\Q^\eps v).
\end{equation*}
We observe that, by \eqref{eq:scal_equiv} and the linearity of $\Q^\eps$, we have
that the map $F$ has compact, convex, non-empty values and is upper semicontinuous on the whole domain with respect to the norm of $E$. Moreover, by \eqref{eq:preLineargrowth} and \eqref{eq:scal_equiv},  for every $\rho\in(0,+\infty)$ there exists a constant $\beta_\rho>0$ such that
\begin{equation*}
F(t,u,v)\subseteq \beta_\rho(1+\abs{v}_E)\BB_1^E,\qquad\text{for every $(t,u,v)\in [0,T]\times E\times \mathbb{M}K$},
\end{equation*}
where $\BB_1^E$ is the unitary ball in $E$ centered at the origin.
We have therefore verified all the hypotheses of Theorem~\ref{th:exist_2sweep}, hence proving the existence of a solution $\eta^\eps\in W^{2,\infty}(0,T;E)$ of \eqref{eq:sweep2form}.  Noticing that $\xeps=\Q^\eps\eta^\eps\in W^{2,\infty}(0,T;X)$, we complete the first part of the proof.

\medskip
It remains to show that such a solution is unique. Therefore, let us now consider $\rho\in(C_\Lambda,+\infty)$ and assume \ref{hyp:E7}, with the consequence that also \eqref{eq:Erho7} holds.
 
Since to every solution $\eta^\eps$ of \eqref{eq:sweep2form} corresponds a solution $\xeps=\Q^\eps\eta^\eps$ of \eqref{eq:trun_dynprob}, which by Corollary~\ref{unifbound} is contained in the open ball $\BB^{X}_{C_\Lambda}$, we deduce that every solution  $\eta^\eps$ of \eqref{eq:sweep2form} is contained in the set $\UU:=\eps^2\M\BB^{X}_{C_\Lambda}$, which is open also in the topology of $E$. Hence condition \ref{cond:j} of Theorem~\ref{th:unique_2sweep} is satisfied. 

We then observe that the function $\widetilde{F}$ can be decomposed in two parts.
The first part $\widetilde F^a(t,v):=\eps \mathbb{V}v+\partial_v\RR_\mathrm{finite}(t,v)$, at each time $t$, is included in the subdifferential with respect to $v$ of a convex function, namely $\widetilde F^a(t,v)\subseteq\partial_v[\eps\scal{\V v}{v}+\RR_\mathrm{finite}(t,v)]$. Hence by monotonicity of the subdifferential it holds:
\begin{equation*}
	\scal{\widetilde f^a_1-\widetilde f^a_2}{v_1-v_2}\geq 0,
\end{equation*}
for every $t\in [0,T]$, $v_1,v_2\in K$, $\widetilde f^a_1\in \widetilde F^a(t,v_1)$, $\widetilde f^a_2\in \widetilde F^a(t,v_2)$. Therefore, taking $\mu_1=\eps^2\mathbb{M} v_1$ and $\mu_2= \eps^2\mathbb{M} v_2$, we infer that
 \begin{equation} \label{eq:stimauniq1}
 \scal{\widetilde f^a_1-\widetilde f^a_2}{\mu_1-\mu_2}_E=\scal{\widetilde f^a_1-\widetilde f^a_2}{\Q^\eps\mu_1-\Q^\eps\mu_2}=\scal{\widetilde f^a_1-\widetilde f^a_2}{v_1-v_2}\geq 0,
 \end{equation}
for every $t\in [0,T]$, $\mu_1,\mu_2\in \M K$, $\widetilde f^a_1\in \widetilde F^a(t,\Q^\eps\mu_1)$, $\widetilde f^a_2\in \widetilde F^a(t,\Q^\eps\mu_2)$.

 Let us now consider the second part $\widetilde F^b(t,u):=D_x\EE^\rho(t,u)$ of $\widetilde{F}$. By 
 \eqref{eq:Erho7} there exists a function $\widetilde\varsigma_\rho\in L^1(0,T)$ such that
\begin{equation*}
|\widetilde F^b(t,u_1)-\widetilde F^b(t,u_2)|_*\le\widetilde\varsigma_\rho(t)\abs{u_1-u_2},
\end{equation*}
for a.e.~$t\in [0,T]$, and for every $u_1,u_2\in \BB^X_{C_\Lambda}$. As before, taking $\eta_1= \eps^2\mathbb{M}u_1$ and $\eta_2= \eps^2\mathbb{M}u_2$, we deduce that
\begin{align}
|\widetilde F^b(t,\Q^\eps\eta_1)-\widetilde F^b(t,\Q^\eps\eta_2)|_E &\leq \frac{1}{\eps \sqrt{m}}|\widetilde F^b(t,\Q^\eps\eta_1)-\widetilde F^b(t,\Q^\eps\eta_2)|_*\notag\\
&\le\frac{\widetilde\varsigma_\rho(t)}{\eps \sqrt{m}}|\Q^\eps\eta_1-\Q^\eps\eta_2|\leq \frac{\widetilde\varsigma_\rho(t)}{\eps^2 \sqrt{mM}}|\eta_1-\eta_2|_E,
\label{eq:stimauniq2}
\end{align}
which therefore holds for a.e. $t\in [0,T]$, and every $\eta_1,\eta_2\in \UU$.

Hence, by combining \eqref{eq:stimauniq1} and \eqref{eq:stimauniq2} we obtain
\begin{align*}
	\scal{f_1-f_2}{\mu_1-\mu_2}_E&\geq
	\scal{\widetilde F^b(t,\Q^\eps\eta_1)-\widetilde F^b(t,\Q^\eps\eta_2)}{\mu_1-\mu_2}_E\\&\geq -|\widetilde F^b(t,\Q^\eps\eta_1)-\widetilde F^b(t,\Q^\eps\eta_2)|_E\abs{\mu_1-\mu_2}_E \\
	&\geq - \frac{\widetilde\varsigma_\rho(t)}{2\eps^2 \sqrt{mM}} (\abs{\eta_1-\eta_2}_E^2+\abs{\mu_1-\mu_2}^2_E),
\end{align*}
for a.e.~$t\in[0,T]$, and for every $\eta_1,\eta_2\in \UU$, $\mu_1,\mu_2\in \mathbb{M}K$, $f_1\in F(t,\eta_1,\mu_1)$,  $f_2\in F(t,\eta_2,\mu_2)$. 

Hence also condition \ref{cond:jj} of Theorem~\ref{th:unique_2sweep} is satisfied, yielding the uniqueness result of the proposition.
\end{proof}
The main result of this section, concerning the original problem \eqref{dynprob}, is a straightforward corollary of Proposition~\ref{prop:trunc_exist}.
	\begin{thm}\label{existencedyn}
	Fix $\eps>0$, let $\mathbb{M}, \mathbb{V}$ be as in Section~\ref{sec:Setting}, and assume that $\RR$ satisfies \ref{hyp:RK} and $\mc E(t,x)=\mc E_\mathrm{sh}(t,\piz (x))$ satisfies \ref{hyp:E1}, \ref{hyp:E3}--\ref{hyp:E5}. Then for every initial values $x^\eps_0\in X$ and $x^\eps_1\in K$ there exists at least a differential solution $\xeps\in W^{2,\infty}(0,T;X)$ to problem \eqref{dynprob}.\par 
	If in addition \ref{hyp:E7} holds, then such a solution is unique.
\end{thm}
\begin{proof}
Let us set $\Lambda:=\max\{\abs{x^\eps_0},\abs{\eps x^\eps_1}\}$. Taken $C_\Lambda>0$ given by Corollary~\ref{unifbound}, we fix $\rho\in(C_\Lambda,+\infty)$. 

We observe that by definition of the truncated energy $\mc E^\rho$ the two problems \eqref{dynprob} and \eqref{eq:trun_dynprob} coincide in the region $(t,x^\eps,\xepsd)\in[0,T]\times \BB^X_\rho \times K$; moreover, by Corollary~\ref{unifbound}, the solutions of both the initial value problems are contained in that region. Hence, the solutions of \eqref{dynprob} and \eqref{eq:trun_dynprob} coincide. Since by Proposition~\ref{prop:trunc_exist} problem \eqref{eq:trun_dynprob} admits at least one differential solution $x^\eps$, which additionally satisfies $\xeps\in W^{2,\infty}(0,T;X)$ and which is unique if also \ref{hyp:E7} is satisfied, so does the original dynamic problem \eqref{dynprob}.
\end{proof}



\section{$\mc R$-absolutely continuous functions  and functions of bounded $\mc R$-variation}\label{secACBV}
In this section we introduce and present the main properties of the analogue of absolutely continuous (vector-valued) functions and of functions of bounded variation when the norm $|\cdot|$ is replaced by a general time-dependent functional $\mc R$. These two notions will be useful to deal with both problems \eqref{dynprob} and \eqref{quasprob}. Here we consider the case of a reflexive Banach space $X$ and instead of limiting ourselves to potentials $\RR$ satisfying \ref{hyp:RK} we consider the larger class of $\psi$-regular functionals used in \cite{HeiMielk} (but still with the additional coercivity assumption, see \ref{psi4} below). This choice is motivated by two reasons: first of all we provide new results which are not investigated in \cite{HeiMielk} and thus we prefer to state them in the broadest possible setting; furthermore all the proofs here presented would not be simplified by restricting to our more specific framework. We want also to recall that a more general theory can be developed even in a metric setting, see for instance \cite{AmbrGiglSav}, Chapter~1.\par
We follow the presentation given in \cite{HeiMielk} for the definition and the main features of functions of bounded $\mc R$-variation when $\mc R$ depends on time, and we provide some more properties we will need during the paper. We also refer to the Appendix of \cite{Brez} for a very well detailed presentation of the classical case in which $\mc R$ is the norm of the Banach space $X$.

We thus consider a reflexive Banach space $X$ and a $\psi$-regular function $\mc R\colon [a,b]\times X\to [0,+\infty]$ in the sense of the following Definition, see also \cite{HeiMielk}:
\begin{defi}\label{psiregular}
		Given an \emph{admissible} function $\psi\colon X\to [0,+\infty]$, namely satisfying 
		\begin{enumerate}[label=\textup{($\psi\arabic*$)}]
		\item \label{psi0} $ \psi(0)=0$;
		\item \label{psi1} $ \psi$ is convex;
		\item \label{psi2} $ \psi$ is positively homogeneous of degree one;
		\item \label{psi3} $\psi$ is lower semicontinuous;
		\item \label{psi4} there exists a positive constant $c>0$ such that $c|\cdot|\le \psi(\cdot)$,
	\end{enumerate}
		we say that $\mc R\colon [a,b]\times X\to[0,+\infty]$ is \emph{$\psi$-regular} if:
		\begin{itemize}
			\item for every $t\in[a,b]$, $\mc R(t,\cdot)$ is convex, positively homogeneous of degree one, lower semicontinuous, and satisfies $\RR(t,0)=0$;
			\item there exist two positive constants $\alpha^*\ge\alpha_*>0$ for which
			\begin{equation}\label{Rbound}
				\alpha_*\psi(v)\le \mc R(t,v)\le \alpha^*\psi(v),\quad\text{ for every }(t,v)\in[a,b]\times X;
			\end{equation}
			\item there exists a nonnegative and nondecreasing function $\sigma\in \CC^0([0,b-a])$ satisfying $\sigma(0)=0$ and for which
			\begin{equation}\label{Rcont}
				|\mc R(t,v)-\mc R(s,v)|\le\psi(v)\sigma(t-s),\text{ for every }a\le s\le t\le b \text{ and for every }v\in X \text{ s.t. }\psi(v)<+\infty.
			\end{equation}
		\end{itemize}
	\end{defi}
	\begin{rmk}
	    We again notice that this definition actually differs from the one considered in \cite{HeiMielk} due to the additional assumption \ref{psi4}, which gives coercivity. Most of the results of this section are however valid without \ref{psi4}, as the reader can check from the proofs; we always stress the points where it is really necessary.
	\end{rmk}
	We want to point out that if $\RR$ satisfies \ref{hyp:RK}, then it is $\psi^K$-regular (with an absolutely continuous $\sigma$) with respect to the admissible function
	\begin{equation}\label{psiK}
\psi^K(v)=\chi_K(v)+\abs{v},
\end{equation}
where $K$ is given by \ref{hyp:RK}. On the other hand, any $\psi$-regular functional $\mc R$ can be written as
\begin{equation*}
    \RR(t,v)=\chi_{\{\psi<+\infty\}}(v)+\RR_{|_{\{\psi<+\infty\}}}(t,v),
\end{equation*}
where $\RR_{|_{\{\psi<+\infty\}}}$ has finite values due to \eqref{Rbound} and the set ${\{\psi<+\infty\}}$ is a nonempty convex cone thanks to \ref{psi0}--\ref{psi2}. However, in general, this set is not closed and moreover the second inequality in \eqref{Rbound} cannot be improved to \ref{hyp:R2}, since no bounds from above for $\psi$ are available. These are the main differences between $\psi$-regular functionals and functionals satisfying \ref{hyp:RK}.\par 
We first deal with the notion of $\mc R$-absolutely continuous functions:
\begin{defi}
	We say that a function $f\colon [a,b]\to X$ is $\mc R$-absolutely continuous, and we write $f\in AC_\mc R([a,b];X)$ if $f$ is absolutely continuous and $\displaystyle \int_{a}^{b}\mc R(\tau,\dot{f}(\tau))\d\tau< +\infty$.
\end{defi}
Next proposition provides a natural link between $\mc R$-absolutely continuous and classical absolutely continuous functions.
\begin{prop}\label{equivrabs}
	Given a function $f\colon [a,b]\to X$, the following are equivalent:
	\begin{enumerate}[label=\textup{(\arabic*)}]
		\item \label{prop:rabs1} $f$ is $\mc R$-absolutely continuous;
		\item \label{prop:rabs2} $f$ is absolutely continuous and $\displaystyle \int_{a}^{b}\psi(\dot{f}(\tau))\d\tau<+\infty$;
		\item \label{prop:rabs3} there exists a nonnegative function $m\in L^1(a,b)$ such that:
		\begin{equation*}
		\psi(f(t)-f(s))\le \int_{s}^{t}m(\tau)\d\tau,\quad\text{ for every }a\le s\le t\le b.
		\end{equation*}
	\end{enumerate}
\end{prop}
\begin{proof}
	The equivalence between \ref{prop:rabs1} and \ref{prop:rabs2} follows by means of \eqref{Rbound}.\par 
	Now assume \ref{prop:rabs2}. Then for every $a\le s\le t\le b$ we have:
	\begin{equation*}
	\psi(f(t)-f(s))=\psi\left(\int_{s}^{t}\dot{f}(\tau)\d\tau\right)\le\int_{s}^{t}\psi(\dot{f}(\tau))\d\tau,
	\end{equation*}
	where in the last step we used Jensen's inequality together with \ref{psi2}. Since $\psi(\dot{f}(\cdot))$ is summable we obtain \ref{prop:rabs3} with $m(t)=\psi(\dot{f}(t))$. \par
	If instead we assume \ref{prop:rabs3}, then by \ref{psi4} we get that $f$ is absolutely continuous, so $\dot{f}$ is well defined almost everywhere in $[a,b]$ as a (strong) limit of differential quotients. By means of \ref{psi2} and \ref{psi3} we thus deduce:
	\begin{align*}
	\psi(\dot{f}(\tau))\le\liminf\limits_{h\searrow 0}\frac{\psi(f(\tau+h)-f(\tau))}{h}\le\liminf\limits_{h\searrow 0}\frac{1}{h}\int_{\tau}^{\tau+h}m(\theta)\d\theta=m(\tau),\quad\text{ for a.e. }\tau\in[a,b],
	\end{align*}
	which implies $\displaystyle \int_{a}^{b}\psi(\dot{f}(\tau))\d\tau\le\int_{a}^{b}m(\tau)\d\tau<+\infty$.
\end{proof}
\begin{rmk}
    In the special case of a potential $\RR$ satisfying \ref{hyp:RK}, namely when $\psi$ has the form \eqref{psiK}, from \ref{prop:rabs2} we deduce that $f\in AC_\mc R([a,b];X)$ if and only if $f$ is absolutely continuous and $\dot{f}(t)\in K$ for almost every time $t\in [a,b]$.
\end{rmk}

Recalling that the notion  of functions of bounded $\mc R$-variation has already been introduced in Definition \ref{def:RBV}, we make some additional remarks and present some of their properties.
\begin{rmk}
	We want to say that the limit in \eqref{rvar} exists and it does not depend on the fine sequence of partitions chosen, thus the Definition is well-posed. If $\mc R$ does not depend on time, the limit in \eqref{rvar} can be replaced by a supremum. For a proof of these facts we refer to \cite{HeiMielk}, Appendix A.
\end{rmk}
\begin{rmk}[\textbf{Notation}]
	During the section it will be useful to consider the variation of a function with respect to the time-independent function $\mc R(\bar t,\cdot)$, namely when the time $t=\bar t$ is frozen. In this case we denote the variation by $V_{\mc R(\bar t\,)}(f;s,t)$. We notice that $V_{\mc R(\bar t\,)}(f;s,t)$ can be obtained by replacing $\mc R(t_k, f(t_k)-f(t_{k-1}))$ with $\mc R(\bar t, f(t_k)-f(t_{k-1}))$ in \eqref{rvar}, or by taking the supremum over finite partitions since the frozen potential does not depend on time.
\end{rmk}
From the Definition \ref{def:RBV} we easily notice that \eqref{Rbound} allows us to deduce that a function $f$ belongs to $BV_{\mc R}([a,b];X)$ if and only it it is a function of bounded $\psi$-variation, i.e. $V_\psi(f;a,b)<+\infty$; moreover by \ref{psi4} we deduce that $f$ is a function of bounded variation in the classical sense. As a byproduct, see for instance the Appendix in \cite{Brez}, we obtain that any $f\in BV_{\mc R}([a,b];X)$ has at most a countable number of discontinuity points, and at every $t\in[a,b]$ there exist right and left (strong) limits of $f$, namely:
\begin{equation}\label{rllimit}
f^+(t):=\lim\limits_{t_k\searrow t}f(t_k),\quad\text{ and }\quad f^-(t):=\lim\limits_{t_k\nearrow t}f(t_k).
\end{equation}
\begin{rmk}
	Given a function $f\colon[a,b]\to X$, with a little abuse of notation we will always consider, and still denote, by $f$ its constant extension to a slightly larger interval $(a-\delta,b+\delta)$, for some $\delta>0$; namely $f(t)=f(a)$ if $t\in(a-\delta,a]$ and $f(t)=f(b)$ if $t\in[b,b+\delta)$. This ensures that the limits in \eqref{rllimit} are well defined also in $t=a,b$ and in particular it holds $f^-(a)=f(a)$ and $f^+(b)=f(b)$.
\end{rmk} 
\begin{rmk}
	In the particular case in which $\RR$ satisfies \ref{hyp:RK}, namely when $\psi$ is given by \eqref{psiK}, it is easy to see that $f\in BV_\mc R([a,b];X)$ if and only if $f$ has bounded variation (in the classical sense) and $f(t)-f(s)\in K$ for every $a\le s\le t\le b$.
\end{rmk} 
	Trivially the $\mc R$-variation of $f$ is monotone in both entries (see \ref{var:a} in the next proposition), thus for every $a\le s\le t\le b$ they are well defined:
\begin{equation*}
\begin{gathered}
V_{\mc R}(f;s,t+):=\lim\limits_{t_k\searrow t}	V_{\mc R}(f;s,t_k),\quad V_{\mc R}(f;s,t-):=\lim\limits_{s\le t_k,t_k\nearrow t}	V_{\mc R}(f;s,t_k),\\ 
V_{\mc R}(f;s-,t):=\lim\limits_{s_k\nearrow s}V_{\mc R}(f;s_k,t),\quad V_{\mc R}(f;s+,t):=\lim\limits_{s_k\le t,s_k\searrow s}	V_{\mc R}(f;s_k,t),\\
V_{\mc R}(f;s-,t+):=\lim\limits_{s_k\nearrow s,t_k\searrow t}	V_{\mc R}(f;s_k,t_k),\\
V_{\mc R}(f;s-,t-):=\lim\limits_{s_k\le t_k,s_k\nearrow s,t_k\nearrow t}	V_{\mc R}(f;s_k,t_k),\\
V_{\mc R}(f;s+,t+):=\lim\limits_{s_k\le t_k,s_k\searrow s,t_k\searrow t}	V_{\mc R}(f;s_k,t_k).
\end{gathered}
\end{equation*}
Next proposition gathers all the properties of the $\mc R$-variation we will need throughout the paper.
\begin{prop}\label{propertiesvariationtime}
	Given a function $f\colon [a,b]\to X$, the following properties hold true:
	\begin{enumerate}[label=\textup{(\alph*)}]
		\item \label{var:a} for every $a\le r \le s\le t\le b$ it holds:
		\begin{equation*}
		V_{\mc R}(f;r,t)=V_{\mc R}(f;r,s)+V_{\mc R}(f;s,t);
		\end{equation*}
		\item \label{var:b} for every $a\le s\le t\le b$ it holds:
		\begin{align*}
		V_{\mc R}(f;s-,t+)=V_{\mc R}(f;s-,s)+V_{\mc R}(f;s,t)+V_{\mc R}(f;t,t+);
		\end{align*}
		\item \label{var:c} if $f\in BV_{\mc R}([a,b];X)$, then for every $t\in [a,b]$ the following equalities hold true:
		\begin{align*}
		&V_{\mc R}(f;t,t+)=V_{\mc R(t)}(f;t,t+)=\lim\limits_{t_k\searrow t}\mc R(t,f(t_k)-f(t)),\quad V_{\mc R}(f;t,t-)=	0,\\
		&V_{\mc R}(f;t-,t)=V_{\mc R(t)}(f;t-,t)=\lim\limits_{t_k\nearrow t}\mc R(t,f(t)-f(t_k)),\quad V_{\mc R}(f;t+,t)=	0\\
		&V_{\mc R}(f;t-,t-)=	0,\quad V_{\mc R}(f;t+,t+)=	0;
		\end{align*}
		\item \label{var:d} if $f\in BV_{\mc R}([a,b];X)$, then $f^+,f^-$ belong to $BV_\mc R([a,b],X)$ and for every $a\le s\le t\le b$ the following inequalities hold true:
		\begin{align*}
		&V_{\mc R}(f;s-,t+)\ge \max\left\{V_{\mc R}(f^+;s-,t+),V_{\mc R}(f^-;s-,t+)\right\},\\
		&V_{\mc R}(f;s+,t+)\ge V_{\mc R}(f^+;s,t+),\\
		&V_{\mc R}(f;s-,t-)\ge V_{\mc R}(f^-;s-,t).
		\end{align*}
	\end{enumerate}
\end{prop}
\begin{proof}
For \ref{var:a} it is enough to take a fine sequence of partions of $[r,t]$ containing $s$. The proof of \ref{var:b} follows easily by \ref{var:a}.\par
The only nontrivial part in \ref{var:c} are the two equalities: 
\begin{equation}\label{toprove}
	V_{\mc R}(f;t,t+)=V_{\mc R(t)}(f;t,t+),\quad\text{ and }\quad V_{\mc R}(f;t-,t)=V_{\mc R(t)}(f;t-,t).
\end{equation}
We prove only the first one, the other being analogous. Exploiting \eqref{Rcont} we deduce that for every $t'>t$ we have:
\begin{align*}
	|V_{\mc R}(f;t,t')-V_{\mc R(t)}(f;t,t')|&\le \limsup\limits_{n\to +\infty}\sum_{k=1}^{n}|\mc R(t_{k-1}, f(t_k)-f(t_{k-1}))-\mc R(t, f(t_k)-f(t_{k-1}))|\\
	&\le \limsup\limits_{n\to +\infty}\sum_{k=1}^{n}\psi(f(t_k)-f(t_{k-1}))\sigma(t_{k-1}-t)\\
	&\le V_\psi(f;t,t')\sigma(t'-t),
\end{align*}
where $\{t_k\}_{k=1}^n$ is a fine sequence of partitions of $[t,t']$. Letting now $t'\searrow t$ we get \eqref{toprove}.\par 
As regards the first inequality in \ref{var:d}, it is enough to prove	\begin{equation}
V_{\mc R}(f;s',t')\ge \max\left\{V_{\mc R}(f^+;s',t'),V_{\mc R}(f^-;s',t')\right\},
\end{equation} 
where $s'<s\le t<t'$ are continuity points of $f$. So we fix $\delta>0$ and a fine sequence of partition of $[s',t']$. Then, exploiting lower semicontinuity and \eqref{Rcont}, for any of these partitions there exists another partition of $[s',t']$, made of continuity points of $f$ and such that each point $\tilde{t}_{k-1}$ belongs to $[t_{k-1},t_k)$, which satisfies:
\begin{align*}
	&\quad\,\sum_{k=1}^{n}\mc R(t_{k-1}, f^+(t_k)-f^+(t_{k-1}))\le \sum_{k=1}^{n}\mc R(t_{k-1}, f(\tilde t_k)-f(\tilde t_{k-1}))+\delta\\
	&\le \sum_{k=1}^{n}\mc R(\tilde t_{k-1}, f(\tilde t_k)-f(\tilde t_{k-1}))+\sum_{k=1}^{n}\psi(f(\tilde t_k)-f(\tilde t_{k-1}))\sigma(\tilde t_{k-1}-t_{k-1})+\delta\\	
	&\le \sum_{k=1}^{n}\mc R(\tilde t_{k-1}, f(\tilde t_k)-f(\tilde t_{k-1}))+V_\psi(f;s',t')\sup_{k=1,\dots n}\sigma(t_k-t_{k-1}) +\delta.
\end{align*}
By letting first $n\to +\infty$ and then $\delta\to 0$, recalling \eqref{finezza} and the uniform continuity of $\sigma$, we get $V_\mc R(f;s',t')\ge V_\mc R(f^+;s',t')$, and arguing in a similar way we also obtain $V_\mc R(f;s',t')\ge V_\mc R(f^-;s',t')$, thus the first inequality in \ref{var:d} is proved. \par 
We now prove the second inequality of \ref{var:d}. We fix $t'>t$ a continuity point of $f$, we consider $\delta>0$ and a fine sequence of partitions of $[s,t']$. As before, for any of these partitions there exist continuity points of $f$ such that each point $\tilde{t}_{k-1}$ belongs to $(t_{k-1},t_k)$ and they satisfy:
\begin{align*}
	&\quad\,\sum_{k=1}^{n}\mc R(t_{k-1}, f^+(t_k)-f^+(t_{k-1}))\\
	&\le \sum_{k=1}^{n}\mc R(\tilde t_{k-1}, f(\tilde t_k)-f(\tilde t_{k-1}))+V_\psi(f;s',t')\sup_{k=1,\dots n}\sigma(t_k-t_{k-1}) +\delta\\
	&=\sum_{k=1}^{n}\mc R(\tilde t_{k-1}, f(\tilde t_k)-f(\tilde t_{k-1}))+\mc R(s,f(\tilde t_0){-}f(s))-\mc R(s,f(\tilde t_0){-}f(s))\\
	&\quad+V_\psi(f;s',t')\sup_{k=1,\dots n}\sigma(t_k-t_{k-1})+\delta.
\end{align*}
Letting $n\to +\infty$, thanks to \eqref{finezza}, we deduce
\begin{align*}
	V_\mc R(f^+;s,t')\le V_\mc R(f;s,t')-V_{\mc R(s)}(f;s,s+)+\delta= V_\mc R(f;s+,t')+\delta.
\end{align*}
Letting now $\delta\to 0$ and $t'\searrow t$ we deduce $V_\mc R(f;s+,t+)\ge 	V_\mc R(f^+;s,t+)$.\par 
The third inequality in \ref{var:d} follows in a similar way, thus we conclude.
\end{proof}
As in the classical case, the inclusion $AC_\mc R([a,b];X)\subseteq BV_\mc R([a,b];X)$ holds true, as stated in the next proposition:
\begin{prop}\label{abscontvar}
	A function $f\colon [a,b]\to X$ is $\mc R$-absolutely continuous if and only if it is of bounded $\mc R$-variation and the function $t\mapsto V_\mc R(f;a,t) $ is absolutely continuous. In this case it holds 
	\begin{equation*}
		V_\mc R(f;s,t)=\int_{s}^{t}\mc R(\tau,\dot{f}(\tau))\d\tau,\quad\text{ for every }a\le s\le t\le b.
	\end{equation*}
\end{prop}
\begin{proof}
	Assume $f$ is $\mc R$-absolutely continuous. We fix $a\le s\le t\le b$ and we consider a fine sequence of partitions of $[s,t]$. Thanks to \eqref{Rbound} and \eqref{Rcont} we estimate:
	\begin{align*}
		\sum_{k=1}^{n}\mc R(t_{k-1},f(t_k)-f(t_{k-1}))&\le \sum_{k=1}^{n}\int_{t_{k-1}}^{t_k}\mc R(t_{k-1},\dot f(\tau))\d\tau\\
		&\le\sum_{k=1}^{n}\left(\int_{t_{k-1}}^{t_k}\mc R(\tau,\dot f(\tau))\d\tau +\int_{t_{k-1}}^{t_k}\psi(\dot{f}(\tau))\sigma(\tau-t_{k-1})\d\tau \right)\\
		&\le \int_{s}^{t}\mc R(\tau,\dot f(\tau))\d\tau+\sup\limits_{k=1,\dots n}\sigma(t_k-t_{k-1})\int_{s}^{t}\psi(\dot f(\tau))\d\tau.
	\end{align*}
	Letting $n\to +\infty$ (we again recall \eqref{finezza}) we deduce 
	\begin{equation}\label{firstineq}
		V_\mc R(f;s,t)\le\int_{s}^{t}\mc R(\tau,\dot{f}(\tau))\d\tau,
	\end{equation}
	thus $f$ is of bounded $\mc R$-variation and the $\mc R$-variation is absolutely continuous.\par 
	To obtain also the other implication and the opposite inequality in \eqref{firstineq} we argue as follows: first of all we notice that \eqref{Rbound} implies:
	\begin{equation}\label{estpsi}
		V_\mc R(f;s,t)\ge \alpha_*V_\psi(f;s,t)\ge\alpha_*\psi(f(t)-f(s)),\quad\text{ for every }a\le s\le t\le b,
	\end{equation}
	and thus $f$ is $\mc R$-absolutely continuous by applying Proposition~\ref{equivrabs} (thus \ref{psi4} here is needed). To conclude, introducing the notation $v_\mc R(t):=V_\mc R(f;a,t)$, we only need to prove that $\dot{v}_\mc R(\tau)\ge \mc R(\tau,\dot{f}(\tau))$ for almost every $\tau\in[a,b]$.\par 
	With this aim we fix a point $\tau$ of differentiability for both $v_\mc R$ and $f$, and we consider $h>0$. By using \eqref{Rcont} we obtain:
	\begin{equation*}
		v_\mc R(\tau+h)-v_\mc R(\tau)=V_\mc R(f;\tau,\tau+h)\ge \mc R(\tau,f(\tau+h)-f(\tau))-V_\psi(f;\tau,\tau+h)\sigma(h).
	\end{equation*}
	Hence, letting $h\to 0$ we deduce:
	\begin{align*}
		\dot{v}_\mc R(\tau)&\ge\liminf\limits_{h\to 0}\mc R\left(\tau,\frac{f(\tau+h)-f(\tau)}{h}\right)-\lim\limits_{h\to 0}\frac 1hV_\psi(f;\tau,\tau+h)\sigma(h)\\
		&\ge \mc R(\tau,\dot{f}(\tau)),
	\end{align*}
	where the limit vanishes if we pick $\tau$ which is also a differentiability point of $V_\psi(f;a,\cdot)$, which is absolutely continuous by \eqref{estpsi}. Hence the proof is complete.
\end{proof}
Like in the classical case, the $\mc R$-variation is pointwise weakly lower semicontinuous, as stated in the following lemma:
\begin{lemma}\label{weaklscvar}
	 Let $\{f_j\}_{j\in \N}$ be a sequence of functions from $[a,b]$ to $X$ such that $f_j(t)\wto f(t)$ weakly for every $t\in[a,b]$. Then one has
	 \begin{equation*}
	 	V_\mc R(f;s,t)\le \liminf\limits_{j\to +\infty}V_\mc R(f_j;s,t),\quad\text{ for every }a\le s\le t\le b.
	 \end{equation*}
\end{lemma}
\begin{proof}
	We only sketch the proof, see the Appendix of \cite{HeiMielk} for more details. If $s=t$ the inequality is trivial, thus let us fix $a\le s< t\le b$ and without loss of generality we assume $\liminf\limits_{j\to +\infty}V_\mc R(f_j;s,t)<+\infty$. We now consider a fine sequence of partitions of $[s,t]$ and, recalling that convexity plus lower semicontinuity implies weak lower semicontinuity, we obtain:
	\begin{align}\label{first}
		\sum_{k=1}^{n}\mc R(t_{k-1},f(t_k)-f(t_{k-1}))&\le \liminf\limits_{j\to +\infty}\sum_{k=1}^{n}\mc R(t_{k-1},f_j(t_k)-f_j(t_{k-1})).
	\end{align}
	We now fix $j\in \N$ and we notice that by subadditivity (ensured by convexity and one homogeneity), \eqref{Rbound} and \eqref{Rcont} we have
	\begin{equation}\label{second}
		\begin{aligned}
		\sum_{k=1}^{n}\mc R(t_{k-1},f_j(t_k)-f_j(t_{k-1}))&\le V_\mc R(f_j;s,t)+V_\psi(f_j;s,t)\sup_{k=1,\dots n}\sigma(t_k-t_{k-1})\\
		&\le V_\mc R(f_j;s,t)\left(1+\frac{1}{\alpha_*}\sup_{k=1,\dots n}\sigma(t_k-t_{k-1})\right).
		\end{aligned}
	\end{equation}
	Combining \eqref{first} and \eqref{second} we hence deduce:
	\begin{equation*}
			\sum_{k=1}^{n}\mc R(t_{k-1},f(t_k)-f(t_{k-1}))\le \liminf\limits_{j\to +\infty}V_\mc R(f_j;s,t)\left(1+\frac{1}{\alpha_*}\sup_{k=1,\dots n}\sigma(t_k-t_{k-1})\right).
	\end{equation*}
	Letting $n\to +\infty$ and recalling \eqref{finezza} we conclude.
\end{proof}
We finally state and prove a useful generalisation in $BV_\mc R([a,b];X)$ of the following classical result: a sequence of nondecreasing and continuous scalar functions pointwise converging to a continuous function (in a compact interval) actually converges uniformly.
\begin{lemma}\label{uniformconv}
    Let $\{f_j\}_{j\in \N}\subseteq BV_\mc R([a,b];X)$ be a sequence of functions pointwise strongly converging to $f\in BV_\mc R([a,b];X)$. Assume that:
    \begin{itemize}
        \item $V_\mc R(f_j;a,\cdot)$ are continuous in $[a,b]$ for every $j\in\N$ and $V_\mc R(f;a,\cdot)$ is continuous in $[a,b]$;
        \item $\lim\limits_{j\to +\infty} V_\mc R(f_j;a,t)=V_\mc R(f;a,t)$, for every $t\in[a,b]$.
    \end{itemize}
    Then the (strong) convergence of $f_j$ to $f$ is actually uniform in $[a,b]$.
\end{lemma}
\begin{proof}
	We denote for simplicity $v^j_\mc R(t):=V_\mc R(f_j;a,t)$ and $v_\mc R(t):=V_\mc R(f;a,t)$. By assumptions and since the $\mc R$-variation is nondecreasing, we deduce that $\{v^j_\mc R\}_{j\in\N}$ is a sequence of nondecreasing and continuous functions pointwise converging to the nondecreasing continuous function $v_\mc R$; this implies that the convergence is actually uniform in $[a,b]$.\par
	We now fix  $s,t\in[a,b]$ and we estimate by using \ref{psi4} and \eqref{Rbound}:
	\begin{align*}
	c\alpha_*|f_j(t)-f_j(s)|&\le |v^j_\mc R(t)-v^j_\mc R(s)|\\
	&\le|v_\mc R(t)-v_\mc R(s)|+|v^j_\mc R(t)-v_\mc R(t)|+|v^j_\mc R(s)-v_\mc R(s)|\\
	&\le |v_\mc R(t)-v_\mc R(s)|+2\max\limits_{\tau\in[a,b]}|v^j_\mc R(\tau)-v_\mc R(\tau)|.
	\end{align*}
	Since $v^j_\mc R$ uniformly converges to $v_\mc R$ and $v_\mc R$ is (uniformly) continuous on $[a,b]$, we get that for every $\eps>0$ there exist $j_\eps\in\N$ and $\delta_\eps>0$ such that, assuming $|t-s|\le\delta_\eps$, it holds:
	\begin{equation}\label{equicont}
		|f_j(t)-f_j(s)|\le \frac{\eps}{3}, \quad\text{ for every }j>j_\eps.
	\end{equation}
	So we fix $\eps>0$ and we consider a finite partition of $[a,b]$ of the form $a=\tau_0<\tau_1<\dots<\tau_{N_\eps}=b$ such that $\max\limits_{k=1,\dots N_\eps}(\tau_k-\tau_{k-1})\le\delta_\eps$. This means that for every $t\in[a,b]$ there exists a point of this partition, denoted by $\tau(t)$, for which $|t-\tau(t)|\le\delta_\eps$. Without loss of generality we can assume that $\delta_\eps$ is also the treshold given by the (uniform) continuity of $f$ (indeed notice that $f$ is continuous since $v_\mc R$ is continuous by assumption). Thus by means of \eqref{equicont} we deduce that for every $j>j_\eps$ and for every $t\in[a,b]$ we have:
	\begin{align*}
		|f_j(t)-f(t)|&\le |f_j(t)-f_j(\tau(t))|+|f_j(\tau(t))-f(\tau(t))|+|f(t)-f(\tau(t))|\\
		&\le \frac \eps 3+\max\limits_{k=0,\dots,N_\eps}|f_j(\tau_k)-f(\tau_k)|+\frac \eps 3.
	\end{align*}
	Since the maximum in the above estimate involves only a finite number of terms, by means of the assumption of pointwise convergence and by considering a possibly greater $J_\eps\ge j_\eps$ we conclude that for every $t\in [a,b]$ it holds
	\begin{equation*}
		|f_j(t)-f(t)|\le \eps,\quad\text{ for every }j>J_\eps,
	\end{equation*}
	and we conclude.
\end{proof}



\section{Differential and energetic solutions for the quasistatic problem}\label{secenergetic}

In this section we discuss the quasistatic problem \eqref{quasprob} and in particular the notion of \emph{energetic solution}, which we recalled in Definition \ref{defenergetic}. Hence all the assumption of the quasistatic problem \eqref{quasprob}, namely \ref{hyp:E1}--\ref{hyp:E5} and \ref{hyp:RK}, hold here. The main purpose of this section is to prove temporal regularity of the energetic solutions to \eqref{quasprob}, which we obtain in Proposition~\ref{regularity}. Such regularity will allow us to deduce the equivalence between the two notions of energetic and differential solutions. We also present some well known cases in which uniqueness for energetic (and differential) solutions holds; we point out that for a general elastic energy, as the one we consider here, the question of uniqueness is still open.

To start, we notice that, in the quasistatic setting, it is possible to provide a characterisation of differential solutions analogous to that of Proposition~\ref{equivdyn} for the dynamic problem. In fact, convexity leads to a better result, which also clarifies Definition \ref{defenergetic} of energetic solutions.
\begin{prop}\label{propdiffenerg}
	A function $x\in AC([0,T];X)$ is a differential solution of the quasistatic problem \eqref{quasprob} if and only if the initial position is attained and one of the following two equivalent conditions is satisfied:
	\smallskip
	\begin{enumerate}[label=\textup{(\arabic*)}]
	\addtolength{\itemindent}{-0.5 cm}
	\item \label{cond:quasi1} $\Biggl\{$\!
	\begin{minipage}{0.9\textwidth}
		\begin{enumerate}[label=\textup{(LS)}]
	\item \label{LS}\, $\mc R(t,v)+\langle D_x \mc E(t,x(t)),v\rangle\ge 0$ \quad for every $t\in[0,T]$ and for every $v\in X$;
	\end{enumerate}
\vspace{1mm}
	\begin{enumerate}[label=\textup{(LEB)}]
	\item \label{LEB}\, $\mc R(t,\dot{x}(t))+\langle D_x \mc E(t,x(t)),\dot{x}(t)\rangle= 0$\quad for a.e. $t\in[0,T]$;
	\end{enumerate}
	\end{minipage}
	\vspace{10 pt}
	\item \label{cond:quasi2} $\left\{\rule{0cm}{8mm}\right.$\!\!
	\begin{minipage}[c][8mm][c]{0.9\textwidth}
		\vspace{4pt}
	\begin{enumerate}[label=\textup{(GS)}]
	\item \label{diffGS}\, $\mc E(t,x(t))\le \mc E(t,v)+\mc R(t,v-x(t))$ \quad for every $t\in[0,T]$  and $v\in X$;
	\end{enumerate}
\begin{enumerate}[label=\textup{(EB)}]	
	\item \label{EB}\, $\displaystyle \mc E(t,x(t))+\!\!\int_{0}^{t}\!\!\!\mc R(\tau,\dot{x}(\tau))\d\tau=\mc E(0,x_0)+\!\!\int_{0}^{t}\!\frac{\partial}{\partial t} \mc E(\tau,x(\tau))\d\tau$\quad for every $t\in[0,T]$.
	\end{enumerate}
\end{minipage}	
\end{enumerate}
\smallskip
\end{prop}
\begin{proof}
	The fact that $x\in AC([0,T];X)$ is a differential solution of \eqref{quasprob} if and only if the initial position is attained and \ref{cond:quasi1} is fulfilled follows by arguing as in the proof of Proposition~\ref{equivdyn}. Notice that the passage from a.e. to every time is granted by continuity.  We only need to show that \ref{cond:quasi1} and \ref{cond:quasi2} are equivalent; first of all we notice that \ref{LEB} is equivalent to \ref{EB} since we can obtain the first one by differentiating the second one. The fact that \ref{diffGS} implies \ref{LS} follows since $\mc R(t,\cdot)$ is one homogeneous, while the contrary follows since the function $v\mapsto \mc E(t,x(t)+v)$ is convex by \ref{hyp:E2}.
\end{proof}
\begin{rmk}
	As the reader can check from the proof, convexity assumption \ref{hyp:E2} is needed only to deduce the global stability \ref{diffGS} from the local one \ref{LS}.
\end{rmk}
\begin{rmk}
    We point out that, by \ref{EB}, any differential solution of \eqref{quasprob} is actually $\RR$-absolutely continuous. In particular, due to Proposition~\ref{abscontvar}, it is an energetic solution. 
\end{rmk}
We now pass to the main object of this section, namely the temporal regularity of energetic solutions. The argument follows the already consolidated ideas of \cite{Mielk}, \cite{MielkRoubbook}, and \cite{Thom}; the first step exploits uniform convexity to improve the estimate furnished by the global stability condition \ref{GS}. However, since in our setting uniform convexity holds only for the restricted energy $\mc E_\mathrm{sh}$, we need to introduce also the notion of restricted dissipation potential from \cite{Gid18}.

Given any functional $\Phi\colon X\to [0,+\infty]$ we thus define its (shape-)restricted version $\Phi_\mathrm{sh}\colon Z\to [0,+\infty]$ in the following way:
\begin{equation}\label{restricted}
\Phi_\mathrm{sh}(z):=\inf\{\Phi(x)\mid x\in X \text{ and } \piz(x)=z\}.
\end{equation}
The following properties are a straightforward consequence of the definition of $\Phi_\mathrm{sh}$:
\begin{itemize}
	\item if $\Phi^1\le \Phi^2$ on $X$, then $\Phi^1_\mathrm{sh}\le \Phi^2_\mathrm{sh}$ on $Z$;
	\item $\Phi_\mathrm{sh}(\piz(x))\le \Phi(x)$ for every $x\in X$;
	\item if $\Phi$ is positively homogeneous of degree one, then $\Phi_\mathrm{sh}$ is positively homogeneous of degree one.
\end{itemize}

Notice that not all the properties of $\RR$ are inherited by $\RR_\mathrm{sh}$: for instance, to obtain an upper bound analogous to \ref{prop:regR2} it is necessary to require  \ref{hyp:R5}, as we show in the following lemma.

\begin{lemma} \label{lemma:Zbound} Suppose in addition that $\RR$ satisfies \ref{hyp:R5}. If $(t,z)\in[0,T]\times Z$ is such that $\RR_\mathrm{sh}(t,z)<+\infty$, then
	\begin{equation*}
		\RR_\mathrm{sh}(t,z)\leq \alpha^*C_K\abs{z}_Z,
	\end{equation*} 
	with $\alpha^*$ and $C_K$ as in \ref{hyp:R2} and \ref{hyp:R5}, respectively.
\end{lemma}
\begin{proof}
	Since $\RR_\mathrm{sh}(t,z)<+\infty$, there exists $\tilde x\in K$ such that $\piz(\tilde x)=z$. Thus, by \ref{hyp:R5} it is possible to select this $\tilde x$ in such a way that $\abs{\tilde x}\leq C_K \abs{z}_Z$. Hence, recalling Corollary~\ref{propertiesR}, we have
\begin{equation*}
\RR_\mathrm{sh}(t,z)\leq \RR(t,\tilde x) \leq \alpha^*\abs{\tilde x}\leq \alpha^*C_K\abs{z}_Z,
\end{equation*}
and we conclude.
\end{proof}

We now prove that the global stability condition \ref{GS} is actually equivalent to an enhanced version of stability.
\begin{lemma}[\textbf{Improved Stability}]
	Fix $t\in[0,T]$. If $x^*\in X$ satisfies
	\begin{equation}\label{stabstar}
	\mathcal{E}(t,x^*)\le\mathcal{E}(t,x)+\mathcal{R}(t,x-x^*),\quad\text{for every }x\in X,
	\end{equation}
	then also the following stronger version of stability holds true:
	\begin{equation}\label{improvedgs}
	\mathcal{E}(t,x^*)+\frac{\mu}{2}|\piz(x^*)-\piz(x)|_Z^2\le\mathcal{E}(t,x)+\mathcal{R}_\mathrm{sh}(t,\piz(x)-\piz(x^*)),\quad\text{for every }x\in X.
	\end{equation}
\end{lemma}
\begin{proof}
	From the definition of restricted dissipation potential \eqref{restricted} and recalling that $\mc E(t,\cdot)=\mc E_\mathrm{sh}(t,\piz(\cdot))$, we deduce that \eqref{stabstar} implies:
	\begin{equation}\label{stabstarsh}
	\mathcal{E}(t,x^*)\le\mathcal{E}(t,x)+\mathcal{R}_\mathrm{sh}(t,\piz(x)-\piz(x^*)),\quad\text{for every }x\in X.
	\end{equation}
	Furthermore, by means of \ref{hyp:E2} we know that for every $x_1,x_2\in X$ and for every $\theta\in(0,1)$ it holds:
	\begin{equation}\label{theta}
	\mc E(t,\theta x_1+(1-\theta)x_2)\le \theta\mc E(t,x_1)+(1-\theta)\mc E(t,x_2)-\frac{\mu}{2}\theta(1-\theta)|\piz(x_1)-\piz(x_2)|_Z^2.
	\end{equation}
	We now fix $x\in X$ and we choose $\theta x+(1-\theta)x^*$ as competitor for $x^*$ in \eqref{stabstarsh}; by using the one-homogeneity of $\mc R_\mathrm{sh}(t,\cdot)$, the linearity of $\piz$, and \eqref{theta}, we get:
	\begin{align*}
	\mathcal{E}(t,x^*)&\le \mc E(t,\theta x+(1-\theta)x^*)+\mc R_\mathrm{sh}(t,\theta(\piz(x)-\piz(x^*)))\\
	&\le \theta\mc E(t,x)+(1-\theta)\mc E(t,x^*)-\frac{\mu}{2}\theta(1-\theta)\abs{\piz(x)-\piz(x^*)}_Z^2+\theta\mc R_\mathrm{sh}(t,\piz(x)-\piz(x^*)).
	\end{align*}
	By subtracting $\mathcal{E}(t,x^*)$ from both sides and dividing by $\theta$ we hence obtain:
	\begin{align*}
	0\le \mc E(t,x)-\mc E(t,x^*)-\frac{\mu}{2}(1-\theta)\abs{\piz(x)-\piz(x^*)}_Z^2+\mc R_\mathrm{sh}(t,\piz(x)-\piz(x^*)).
	\end{align*}
	We conclude letting $\theta\searrow 0$.
\end{proof}
Next lemma will be used in the proof of Proposition~\ref{regularity}.
\begin{lemma}\label{inequalityintegral}
	Let $(V,\Vert\cdot\Vert)$ be a normed space and let $f\colon[a,b]\to V$ be a bounded measurable function such that:
	\begin{equation}\label{hypo}
	\Vert f(t)-f(s)\Vert^2\le \int_{s}^{t}\Vert f(t)-f(\tau)\Vert g(\tau)d\tau+\Vert f(t)-f(s)\Vert\int_{s}^{t}h(\tau)d\tau,\quad\text{for every }a\le s\le t\le b,
	\end{equation}
	for some nonnegative $g,h\in L^1(a,b)$. Then it holds:
	\begin{equation*}
	\Vert f(t)-f(s)\Vert\le \int_{s}^{t}\big(g(\tau)+h(\tau)\big)d\tau,\quad\text{for every }a\le s\le t\le b.
	\end{equation*}
\end{lemma}
\begin{proof}
	Fix $t\in[a,b]$. For $s\in[a,t]$ we define the functions $\beta_t(s):=\Vert f(t)-f(s)\Vert$ and $\overline{\beta_t}(s):=\sup\limits_{\theta\in[s,t]}\beta_t(\theta)$, where the latter is finite since $f$ is bounded.
	
	We now fix $s\in [a,t]$ and, by using \eqref{hypo}, for every $\theta\in[s,t]$ we hence obtain:
	\begin{align*}
	\beta_t(\theta)^2&\le \int_{\theta}^{t}\beta_t(\tau)g(\tau)\d\tau+\beta_t(\theta)\int_{\theta}^{t}h(\tau)\d\tau\\
	&\le \overline{\beta_t}(s)\int_{s}^{t}\big(g(\tau)+h(\tau)\big)d\tau,
	\end{align*}
	which implies
	\begin{equation*}
	\overline{\beta_t}(s)^2\le \overline{\beta_t}(s)\int_{s}^{t}\big(g(\tau)+h(\tau)\big)d\tau,\quad\text{for every }a\le s\le t\le b.
	\end{equation*}
	Since $\beta_t(s)\le\overline{\beta_t}(s)$, we conclude.
\end{proof}
We are now in a position to state and prove the main result of this section:
\begin{prop}\label{regularity}
Assume that $\RR$ satisfies \ref{hyp:RK} and $\mc E(t,x)=\mc E_\mathrm{sh}(t,\piz (x))$ satisfies \ref{hyp:E1}--\ref{hyp:E5}. Then any energetic solution $x$ for \eqref{quasprob} is continuous.\par 
Suppose in addition that \ref{hyp:R5} holds or, alternatively, that $\mc R$ does not depend on time. Then $x$ is $\mc R$--absolutely continuous and, therefore, a differential solution of \eqref{quasprob}.
\end{prop}
\begin{proof}
	We fix $0\le s\le t\le T$; since $x$ satisfies \ref{GS} we can pick $x(t)$ as a competitor for $x(s)$ in \eqref{improvedgs}, getting:
	\begin{align*}
	&\quad\,\frac{\mu}{2}|\piz(x(t))-\piz(x(s))|_Z^2\\&\le \mc E(s,x(t))+\mc R_\mathrm{sh}(s,\piz(x(t))-\piz(x(s)))-\mc E(s,x(s))\\
	&=\mc E(s,x(t))-\mc E(t,x(t))+\mc E(t,x(t))-\mc E(s,x(s))+\mc R_\mathrm{sh}(s,\piz(x(t))-\piz(x(s)))\\
	&=\!\int_{s}^{t}\!\!\!\Big(\frac{\partial}{\partial t}\mc E(\tau,x(\tau)){-}\frac{\partial}{\partial t}\mc E(\tau,x(t))\Big)\!\d\tau+\mc R_\mathrm{sh}(s,\piz(x(t)){-}\piz(x(s)))-V_\mc R(x;s,t),
	\end{align*}
	where for the last equality we exploited \ref{WEB}.\par 
	We recall that $x$ is bounded since it belongs to $BV_{\mc R}([0,T];X)$; thus there exists $R>0$ such that $|x(t)|\le R$ for every $t\in[0,T]$. Hence we can use \ref{hyp:E5} and continue the above inequality:
	\begin{equation}\label{estimate1}
	\begin{aligned}
	&\quad\,\frac{\mu}{2}|\piz(x(t)){-}\piz(x(s))|_Z^2\\&\le\int_{s}^{t}|\piz(x(t)){-}\piz(x(\tau))|_Z\eta_R(\tau)\d\tau+\mc R_\mathrm{sh}(s,\piz(x(t)){-}\piz(x(s)))-V_\mc R(x;s,t).
	\end{aligned}
	\end{equation}
	To estimate the term outside the integral we exploit \ref{hyp:R2} and \ref{hyp:R3}, getting:
	\begin{align*}
	V_\mc R(x;s,t)&\ge V_{\mc R(s)}(x;s,t)-V(x;s,t)\int_{s}^{t}\rho(\tau)\d\tau\\
	&\ge \left(1-\frac{1}{\alpha_*}\int_{s}^{t}\rho(\tau)\d\tau\right)V_{\mc R(s)}(x;s,t).
	\end{align*}
	The above inequality finally implies:
	\begin{equation}\label{importantestimate}
	V_\mc R(x;s,t)\ge \left(1-\frac{1}{\alpha_*}\int_{s}^{t}\rho(\tau)\d\tau\right)\RR_\mathrm{sh}(s,\piz(x(t))-\piz(x(s))).
	\end{equation}
	Indeed, if the term within parentheses is negative the inequality is trivial; otherwise we observe that $V_{\mc R(s)}(x;s,t)\ge \mc R(s,x(t)-x(s))\ge \RR_\mathrm{sh}(s,\piz(x(t))-\piz(x(s)))$.\par
	By plugging \eqref{importantestimate} into \eqref{estimate1} we thus obtain
	\begin{equation}\label{est2}
	\begin{aligned}
	&\quad\,\frac{\mu}{2}|\piz(x(t)){-}\piz(x(s))|_Z^2\\&\le\int_{s}^{t}|\piz(x(t)){-}\piz(x(\tau))|_Z\eta_R(\tau)\d\tau+\frac{1}{\alpha_*}\left(\int_{s}^{t}\rho(\tau)\d\tau \right) \RR_\mathrm{sh}(s,\piz(x(t)){-}\piz(x(s))).
	\end{aligned} 
	\end{equation}
	Since $x$ is bounded, we deduce that $|\piz(x(t))-\piz(x(\tau))|_Z$ is bounded by a constant independent of $t$ and $\tau$. Moreover, by \ref{prop:regR2} in Corollary~\ref{propertiesR}, we have
	\begin{align*}
\RR_\mathrm{sh}(s,\piz(x(t))-\piz(x(s)))&\leq \RR(s,x(t)-x(s)) \leq  V_{\RR (s)}(x;s,t)\\
&\leq \frac{\alpha^*}{\alpha_*}V_{\RR }(x;s,t)\leq \frac{\alpha^*}{\alpha_*} V_{\RR}(x;0,T).
	\end{align*}
Hence, from estimate \eqref{est2} we infer:
	\begin{equation*}
	|\piz(x(t))-\piz(x(s))|_Z\le C\left(\int_{s}^{t}\big(\eta_R(\tau)+\rho(\tau))\d\tau\right)^\frac 12,
	\end{equation*}
	for some constant $C>0$, and thus $\piz\circ x$ is continuous from $[0,T]$ to $Z$. Since $\mc E(t,x(t))=\mc E_\mathrm{sh}(t,\piz(x(t)))$ and $\mc E_\mathrm{sh}$ is continuous in $[0,T]\times Z$ by \ref{hyp:E1} and \ref{hyp:E3}, we easily deduce that $t\mapsto \mc E(t,x(t))$ is continuous too. Thus by \ref{WEB} we obtain that the $\mc R$-variation of $x$ is continuous as a function of $t\in [0,T]$; by employing \ref{var:c} in Proposition~\ref{propertiesvariationtime} together with \ref{hyp:R2}, we finally obtain that $x$ itself is continuous too.
	
	\medskip
	Let us now prove the $\RR$-absolute continuity of $x$ under the stronger assumptions \ref{hyp:R5} or $\RR$ autonomous. The first step is to show that both the alternative assumptions imply 
		\begin{equation}\label{mainest}
	|\piz(x(t))-\piz(x(s))|_Z\le C\int_{s}^{t}\big(\eta_R(\tau)+\rho(\tau)\big)\d\tau,\quad\text{ for every }0\le s\le t\le T,
	\end{equation}
	for some constant $C>0$. With this aim we notice that, in the case where $\mc R$ does not depend on time, the term outside the integral in \eqref{estimate1} is less or equal than zero, since in this case trivially it holds
	\begin{equation*}
	\mc R_\mathrm{sh}(\piz(x(t))-\piz(x(s)))\le \mc R(x(t)-x(s))\le V_\mc R(x;s,t).
	\end{equation*}
	Thus \eqref{mainest} follows, actually with only $\eta_R$ inside the integral, from Lemma~\ref{inequalityintegral} applied to this improved version of \eqref{estimate1}.
	
If instead $\RR$ depends on time, but satisfies \ref{hyp:R5}, we can apply Lemma~\ref{lemma:Zbound} to the rightmost term of \eqref{est2} and then apply directly Lemma~\ref{inequalityintegral} to obtain \eqref{mainest}.
	
Now that we have obtained \eqref{mainest} in both the alternative cases, the second step is to deduce $\RR$-absolute continuity. Firstly, we deduce from \eqref{mainest} that the function $\piz\circ x$ is absolutely continuous from $[0,T]$ into $Z$. We now prove that $t\mapsto \mc E(t,x(t))$ is an absolutely continuous function. With this aim we fix $0\le s\le t\le T$ and we estimate:
	\begin{align*}
	|\mc E(t,x(t))-\mc E(s,x(s))|&\le |\mc E(t,x(t))-\mc E(t,x(s))|+|\mc E(t,x(s))-\mc E(s,x(s))|\\
	&\le C_R|\piz(x(t))-\piz(x(s))|_Z+\int_{s}^{t}\left|\frac{\partial}{\partial t}\mc E(\tau,x(s))\right|\d\tau\\
	&\le C_R|\piz(x(t))-\piz(x(s))|_Z+\int_{s}^{t}\omega(\mc E(\tau,x(s)))\gamma(\tau)\d\tau.
	\end{align*}
	The second term on the right-hand side have been estimated using \ref{hyp:E4}; instead for the first term we have used the fact that  $x$ is bounded by some $R>0$ and, by \ref{hyp:E3} and compactness, $\EE_\mathrm{sh}(t,\cdot)$ is Lipschitz continuous on $\overline{\BB^Z_R}$ with some constant $C_R$, which can be taken uniformly in $t\in[0,T]$.
	Moreover, since $\mc E$ is bounded on $[0,T]\times \overline{\BB^X_R}$ by continuity, from the above inequality we deduce that:
	\begin{align*}
	|\mc E(t,x(t)){-}\mc E(s,x(s))|&\le C_R|\piz(x(t)){-}\piz(x(s))|_Z+\omega(M_R)\int_{s}^{t}\!\!\!\gamma(\tau)\d\tau,\text{ for every }0\le s\le t\le T.
	\end{align*}
	Thus we proved that $t\mapsto \mc E(t,x(t))$ is absolutely continuous. We now conclude since by using \ref{WEB} we have:
	\begin{align*}
	V_{\mc R}(x;s,t)=\mc E(s,x(s))-\mc E(t,x(t))+\int_{s}^{t}\frac{\partial}{\partial t} \mc E(\tau,x(\tau))\d\tau,\quad \text{ for every }0\le s\le t\le T,
	\end{align*}
	and thus, by using Proposition~\ref{abscontvar}, $x$ is $\mc R$-absolutely continuous since $\frac{\partial}{\partial t} \mc E(\cdot,x(\cdot))\in L^1(0,T)$ thanks to \ref{hyp:E4}.
\end{proof}

We conclude this section by listing some of the known important cases in which the quasistatic problem \eqref{quasprob} admits at most one solution. In the general framework the issue of uniqueness is not completely clear yet.
We first discuss the case  $\dim Z=\dim X$, corresponding to a coercive energy $\EE$.
\begin{lemma}\label{lemma:uniqquas}
	Assume that $\dim Z=\dim X$, $\RR$ satisfies \ref{hyp:RK} and $\mc E(t,x)=\mc E_\mathrm{sh}(t,\piz (x))$ satisfies \ref{hyp:E1}--\ref{hyp:E5}. 
	Then each of the following additional assumptions is a sufficient condition for uniqueness of energetic solutions to \eqref{quasprob}:
	\begin{enumerate}[label=\textup{(U\arabic*)}]
		\item \label{hyp:U1}  $\mc R$ does not depend on time and $\EE_\mathrm{sh}$ belongs to $\CC^3([0,T]\times Z)$;
		\item \label{hyp:U2} $\mc R$ does not depend on time, $\EE_\mathrm{sh}(t,z)=\mc V(z)-\scal{g(t)}{z}$ with $\mc V$ strictly convex, $g\in AC([0,T];Z^*)$, and the stable sets $$\mc S(t)=\{z\in Z\mid \EE_\mathrm{sh}(t,z)\le \EE_\mathrm{sh}(t,w)+\mc R(w-z)\text{ for every }w\in Z \},$$ are convex for every $t\in [0,T]$;
		\item\label{hyp:U3} $K=X$ and $\EE_\mathrm{sh}$ satisfies \ref{hyp:QE} with $\ell_\mathrm{sh}\in W^{1,\infty}(0,T;Z)$.
	\end{enumerate}
\end{lemma}
\begin{proof}
	The case when $\mc R$ does not depend on time is well studied; the proof of uniqueness under \ref{hyp:U1} or \ref{hyp:U2}, and several discussions on their applicability, can be found for instance in \cite[Theorems~4.1 and 4.2]{Mielk}, or \cite[Section~3.4.4]{MielkRoubbook}, or \cite[Theorems~6.5 and 7.4]{MielkTheil}. Case \ref{hyp:U3} has been proved in \cite[Theorem~4.7]{HeiMielk}. 
\end{proof}

The locomotion case  $\dim Z<\dim X$ has been deeply analysed in \cite{Gid18} in the case of quadratic energies; in particular we mention Theorem~4.3 for the uniqueness result, and Example~3.2 to illustrate the necessity of condition \ref{hyp:star} below. We present here a generalized result applying the very same argument.

\begin{lemma}\label{lemma:uniqquas2}
Assume that  $\RR$ satisfies \ref{hyp:RK} and $\mc E(t,x)=\mc E_\mathrm{sh}(t,\piz (x))$ satisfies \ref{hyp:E1}--\ref{hyp:E5}. Suppose in addition that at least one of \ref{hyp:U1}, \ref{hyp:U2} or \ref{hyp:U3} holds, and that for almost every $t\in[0,T]$ we have
\begin{enumerate}[label=\textup{(\textasteriskcentered)}]
\item \label{hyp:star} for every $z\in Z$ with $\RR_\mathrm{sh}(t,z)<+\infty$, there exists a unique $x\in X$ such that $\piz(x)=z$ and 
\begin{equation*}
	\RR_\mathrm{sh}(t,z)=\RR(t,x)<\RR(t,v), \qquad\text{for every $v\neq x$ such that $\piz(v)=z$}.
\end{equation*}
\end{enumerate}
Then the differential solution to \eqref{quasprob} is unique. In particular, since in each case we can apply Proposition~\ref{regularity}, uniqueness holds true also for energetic solutions.
\end{lemma}
\begin{proof}
It is well known that $x(t)$ is a differential solution of 	\eqref{quasprob} if and only if it satisfies the initial condition and the \emph{variational inequality}
\begin{equation}
\scal{D_x\EE(t,x(t))}{v-\dot x(t)}+\RR(t,v)-\RR(t,\dot x(t))\geq 0, \quad \text{for every $v\in X$ and a.e.~$t\in[0,T]$}. 
\label{eq:VI}
\end{equation}
Writing $z(t):=\piz(x(t))$, inequality \eqref{eq:VI} can be equivalently split in the two conditions
\begin{gather}
\RR_\mathrm{sh}(t,\dot z(t))=\RR(t,\dot x(t))\leq\RR(t,v), \text{ for every $v\in X$ such that $\piz(v)=\dot z(t)$ and a.e.~$t\in[0,T]$;} \label{eq:Rmin}\\
\scal{D_z\EE_\mathrm{sh}(t,z(t))}{w-\dot z(t)}_Z+\RR_\mathrm{sh}(t,w)-\RR_\mathrm{sh}(t,\dot z(t))\geq 0,\text{ for every $w\in Z$ and a.e.~$t\in[0,T]$.} 
	\label{eq:VIsh}
\end{gather}
Following the same argument of \cite[Lemmata 2.1 and 4.1]{Gid18}, it can be observed that the functional $\RR_\mathrm{sh}$, defined according to \eqref{restricted}, inherits the regularity properties \ref{prop:regR1} and \ref{prop:regR3} of Corollary~\ref{propertiesR}, with also \ref{prop:regR2} if $K=X$. These, combined with the one of  \ref{hyp:U1}, \ref{hyp:U2} or \ref{hyp:U3} which is holding, allows to apply the results mentioned in the proof of the previous lemma, to obtain the uniqueness of a solution $z(t)$ of \eqref{eq:VIsh}. Hence, if two differential solutions $x_1,x_2$ of \eqref{quasprob} exist, they must satisfy $\piz(\dot x_1(t))=\piz(\dot x_2(t))=\dot z(t)$ almost everwhere.  This, combined with \eqref{eq:Rmin}, implies that $\RR(t,\dot x_1(t))=\RR(t,\dot x_2(t))$ a.e., in contradiction with \ref{hyp:star}, since $\RR(t,\dot x(t))<+\infty$ a.e. along solutions. Therefore the differential solution of \eqref{quasprob} is unique.
\end{proof}



\section{Quasistatic limit}\label{seclimit}
This section is devoted to the proof of the main result of the paper, namely we discuss the convergence as $\eps$ goes to $0$ of a differential solutions $x^\eps$ of the dynamic problems \eqref{dynprob}, given by Theorem~\ref{existencedyn}, to a (energetic or differential) solution of the quasistatic problem \eqref{quasprob}.

Hence in this section we are assuming all the basic hypotheses of the dynamic and quasistatic problems: $X$ is a finite dimensional normed space, $\mathbb{M}$ and $\mathbb{V}$ are as in Section~\ref{sec:Setting}, $\mc E(t,x)=\mc E_\mathrm{sh}(t,\piz (x))$ satisfies \ref{hyp:E1}--\ref{hyp:E5} and $\RR$ satisfies \ref{hyp:RK}. We however point out that \ref{hyp:E2}, i.e. convexity, will not be necessary for the first part of the vanishing inertia analysis, as stressed in Remark~\ref{rmk:convexity}. Moreover we assume that the initial velocity $ x^\eps_1$ satisfy the admissibility condition \eqref{eq:dyn_admiss}.

We proceed as follows. Firstly, we use the uniform bound on the energy of $\xeps$, obtained in Proposition~\ref{energybound}, to deduce the existence of a convergent subsequence by means of a compactness argument involving Helly's Selection Theorem. Then, we prove that the limit obtained from the subsequence is actually an energetic (and thus, from Proposition~\ref{regularity}, a differential) solution of the quasistatic problem \eqref{quasprob}. The main results are collected in Theorems~\ref{almostfinalthm} and \ref{finalthm}.\par

\begin{thm}\label{convsubseq}
	Assume that $x^\eps_0$ and $\eps x^\eps_1$ are uniformly bounded, namely \eqref{unifdata} is satisfied. Then there exists a subsequence $\epsj\searrow 0$ and a function $x\in BV_{\mc R}([0,T];X)$ such that:
	\begin{enumerate}[label=\textup{(\alph*)}]
		\item \label{conv:a}$\lim\limits_{j\to +\infty}\xepsj(t)=x(t)$, for every $t\in[0,T]$;
		\item\label{conv:b} $\displaystyle V_{\mc R}(x;s,t)\le \liminf\limits_{j\to +\infty}\int_{s}^{t}\mc R(\tau,\xepsjd(\tau))\d\tau$, for every $0\le s\le t\le T$;
		\item\label{conv:c}$\lim\limits_{j\to +\infty}\epsj |\xepsjd(t)|_{\mathbb{M}}=0$, for every $t\in(0,T]\setminus J_x$, where $J_x$ is the jump set of the limit function $x$.
	\end{enumerate}
\end{thm}
\begin{proof}
	By the uniform bounds \ref{prop:stima_i} and \ref{prop:stima_ii} of Corollary~\ref{unifbound} together with \ref{hyp:R2}, the family $\{\xeps\}_{\eps>0}$ is uniformly equibounded with uniformly equibounded variation. By means of the classical Helly's Selection Theorem we get the existence of a subsequence $\epsj\searrow 0$ and a function $x\in BV([0,T];X)$ for which \ref{conv:a} holds true. Thanks to Proposition~\ref{abscontvar} and Lemma~\ref{weaklscvar}, we also infer that actually $x$ belongs to $ BV_{\mc R}([0,T];X)$  and that property \ref{conv:b} holds.\par 
	To get \ref{conv:c} we first notice that, by \ref{prop:stima_ii} of Corollary~\ref{unifbound} and \ref{hyp:R2}, we deduce that 
	\begin{equation*}
		\lim\limits_{\eps\to 0}\eps\int_{0}^{T}|\xepsd(\tau)|\d\tau=0,
	\end{equation*}
	from which we can assume without loss of generality that
	\begin{equation}\label{aevanish}
		\lim\limits_{j\to +\infty}\epsj \xepsjd(t)=0,\quad\text{ for a.e. }t\in[0,T],
	\end{equation}
	which implies the validity of \ref{conv:c} almost everywhere thanks to \eqref{boundsassumption1}.\par 
	Let us now fix $t\in (0,T]\setminus J_x$ and consider two sequences $s_k\nearrow t$ and $t_k\searrow t$ at which \eqref{aevanish} holds true. By means of the energy balance (EB$^\epsj$) and exploiting the nonnegativity of $\mc R$ and $|\cdot|^2_{\mathbb{V}}$ we deduce:
	\begin{align*}
	&\quad\,\frac{\eps_j^2}{2}|\xepsjd(t_k)|^2_{\mathbb M}+\mc E(t_k,\xepsj(t_k))-\mc E(t,\xepsj(t))-\int_{t}^{t_k}\frac{\partial}{\partial t}\mc E(\tau,\xepsj(\tau))\d\tau\\
	&\qquad\le\frac{\eps_j^2}{2}|\xepsjd(t)|^2_{\mathbb M}\\
	&\qquad\le \frac{\eps_j^2}{2}|\xepsjd(s_k)|^2_{\mathbb M}+\mc E(s_k,\xepsj(s_k))-\mc E(t,\xepsj(t))+\int_{s_k}^{t}\frac{\partial}{\partial t}\mc E(\tau,\xepsj(\tau))\d\tau.
	\end{align*}
	Letting first $j\to +\infty$ we obtain:
	\begin{multline*}
	\quad\,\mc E(t_k,x(t_k))-\mc E(t,x(t))-\int_{t}^{t_k}\frac{\partial}{\partial t}\mc E(\tau,x(\tau))\d\tau\\
	\le\liminf\limits_{j\to +\infty}\frac{\eps_j^2}{2}|\xepsjd(t)|^2_{\mathbb M}\le\limsup\limits_{j\to +\infty}\frac{\eps_j^2}{2}|\xepsjd(t)|^2_{\mathbb M}\\
	\le \mc E(s_k,x(s_k))-\mc E(t,x(t))+\int_{s_k}^{t}\frac{\partial}{\partial t}\mc E(\tau,x(\tau))\d\tau.
	\end{multline*}
	Here we used the continuity of $\mc E$ and the dominated convergence theorem on the integral terms, exploiting assumption \ref{hyp:E5}.\par 
	Since $t\notin J_x$, letting now $k\to +\infty$ we prove \ref{conv:c}. 
\end{proof}
Our aim now is to prove that such a limit function $x$ is an energetic solution of problem \eqref{quasprob}; we thus need to show the validity of the global stability condition \ref{GS} and the weak energy balance \ref{WEB}. The strategy consists in passing to the limit the dynamic local stability condition \ref{LSeps} and the dynamic energy balance \ref{EBeps}. This first proposition deals with stability conditions:
\begin{prop}
	Assume that $x^\eps_0$ and $\eps x^\eps_1$ are uniformly bounded. Then the limit function $x$ obtained in Theorem~\ref{convsubseq} fulfils the following inequality:
	\begin{equation}\label{weakls}
		\int_{s}^{t}\Big(\mc R(\tau,v)+\langle D_x \mc E(\tau,x(\tau)),v\rangle\Big)\d\tau\ge 0,\quad \text{for every }v\in X\text{ and for every }0\le s\le t\le T.
	\end{equation}
	In particular the right and the left limit of $x$ are locally stable, meaning that:
	\begin{enumerate}[label=\textup{(LS$^+$)}]
		\item\label{LSplus} $\mc R(t,v)+\langle D_x \mc E(t,x^+(t)),v\rangle\ge 0,\quad \text{for every }v\in X\text{ and for every }t\in[0,T]$;
	\end{enumerate}
	\begin{enumerate}[label=\textup{(LS$^-$)}]
		\item\label{LSminus} $\mc R(t,v)+\langle D_x \mc E(t,x^-(t)),v\rangle\ge 0,\quad \text{for every }v\in X\text{ and for every }t\in(0,T]$.
	\end{enumerate}
\end{prop}
\begin{proof}
	Let $\epsj$ be the subsequence obtained in Theorem~\ref{convsubseq}. We now fix $v\in K$, being \eqref{weakls} trivial if $v\notin K$, and by integrating the local stability condition (LS$^\epsj$) between arbitrary $0\le s\le t\le T$ we deduce:
	\begin{align*}
		0&\le \int_{s}^{t}\Big(\mc R(\tau,v)+\langle D_x\mc E(\tau,\xepsj(\tau))+\eps_j^2 \mathbb{M}\xepsjdd(\tau)+\epsj\mathbb{V}\xepsjd(\tau),v\rangle\Big)\d\tau\\
		&=\int_{s}^{t}\Big(\mc R(\tau,v)+\langle D_x\mc E(\tau,\xepsj(\tau)),v\rangle\Big)\d\tau+\eps_j^2\langle\mathbb{M}(\xepsjd(t)-\xepsjd(s)),v \rangle+\eps_j\int_{s}^{t}\langle\mathbb{V}\xepsjd(\tau),v \rangle\d\tau.
	\end{align*}
	Letting $j\to +\infty$ we obtain \eqref{weakls} by dominated convergence on the first term (using \ref{hyp:E3}), while the second and the third term vanish by means of \ref{prop:stima_ii} and \ref{prop:stima_iii} of Corollary~\ref{unifbound} together with \eqref{boundsassumption1}, \eqref{boundsassumption2}, and \ref{hyp:R2}.\par 
	The validity of (LS$^\pm$) easily follows from \eqref{weakls} since by \ref{hyp:E3} and \ref{hyp:R3} the map $t\mapsto \mc R(t,v)+\langle D_x \mc E(t,x^\pm(t)),v\rangle$ is right continuous with $x^+$ and left continuous with $x^-$.
\end{proof}
Next proposition exploits the lower semicontinuity of the $\mc R$-variation (Lemma~\ref{weaklscvar}) to obtain an estimate from above of the quasistatic energy:
\begin{prop}[\textbf{Lower Energy Estimates}]
	Assume that $x^\eps_0$ and $\eps x^\eps_1$ are uniformly bounded. Then the limit function $x$ obtained in Theorem~\ref{convsubseq} fulfils the following energy inequalities:
	\begin{subequations}
		\begin{equation}\label{LEEa}
		\mc E(t,x^+(t))+V_{\mc R}(x;s-,t+)\le \mc E(s,x^-(s))+\int_{s}^{t}\frac{\partial}{\partial t}\mc E(\tau,x(\tau))\d\tau,\quad\text{ for every }0<s\le t\le T.
		\end{equation}
		\begin{equation}\label{LEEb}
		\mc E(t,x^+(t))+V_{\mc R}(x;s+,t+)\le \mc E(s,x^+(s))+\int_{s}^{t}\frac{\partial}{\partial t}\mc E(\tau,x(\tau))\d\tau,\quad\text{ for every }0\le s\le t\le T.
		\end{equation}
				\begin{equation}\label{LEEc}
				\mc E(t,x^-(t))+V_{\mc R}(x;s-,t-)\le \mc E(s,x^-(s))+\int_{s}^{t}\frac{\partial}{\partial t}\mc E(\tau,x(\tau))\d\tau,\quad\text{ for every }0< s\le t\le T.
				\end{equation}
	\end{subequations}	
	If in addition $\lim\limits_{\eps\to 0}\eps x^\eps_1=0$, then \eqref{LEEa} and \eqref{LEEc} hold true also for $s=0$.
\end{prop}
\begin{proof}
	We prove only \eqref{LEEa}, being the other inequalities analogous. We fix $0<s\le t\le T$ and we consider two sequences $s_k\nearrow s$ and $t_k\searrow t$ such that $s_k,t_k\notin J_x$. By means of Theorem~\ref{convsubseq} and by using the nonnegativity of $|\cdot|_\mathbb{V}^2$ together with the energy balance (EB$^\epsj$) we get:
	\begin{align*}
	&\quad\, \mc E(t_k,x(t_k))+V_{\mc R}(x;s_k,t_k)\\
	&\le\liminf\limits_{j\to +\infty}\left(\frac{\eps_j^2}{2}|\xepsjd(t_k)|^2_{\mathbb M}+\mc E(t_k,\xepsj(t_k))+\int_{s_k}^{t_k}\mc R(\tau,\xepsjd(\tau))\d\tau+\epsj\int_{s_k}^{t_k}|\xepsjd(\tau)|^2_{\mathbb{V}}\d\tau\right)\\
	&=\liminf\limits_{j\to +\infty}\left(\frac{\eps_j^2}{2}|\xepsjd(s_k)|^2_{\mathbb M}+\mc E(s_k,\xepsj(s_k))+\int_{s_k}^{t_k}\frac{\partial}{\partial t}\mc E(\tau,\xepsj(\tau))\d\tau \right)\\
	&=\mc E(s_k,x(s_k))+\int_{s_k}^{t_k}\frac{\partial}{\partial t} \mc E(\tau,x(\tau))\d\tau,
	\end{align*}
	where in the last equality we employed once again the continuity of $\mc E$ and \ref{hyp:E5}. Letting now $k\to +\infty$ we obtain \eqref{LEEa}.\par
	If in addition $\lim\limits_{\eps\to 0}\eps x^\eps_1=0$, the same argument works choosing $s_k\equiv0$; thus we conclude.
\end{proof}
\begin{rmk}\label{rmk:convexity}
	We want to highlight that up to this point the convexity assumption \ref{hyp:E2} was not needed. Thus even without convexity the limit function $x$ satisfies the right and left local stability conditions (LS$^\pm$) plus the energy inequality \eqref{LEEa}. Usually a function satisfying these properties is called local solution to the quasistatic problem \eqref{quasprob}, see \cite[Chapter~3]{MielkRoubbook}. Inequality \eqref{LEEa} can be also reformulated as an energy equality in a very implicit way by introducing a so called defect measure $\mu_D$ such that:
	\begin{equation*}
		\mc E(t,x^+(t))+V_{\mc R}(x;s-,t+)+\mu_D([s,t])= \mc E(s,x^-(s))+\!\int_{s}^{t}\!\!\frac{\partial}{\partial t}\mc E(\tau,x(\tau))\d\tau,\text{ for every }0\le\! s\le\! t\le\! T.
	\end{equation*}
	The positive measure $\mu_D$ is no other than the opposite of the distributional derivative of the function $\displaystyle t\mapsto \mc E(t,x(t))+V_\mc R(x;0,t)-\int_{0}^{t}\frac{\partial}{\partial t}\mc E(\tau,x(\tau))\d\tau$. The presence of such a defect measure, which somehow takes into account the possible losses of energy in the system, appears in many asymptotical studies of mechanical models: we refer for instance to \cite{AgoRos, EfMielk, MielkRosSav09, MielkRosSav12, MielkRosSav16, Roub} for a vanishing viscosity analysis and the notion of Balanced Viscosity solutions in both finite and infinite dimension, or to \cite{ScilSol} for a vanishing inertia and viscosity analysis (without a rate-independent dissipation) in finite dimension.\par 
	The fine properties of $\mu_D$ in our context where a rate-independent dissipation is also present are beyond the scopes of the present work, thus we leave this analysis open for future research. We simply notice that, as we will see in Theorem~\ref{almostfinalthm}, the (uniform) convexity assumption \ref{hyp:E2} will ensure that $\mu_D$ is the null measure.
\end{rmk}

From now on we will exploit the convexity assumption \ref{hyp:E2}. This allows us to deduce that the local conditions \ref{LSplus} and \ref{LSminus} are equivalent to their global counterpart:
\begin{enumerate}[label=\textup{(GS$^+$)}]
		\item \label{GSplus} $\mc E(t,x^+(t))\le\mc E(t,v)+\mc R(t,v-x^+(t)),\quad \text{for every }v\in X\text{ and for every }t\in[0,T]$;
		\end{enumerate}
\begin{enumerate}[label=\textup{(GS$^-$)}]
		\item \label{GSminus}$\mc E(t,x^-(t))\le\mc E(t,v)+\mc R(t,v-x^-(t)),\quad \text{for every }v\in X\text{ and for every }t\in(0,T]$.
	\end{enumerate}
These global conditions permit to get also a bound from below of the energy, see Lemma~\ref{lemmaUEE} and Proposition~\ref{propUEE}. We warn the reader that for the proof of next lemma in the case of a general elastic energy $\mc E$ we need to add the assumption \ref{hyp:E6}.

\begin{lemma}\label{lemmaUEE}
	Assume \ref{hyp:E6}. Assume that $x^\eps_0$ and $\eps x^\eps_1$ are uniformly bounded. Then the right and left limit of the function $x$ obtained in Theorem~\ref{convsubseq} fulfil the following inequalities:
	\begin{subequations}
		\begin{equation}\label{UEE+}
			\mc E(t,x^+(t))+V_{\mc R}(x^+;s,t)\ge \mc E(s,x^+(s))+\int_{s}^{t}\frac{\partial}{\partial t}\mc E(\tau,x(\tau))\d\tau,\quad\text{ for every }0\le s\le t\le T;
		\end{equation}
		\begin{equation}\label{UEE-}
			\mc E(t,x^-(t))+V_{\mc R}(x^-;s,t)\ge \mc E(s,x^-(s))+\int_{s}^{t}\frac{\partial}{\partial t}\mc E(\tau,x(\tau))\d\tau,\quad\text{ for every }0< s\le t\le T.
		\end{equation}
	\end{subequations}
If in addition $x_0:=x(0)$ satisfies \eqref{eq:adm_quasistat}, namely $\mc E(0,x_0)\le\mc E(0,v)+\mc R(0,v-x_0)$ for every $v\in X$, then \eqref{UEE-} holds true also for $s=0$.
\end{lemma}
\begin{proof}
	Inequality \eqref{UEE+} is trivially satisfied for $s=t$, so let us fix $0\le s< t\le T$ and consider a fine sequence of partitions of $[s,t]$ such that:
	\begin{equation}\label{P1}
		\lim\limits_{n\to +\infty}\sum_{k=1}^{n}\left|(t_k-t_{k-1})\frac{\partial}{\partial t}\mc E(t_k,x^+(t_k))-\int_{t_{k-1}}^{t_k}\frac{\partial}{\partial t}\mc E(\tau,x(\tau)) d\tau\right|=0.
	\end{equation}
	Such a sequence of partitions exists since $\frac{\partial}{\partial t}\mc E(\cdot,x(\cdot))\in L^1(0,T)$, see for instance \cite[Lemma~4.5]{FrancMielk}.\par 
	So let us fix one of these partitions and by means of \ref{GSplus} we deduce that for every $k=1,\dots,n$ we have:
	\begin{align*}
	\mathcal{E}(t_{k-1}, x^+(t_{k-1}))\le \mathcal{E}(t_{k-1}, x^+(t_k))+\mathcal{R}(t_{k-1},x^+(t_k)-x^+(t_{k-1})),
	\end{align*}
	and thus we obtain:
	\begin{multline*}
	\quad\, \mathcal{E}(t_k, x^+(t_k))-\mathcal{E}(t_{k-1}, x^+(t_{k-1}))+\mathcal{R}(t_{k-1},x^+(t_k)-x^+(t_{k-1}))\\
	\ge\mathcal{E}(t_k, x^+(t_k))-\mathcal{E}(t_{k-1}, x^+(t_k))=\int_{t_{k-1}}^{t_k}\frac{\partial}{\partial t}\mathcal{E}(\tau, x^+(t_k))d\tau.
	\end{multline*}
	By summing the above inequality from $k=1$ to $k=n$ we get:
	\begin{align}\label{In}
	\quad\,\mathcal{E}(t, x^+(t)){-}\mathcal{E}(s, x^+(s))+\sum_{k=1}^{n}\mathcal{R}(t_{k-1},x^+(t_k){-}x^+(t_{k-1}))
	\ge \sum_{k=1}^{n}\int_{t_{k-1}}^{t_k}\frac{\partial}{\partial t}\mc E(\tau,x^+(t_k))d\tau=:I_n.
	\end{align}
	By letting $n\to +\infty$, we get \eqref{UEE+} if we show that $\displaystyle\lim\limits_{n\to +\infty} I_n=\int_{s}^{t}\frac{\partial}{\partial t}\mc E(\tau,x(\tau))\d\tau$. To prove it we argue as follows:
	\begin{align*}
		&\quad\,\left|I_n-\int_{s}^{t}\frac{\partial}{\partial t}\mc E(\tau,x(\tau))\d\tau\right|=\left|\sum_{k=1}^{n}\int_{t_{k-1}}^{t_k}\Big(\frac{\partial}{\partial t}\mc E(\tau,x^+(t_k))-\frac{\partial}{\partial t}\mc E(\tau,x(\tau))\Big)\d\tau\right|\\
		&\le\!\sum_{k=1}^{n}\!\int_{t_{k-1}}^{t_k}\!\left|\frac{\partial}{\partial t}\mc E(\tau,x^+(t_k)){-}\frac{\partial}{\partial t}\mc E(t_k,x^+(t_k))\right|\!\d\tau\\&+\sum_{k=1}^{n}\left|(t_k{-}t_{k-1})\frac{\partial}{\partial t}\mc E(t_k,x^+(t_k)){-}\!\int_{t_{k-1}}^{t_k}\!\frac{\partial}{\partial t}\mc E(\tau,x(\tau)) \d\tau\right|\!.
	\end{align*}
	The second term vanishes as $n\to +\infty$ thanks to \eqref{P1}, while to deal with the first one we use \ref{hyp:E6}: we first fix $\lambda>0$ and we pick $R=C_\Lambda |\piz|_*$, where $C_\Lambda$ is the constant provided by  Corollary~\ref{unifbound}. Then let $\delta$ be given accordingly by \ref{hyp:E6}. By means of \eqref{finezza} we know that $\max\limits_{k=1,\dots,n}\left|t_k-t_{k-1}\right|\le\delta$ for $n$ large enough, thus \ref{hyp:E6} implies:
	\begin{equation*}
		\sum_{k=1}^{n}\int_{t_{k-1}}^{t_k}\left|\frac{\partial}{\partial t}\mc E(\tau,x^+(t_k))-\frac{\partial}{\partial t}\mc E(t_k,x^+(t_k))\right|\d\tau\le\lambda(t-s),
	\end{equation*}
	and hence \eqref{UEE+} is proved.\par
	Inequality \eqref{UEE-} can be obtained arguing in the same way replacing $x^+$ with $x^-$, and recalling that \ref{GSminus} holds true only if $t>0$. If in addition $x_0$ satisfies \eqref{eq:adm_quasistat}, then \ref{GSminus} holds true also in $t=0$ and the whole argument can be performed also in $s=0$.
\end{proof}
We want to point out that condition \ref{hyp:E6} is not necessary for the validity of Lemma~\ref{lemmaUEE}, but it is useful to treat the case of a general elastic energy. Indeed, if we restrict for instance our attention to the concrete case of a quadratic energy $\mc E_\mathrm{sh}(t,z)=\frac 12\langle\mathbb{A}_\mathrm{sh}(z-\ell_\mathrm{sh}(t)),z-\ell_\mathrm{sh}(t)\rangle_Z$ as in \ref{hyp:QE}, it is easy to verify that conditions \ref{hyp:E1}--\ref{hyp:E5} are satisfied, but \ref{hyp:E6} does not hold true if $\dot{\ell}_\mathrm{sh}$ is not continuous. However, Lemma~\ref{lemmaUEE} is still valid.

\begin{lemma}\label{lemma: QE}
If in Lemma~\ref{lemmaUEE} assumption  \ref{hyp:E6} is replaced by \ref{hyp:QE}, the same conclusions hold.
\end{lemma}
\begin{proof}
The proof follows the same strategy used for Lemma~\ref{lemmaUEE}, with some adaptations. Firstly, we need to choose fine partitions satisfying instead:
\begin{subequations}
	\begin{equation}\label{Pa}
		\lim\limits_{n\to +\infty}\sum_{k=1}^{n}(t_k-t_{k-1})\langle\mathbb{A}_\mathrm{sh}(\piz (x^+(t_k)){-}\ell_\mathrm{sh}(t_k)),\dot{\ell}_\mathrm{sh}(t_k)\rangle_Z=\!\!\int_{s}^{t}\!\!\!\langle\mathbb{A}_\mathrm{sh}(\piz (x(\tau)){-}\ell_\mathrm{sh}(\tau)),\dot{\ell}_\mathrm{sh}(\tau)\rangle_Zd\tau;
	\end{equation}
	\begin{equation}\label{Pb}
		\lim\limits_{n\to +\infty}\sum_{k=1}^{n}\left|(t_k-t_{k-1})\dot{\ell}_\mathrm{sh}(t_k)-\int_{t_{k-1}}^{t_k}\dot{\ell}_\mathrm{sh}(\tau)\d\tau\right|_Z=0.
	\end{equation}
\end{subequations}
As before, the existence of such a sequence of partitions is ensured by \cite[Lemma~4.5]{FrancMielk}. In this case the integral term $I_n$ defined in \eqref{In} takes the form:
\begin{equation*}
	I_n=-\sum_{k=1}^{n}\int_{t_{k-1}}^{t_k}\langle\mathbb{A}_\mathrm{sh}(\piz (x^+(t_k))-\ell_\mathrm{sh}(\tau)),\dot{\ell}_\mathrm{sh}(\tau)\rangle_Z d\tau,
\end{equation*}
and we conclude if we prove that $\displaystyle\lim\limits_{n\to +\infty} I_n=-\int_{s}^{t}\langle\mathbb{A}_\mathrm{sh}(\piz (x(\tau))-\ell_\mathrm{sh}(\tau)),\dot{\ell}_\mathrm{sh}(\tau)\rangle_Zd\tau$. With this aim we rewrite $I_n$ as:
\begin{align*}
	I_n=&-\sum_{k=1}^{n}(t_k-t_{k-1})\langle\mathbb{A}_\mathrm{sh}(\piz (x^+(t_k))-\ell_\mathrm{sh}(t_k)),\dot{\ell}_\mathrm{sh}(t_k)\rangle_Z\\
	&+\sum_{k=1}^{n}\left\langle\mathbb{A}_\mathrm{sh}(\piz (x^+(t_k))-\ell_\mathrm{sh}(t_k))\,,\,(t_k-t_{k-1})\dot{\ell}_\mathrm{sh}(t_k)-\int_{t_{k-1}}^{t_k}\dot{\ell}_\mathrm{sh}(\tau)\d\tau\right\rangle_Z\\
	&+\sum_{k=1}^{n}\int_{t_{k-1}}^{t_k}\langle\mathbb{A}_\mathrm{sh}(\ell_\mathrm{sh}(t_k)-\ell_\mathrm{sh}(\tau)),\dot{\ell}_\mathrm{sh}(\tau)\rangle_Z\d\tau=:J^1_n+J^2_n+J^3_n.
\end{align*}
By means of \eqref{Pb} it is easy to see that $\lim\limits_{n\to +\infty}J^2_n=0$, while exploiting the absolute continuity of $\ell_\mathrm{sh}$ together with \eqref{finezza} we also deduce that $\lim\limits_{n\to +\infty}J^3_n=0$. By using \eqref{Pa} we conclude. 
\end{proof}

As a simple corollary we get:

\begin{prop}[\textbf{Upper Energy Estimate}]\label{propUEE}
	Assume \ref{hyp:E6} or \ref{hyp:QE}, and assume that $x^\eps_0$ and $\eps x^\eps_1$ are uniformly bounded. Then the limit function $x$ obtained in Theorem~\ref{convsubseq} fulfils the following inequality for every $0< s\le t\le T$:
	\begin{equation}\label{UEE}
		\mc E(t,x^+(t))+\min\big\{V_{\mc R}(x^+;s-,t),V_{\mc R}(x^-;s,t+)\big\}\ge \mc E(s,x^-(s))+\int_{s}^{t}\frac{\partial}{\partial t}\mc E(\tau,x(\tau))\d\tau.
	\end{equation}
	If in addition $x_0=x(0)$ satisfies \eqref{eq:adm_quasistat}, then it also holds:
	\begin{equation}\label{UEE0}
		\mc E(t,x^+(t))+V_{\mc R}(x^-;0,t+)\ge \mc E(0,x_0)+\int_{0}^{t}\frac{\partial}{\partial t}\mc E(\tau,x(\tau))\d\tau,\quad\text{ for every }t\in [0,T].
	\end{equation}
\end{prop}
\begin{proof}
	We fix $0<s\le t\le T$ and we consider two sequences $s_k\nearrow s$ and $t_k\searrow t$. By means of \eqref{UEE+} and \eqref{UEE-} we thus deduce:
	\begin{equation*}
		\mc E(t,x^+(t))+V_{\mc R}(x^+;s_k,t)\ge \mc E(s_k,x^+(s_k))+\int_{s_k}^{t}\frac{\partial}{\partial t}\mc E(\tau,x(\tau))\d\tau,
	\end{equation*}
	\begin{equation}\label{s0}
		\mc E(t_k,x^-(t_k))+V_{\mc R}(x^-;s,t_k)\ge \mc E(s,x^-(s))+\int_{s}^{t_k}\frac{\partial}{\partial t}\mc E(\tau,x(\tau))\d\tau.
	\end{equation}
	Letting $k\to +\infty$ and since $\mc E$ is continuous in $[0,T]\times X$ we obtain \eqref{UEE}.\par 
	If in addition $x_0$ satisfies \eqref{eq:adm_quasistat} we can set $s=0$ in \eqref{s0}, thus also \eqref{UEE0} follows by letting $k\to +\infty$.
\end{proof}
Combining all the results of this section we are finally able to prove that the limit function $x$ is actually an energetic solution of the quasistatic problem \eqref{quasprob}. The rigorous statement is the following: 
\begin{thm}\label{almostfinalthm}
	Assume \ref{hyp:E6} or \ref{hyp:QE}, and assume that $x^\eps_0$ and $\eps x^\eps_1$ are uniformly bounded. Then the limit function $x$ obtained in Theorem~\ref{convsubseq} is continuous in $(0,T]$ and its right limit $x^+$ is an energetic solution for \eqref{quasprob} with initial position $x^+(0)$ in the sense of Definition~\ref{defenergetic}.
	\par 
	If in addition $x_0=x(0)$ satisfies \eqref{eq:adm_quasistat} and $\lim\limits_{\eps\to 0}\eps x^\eps_1=0$, then $x$ is continuous also in $t=0$ and it is an energetic solution for \eqref{quasprob} with initial position $x_0$.
\end{thm}
\begin{proof}
	We first prove that the right limit $x^+$ is an energetic solution for \eqref{quasprob} with initial position $x^+(0)$. We only need to prove the weak energy balance \ref{WEB}, since we already know $x^+$ is globally stable, see \ref{GSplus}. With this aim we first fix $t\in [0,T]$ and by combining \eqref{LEEb} and \eqref{UEE+} we get:
	\begin{align*}
	\mc E(t,x^+(t))+V_{\mc R}(x;0+,t+)&\le \mc E(0,x^+(0))+\int_{0}^{t}\frac{\partial}{\partial t}\mc E(\tau,x(\tau))\d\tau\le 	\mc E(t,x^+(t))+V_{\mc R}(x^+;0,t)\\
	&\le 	\mc E(t,x^+(t))+V_{\mc R}(x^+;0,t+).
	\end{align*}
	By means of \ref{var:d} in Proposition~\ref{propertiesvariationtime} we hence deduce that $V_{\mc R}(x;0+,t+)=V_{\mc R}(x^+;0,t+)=V_{\mc R}(x^+;0,t)$ and also the validity of \ref{WEB}:
	\begin{equation*}
	\mc E(t,x^+(t))+V_{\mc R}(x^+;0,t)= \mc E(0,x^+(0))+\int_{0}^{t}\frac{\partial}{\partial t}\mc E(\tau,x(\tau))\d\tau,\quad\text{ for every } t\in[0,T].
	\end{equation*}
	Thus $x^+$ is an energetic solution starting from $x^+(0)$ and in particular, by means of Proposition~\ref{regularity}, it is continuous in $[0,T]$ with continuous $\mc R$-variation $V_{\mc R}(x^+;0,\cdot)$.\par
	We now show that $x(t)=x^+(t)$ for every $t\in (0,T]$. By means of \eqref{LEEa} and \eqref{UEE} and reasoning as before we get:
	\begin{equation*}
	V_{\mc R}(x;t-,t+)=V_{\mc R}(x^+;t-,t),\quad\text{ for every }t\in (0,T].
	\end{equation*}
	Since $x^+$ has continuous $\mc R$-variation, we deduce that $V_{\mc R}(x;t-,t+)=V_{\mc R}(x^+;t-,t)=0$ if $t\in(0,T]$; this implies that the $\mc R$-variation of $x$ is continuous in $(0,T]$, and thus in particular $x$ itself is continuous in $(0,T]$ (see \ref{var:c} in Proposition~\ref{propertiesvariationtime}). This means in particular that $x(t)=x^+(t)$ for every $t\in(0,T]$.\par 
	If in addition $x_0$ satisfies \eqref{eq:adm_quasistat} and $\lim\limits_{\eps\to 0}\eps x^\eps_1=0$, then we can use \eqref{LEEa} in $s=0$ and \eqref{UEE0}; since we now know that both $x$ and $V_\mc R(x;0,\cdot)$ are continuous in $(0,T]$, arguing as before we obtain:
	\begin{equation*}
	\mc E(t,x(t))+V_{\mc R}(x;0,t)= \mc E(0,x_0)+\int_{0}^{t}\frac{\partial}{\partial t}\mc E(\tau,x(\tau))\d\tau,\quad\text{ for every }t\in(0,T].
	\end{equation*}
	Since the above equality is trivially satisfied in $t=0$, we deduce that $x$ satisfies \ref{WEB}; since \eqref{eq:adm_quasistat} holds, from \ref{GSminus} we also deduce that $x$ satisfies \ref{GS}, and thus it is an energetic solution for \eqref{quasprob} with initial position $x_0$. Thus we conclude.
\end{proof}
 We conclude this section by stating the main theorem of the paper, which gathers and summarises what we have proved up to now about the convergence of dynamic solutions of problem \eqref{dynprob} to quasistatic solutions of \eqref{quasprob} when inertia vanishes.
\begin{thm}\label{finalthm}
	Let $\M,\V$ be as in Section~\ref{sec:Setting}, and assume that $\mc R$ satisfies \ref{hyp:RK}, and that $\mc E(t,x)=\mc E_\mathrm{sh}(t,\piz (x))$ satisfies \ref{hyp:E1}--\ref{hyp:E6} or \ref{hyp:QE}. For every $\eps>0$, let $\xeps$ be a differential solution of the dynamic problem \eqref{dynprob} related to the initial position $x^\eps_0\in X$ and the initial velocity $ x^\eps_1\in K$, and assume that $x^\eps_0$ and $\eps x^\eps_1$ are uniformly bounded. Then there exist a subsequence $\epsj\searrow 0$ and a function $x\in BV_{\mc R}([0,T];X)\cap \CC^0((0,T];X)$ such that its right limit $x^+$ is an energetic solution for \eqref{quasprob} in the sense of Definition~\ref{defenergetic} with initial position $x^+(0)$ and:
	\begin{enumerate}[label=\textup{(\alph*')}]
		\item \label{thm:a}$\lim\limits_{j\to +\infty}\xepsj(t)=x(t)$ for every $t\in[0,T]$, and the convergence is uniform in any compact interval contained in $(0,T]$;
		\item\label{thm:b} $\displaystyle\lim\limits_{j\to +\infty}\int_{s}^{t}\mc R(\tau,\xepsjd(\tau))\d\tau=V_\mc R(x;s,t)$ for every $0< s\le t\le T$, and the convergence is uniform in $[s,T]$;
		\item\label{thm:c} $\lim\limits_{j\to +\infty}\epsj |\xepsjd(t)|_{\mathbb{M}}=0$ for every $t\in(0,T]$, and the convergence is uniform in any compact interval contained in $(0,T]$;
		\item \label{thm:d}$\displaystyle\lim\limits_{j\to +\infty}\epsj\int_{s}^{T}|\xepsjd(\tau)|^2_\mathbb{V}\d\tau=0$ for every $0< s\le T$.
	\end{enumerate}
If in addition $x_0:=x(0)$ satisfies \eqref{eq:adm_quasistat}, namely $\mc E(0,x_0)\le\mc E(0,v)+\mc R(0,v-x_0)$ for every $v\in X$, and $\lim\limits_{\eps\to 0}\eps x^\eps_1=0$, then the limit function $x$ is continuous in the whole $[0,T]$, and it is an energetic solution of \eqref{quasprob} with initial position $x_0$; moreover the convergence in \ref{thm:a} and \ref{thm:c} is uniform in the whole $[0,T]$, while \ref{thm:b} and \ref{thm:d} hold true also in $s=0$.\par
Finally, if also \ref{hyp:R5} holds or  if $\mc R$ does not depend on time, then $x$ is actually $\mc R$-absolutely continuous in $[0,T]$, and thus a differential solution of \eqref{quasprob}.
\end{thm}
\begin{rmk}[\textbf{Uniqueness}]
    If in particular one of the assumptions of Lemma~\ref{lemma:uniqquas} or Lemma \ref{lemma:uniqquas2} is satisfied, and if $\lim\limits_{\eps\to 0}\eps x_1^\eps=0$ and  $\lim\limits_{\eps\to 0}x_0^\eps=x_0$, for some $x_0$ satisfying \eqref{eq:adm_quasistat}, then there is no need to pass to a subsequence in the previous theorem. Indeed in this case the whole sequence $x^\eps$ converges in the sense of \ref{thm:a}--\ref{thm:d} (even in $t=0$) towards the unique differential solution $x$ to \eqref{quasprob}.
\end{rmk}
\begin{proof}[Proof of Theorem~\ref{finalthm}]
	Combining Theorems~\ref{convsubseq}, \ref{almostfinalthm} and exploiting Proposition~\ref{regularity} we get the existence of a subsequence $\epsj\searrow 0$ and of a function $x\in BV_{\mc R}([0,T];X)\cap \CC^0((0,T];X)$ with the property that the right limit $x^+$ is an energetic solution for \eqref{quasprob} with initial position $x^+(0)$ and for which the pointwise convergence in \ref{thm:a} and \ref{thm:c} hold. We now observe that by the energy balances (EB$^\epsj$) and \ref{WEB} for every $0<s\le t\le T$ we have:
	\begin{align}\label{estunif}
	&\quad\,\epsj\int_{s}^{t}|\xepsjd(\tau)|^2_\mathbb{V}\d\tau+\int_{s}^{t}\mc R(\tau,\xepsjd(\tau))\d\tau-V_\mc R(x;s,t)\nonumber\\
	&=\frac{\eps^2_j}{2}|\xepsjd(s)|^2_\mathbb{M}-\frac{\eps^2_j}{2}|\xepsjd(t)|^2_\mathbb{M}+\mc E(s,\xepsj(s))-\mc E(s,x(s))+\mc E(t,x(t))-\mc E(t,\xepsj(t))\\
	&\quad+\int_{s}^{t}\Big(\frac{\partial}{\partial t}\mc E(\tau,\xepsj(\tau))-\frac{\partial}{\partial t}\mc E(\tau,x(\tau))\Big)\d\tau.\nonumber
	\end{align}
	By means of the pointwise convergence in \ref{thm:a} and \ref{thm:c} and recalling \ref{hyp:E5} we deduce that the right-hand side of the above inequality vanishes as $j\to+\infty$. Thus the pointwise convergence in \ref{thm:b} and \ref{thm:d} easily follows, since by \ref{conv:b} in Theorem~\ref{convsubseq} we already know that \begin{equation*}
	    \displaystyle \liminf\limits_{j\to +\infty}\left(\int_{s}^{t}\mc R(\tau,\xepsjd(\tau))\d\tau-V_\mc R(x;s,t)\right)\ge 0.
	\end{equation*}
	By means of Lemma~\ref{uniformconv} we now deduce that the convergence in \ref{thm:a} is uniform in any compact interval contained in $(0,T]$, while the uniform convergence in \ref{thm:b} is due to the standard result that a sequence of nondecreasing and continuous scalar functions pointwise converging to a continuous function on a compact interval actually converges uniformly. The uniform convergence in \ref{thm:c} now follows by rearranging equality \eqref{estunif} and by exploiting \ref{hyp:E3}, \ref{hyp:E5} and the just obtained uniform convergence in \ref{thm:a}, \ref{thm:b} and \ref{thm:d}.\par 
	If in addition $x_0$ satisfy \eqref{eq:adm_quasistat} and $\lim\limits_{\eps\to 0}\eps x^\eps_1=0$, we know by Proposition~\ref{regularity} and Theorem~\ref{almostfinalthm} that $x$ is continuous in $[0,T]$ and it is an energetic solution with initial position $x_0$. Arguing as before we obtain the uniform convergence in $[0,T]$ for \ref{thm:a} and \ref{thm:c} and the validity of \ref{thm:b} and \ref{thm:d} also in $s=0$.\par 
	To conclude, if \ref{hyp:R5} holds or  if $\mc R$ does not depend on time, always by means of Proposition~\ref{regularity} we deduce that $x$ is $\mc R$-absolutely continuous in $[0,T]$.
\end{proof}

We want to point out that our result is sharp, in the sense that, without additional assumptions, no better kind of convergence (for instance in $W^{1,1}$) can be achieved in the quasistatic limit. It is enough to consider the simplest case $X=Z=\R$, with $\mathbb{M}=\Id$, $\mathbb{V}=0$, dissipation potential $\RR(t,v)=|v|$ and a quadratic elastic energy $\EE(t,x)=\frac 12 (x-t-1)^2$. Indeed it is easy to verify that in this setting the unique differential solution of the dynamic problem \eqref{dynprob}, with initial position $x_0^\eps=0$ and initial velocity $x_1^\eps=2$, is the function
\begin{equation*}
    \xeps(t)=t+\eps\sin\left(\frac t\eps\right),
\end{equation*}
 which of course converges as $\eps\to 0^+$ towards $x(t)=t$, namely the unique differential solution of the quasistatic problem \eqref{quasprob} with initial position $x_0=0$, in the sense of previous theorem.\par 
 However $x^\eps$ does not converge to $x$ in $W^{1,1}(0,T)$, indeed
 \begin{equation*}
     \int_{0}^{T}\abs{\xepsd(\tau)-\dot{x}(\tau)}\d\tau=\int_{0}^{T}\abs{\cos\left(\frac \tau\eps\right)}\d\tau,
 \end{equation*}
 which does not vanish as $\eps\to 0^+$.
	\bigskip

	

	\section{Applications and examples} \label{sec:models}
	
	In this last section we illustrate several examples which can be described by our abstract formulation; in particular they explain and motivate our framework. 	Since the applications we present here are all set in $X=\R^N$, endowed with the euclidean norm, for simplicity we will always identify canonically the dual space $X^*$ with $\R^N$, so that the dual coupling $\scal{\cdot}{\cdot}$ coincides with the scalar product.
	
	\begin{figure}
	\begin{center}
\begin{tikzpicture}[line cap=round,line join=round,>=triangle 45,x=1.0cm,y=1.0cm, line width=0.6pt, scale=1]
\clip(-1.,1.5) rectangle (9.,5.5);
\draw (8.,2.)-- (0.,2.);
\draw[dashed] (0.,3.)-- (8.,3.);
\draw[decoration={aspect=0.45, segment length=2mm, amplitude=2mm,coil},decorate] (1.5,3.)-- (7.,3.) ;
\draw [fill=black] (1.7,3.)--(1.7,3.1)--(1.3,3.1)--   (1.3,3.)--(1.5,2)--(1.7,3.);
\draw [fill=black] (1.3,3.1) -- (1.7,3.1) arc(0:180:0.2) --cycle;
\draw [line width=1pt, fill=white] (7.,3.) circle (3pt);
\fill [pattern = north east lines] (0.,2) rectangle (8.,1.6);
\draw (1.5,3.6) node {$x(t)$};
\draw (7,3.6) node {$p(t)$};
\end{tikzpicture}
	\end{center}
\caption{A mechanical model of the scalar play operator, discussed in Subsection~\ref{subsec:play}.}
\label{fig:play}
\end{figure}
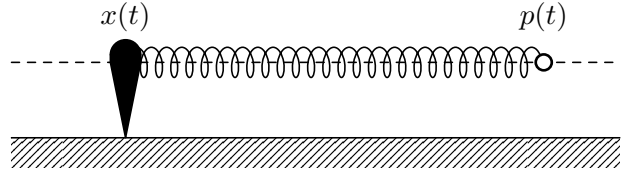

\subsection{The minimal example: the play operator} \label{subsec:play}
To gently introduce the Reader to our examples, we begin by presenting a very simple model,  illustrated in Figure \ref{fig:play}. We have a mass $m>0$ with position $x(t)$ on a line, and subject to (isotropic) dry friction. The mass is connected to a (linear) spring, whose other end is moved according to the function $p(t)\in W^{1,1}(0,T)$. Thus the dynamic evolution of the system is described by the inclusion \eqref{dynprob}, where:
\begin{align*}
	X=Z=K=\R, && \RR(t,v)=\RR(v)=\alpha\abs{v}, && \EE(t,x)=\EE_\mathrm{sh}(t,x)=\frac k2(x-p(t)+L^\mathrm{rest})^2,
\end{align*}
and $\piz$ is the identity. Notice that \ref{hyp:QE} holds. Clearly $\M=m>0$, while we may assume either $\mathbb{V}=0$, or add an additional viscous resistance to $\dot x$, so that the resulting friction force-velocity law for the mass is of Bingham type.

The relevance of this model is due to the fact that its quasistatic evolution corresponds to the \emph{(scalar) play operator} \cite{Krejbook}; indeed a straightforward computation shows that \eqref{quasprob} in this case reads as
\begin{equation}\label{ex1:quas}
	\begin{cases}
	    p(t)-L^\mathrm{rest}-x(t)\in \frac{\alpha}{k}\,\partial \abs{\dot x(t)},\\
	    x(0)=x_0,
	\end{cases}
\end{equation}
and hence, setting $u(t)=p(t)-L^\mathrm{rest}$, we notice that \eqref{ex1:quas} is equivalent to
\begin{equation*}
\begin{cases}
 \abs{u(t)-x(t)}\leq \frac{\alpha}{k},\\
(u(t)-x(t)-v)\,\dot x(t)\geq 0, & \text{for every $v\in \left[-\frac{\alpha}{k},\frac{\alpha}{k}\right]$},\\
x(0)=x_0.
\end{cases}
\end{equation*}
More advanced models may be built by considering analogously a mass on a plane (or abstractly in an $N$-dimensional space), or considering nonautonomous friction coefficients. Such quasistatic systems may be advantageously expressed as a sweeping process: we comment the meaning of the dynamic approximation in such formulation in Subsection~\ref{subsec:sp}.

\subsection{Soft crawlers} \label{subsec:crawler}

We now illustrate minutely how the family of models represented in Figure~\ref{fig:crawler4} and described in Section~\ref{sec:intro} fits in our mathematical framework. 
Their quasistatic version has been extensively discussed in \cite{Gid18}, to which we refer for more details. We also mention \cite{BSZZB17}, where similar models have been studied in the dynamic case.

We are considering a model with $N\geq 2$ blocks on a line, with adjacent blocks joined by an actuated soft link. We describe with $x_i$ the position of the $i$-th block. The elastic energy of the system will not depend directly on any of the positions of the block, but only on the distances $x_i-x_{i-1}$ between two consecutive blocks. Hence we set
\begin{align*}
X=\R^N, &&Z=\R^{N-1}, && \piz(x_1,\dots,x_N)=(x_2-x_1,\dots,x_N-x_{N-1}).
\end{align*}
We now discuss separately each of the elements of the dynamics.
\subsubsection*{Mass distribution} Denoting with $m_i>0$ the mass of the $i$-th block, the linear operator $\M$ is
\begin{equation*}
	\M=\diag(m_1,\dots,m_N).
\end{equation*}

\subsubsection*{Viscous friction}
There are two main situation in which we may consider viscous friction. The first one is 
to assume an additional viscous friction resistance when the blocks slide, in addition to dry friction we discuss below. Such forces are described by a diagonal matrix
\begin{equation*}
\mathbb{V}_\mathrm{ext}=\diag(\nu^\mathrm{ext}_1,\dots,\nu^\mathrm{ext}_N),
\end{equation*}
for some nonnegative coefficients $\nu^\mathrm{ext}_i\geq0$.  This also means that the total friction force acting on each block is of Bingham type, and may be justified by lubrication with a non-Newtonian fluid \cite{DeSGNT}.
 
The second possible way to introduce viscosity in the model is to assume a viscous resistance to deformation in the links. This is represented by the matrix 
\begin{equation*}
\mathbb V_\mathrm{link}=\begin{pmatrix} 
	\nu^\mathrm{link}_1  & -\nu^\mathrm{link}_1 & 0 & \cdots & 0 & 0\\
	-\nu^\mathrm{link}_1 &  \nu^\mathrm{link}_1+\nu^\mathrm{link}_2 & -\nu^\mathrm{link}_2& \cdots & 0 & 0\\
	0  & -\nu^\mathrm{link}_2 & \nu^\mathrm{link}_2+\nu^\mathrm{link}_3 & \cdots & 0 & 0\\
	\vdots & \vdots & \vdots & \ddots & \vdots & \vdots \\
	0 & 0 & 0 & \cdots & \nu^\mathrm{link}_{N-2}+\nu^\mathrm{link}_{N-1} & -\nu^\mathrm{link}_{N-1}\\
	0 & 0 & 0 & \cdots & -\nu^\mathrm{link}_{N-1} & \nu^\mathrm{link}_{N-1}\\
	\end{pmatrix}
\end{equation*}
for some nonnegative coefficients $\nu^\mathrm{link}_i\geq0$. 

Accounting for these two effects, a general viscosity matrix $\mathbb V$ takes the form $	\mathbb V=\mathbb V_\mathrm{link}+\mathbb V_\mathrm{ext}$.

\subsubsection*{Dry friction}
Since each block is affected independently by dry friction, the rate-independent dissipation potential can be represented as the sum
\begin{equation*}
\RR_\mathrm{finite}(t,v)=\sum_{i=1}^N\RR_i(t,v_i),
\end{equation*}
of $N$ dissipation potentials $\RR_i\colon[0,T]\times \R \to [0,+\infty)$, each of the form
\begin{equation}\label{eq:Rfinexample}
\RR_i(t,v)=\begin{cases}\mu_i^+(t)v, &\text{if $v\geq 0$},\\
\mu_i^-(t)v, &\text{if $v\leq 0$},
\end{cases}
\end{equation}
where the functions $\mu_i^\pm\colon[0,T]\to (0,+\infty)$ are strictly positive and absolutely continuous. Concretely, it means that each block has two dry friction coefficients, one for forward and one for backward movements, possibly varying in time. By compactness, we observe that in this framework the assumptions \ref{hyp:R1}--\ref{hyp:R3} are satisfied. As argued in \cite[Lemma 3.2]{Gid18}, the uniqueness condition \ref{hyp:star} of Lemma~\ref{lemma:uniqquas2} for the quasistatic problem is satisfied if, for every subset of indices $J\subseteq\{1,2,\dots,N\}$ we have
\begin{equation}
\sum_{i\in J}\mu_i^+(t)\neq \sum_{i\in J^C}\mu_i^-(t), \qquad\text{for a.e.~$t\in[0,T]$},
\end{equation}
where $J^C=\{1,2,\dots,N\}\setminus J$.

\subsubsection*{Velocity constraint}
Most of the models of crawlers usually fit in the $K=X$ case: indeed, the possibility to move the body both backwards and forwards is often appreciable in locomotion. In some situations, however, backward friction is extremely higher than forward friction, so that in fact no backwards movement occurs. For this reason, sometimes it is convenient to assume an infinite friction coefficient, namely a constraint on velocities. With our notation, this corresponds to set 
\begin{equation*}
K=\bigcap_{i=1}^N K^+_i, \qquad\qquad\text{where} \quad K_i^+=\{v\in \R^N \mid v_i\geq 0\}.
\end{equation*}
We observe that the set $K$ is a polyhedral cone, satisfying condition \ref{cond:R5char3} of Proposition~\ref{prop:R5char}. Notice also that, in this case, the coefficients $\mu_i^-$ in \eqref{eq:Rfinexample} can be freely chosen, for instance equal to a positive constant, since they are not involved in the dynamics.
More generally, we can introduce analogously the halfplanes $K_i^-=\{v\in \R^N \mid v_i\leq 0\}$, and set $K$ as the intersection of an arbitrary selection of sets $K_i^\pm$, although this would result often in something less pragmatical in terms of locomotion.
In particular, if $K\subseteq K_i^+\cap K^-_i$, the $i$-th block would be completely anchored on the surface.

\subsubsection*{Elastic energy}
The total elastic energy will be the sum of the elastic energies of each link. Hence we have
\begin{align*}
 \EE(t,x)=\sum_{i=1}^{N-1}\EE^\mathrm{link}_i(t,x_{i+1}-x_i), &&\text{or equivalently} && \EE_\mathrm{sh}(t,z)=\sum_{i=1}^{N-1}\EE^\mathrm{link}_i(t,z_i).
\end{align*}
In order for $\EE_\mathrm{sh}$ to satisfy any of the properties \ref{hyp:E1}--\ref{hyp:E7}, it is sufficient to ask each of the energies $\EE^\mathrm{link}_i\colon [0,T]\times \R\to [0,+\infty)$  of the links to satisfy the same condition being required on $\EE_\mathrm{sh}$. The quadratic case \ref{hyp:QE} corresponds to the case in which each of the link energies is quadratic, namely it follows Hooke's law
\begin{equation*}
\EE^\mathrm{link}_i(t,z_i)=\frac{k_i}{2}\,\left(z_i-\ell_i(t)\right)^2,
\end{equation*}
for a positive elastic constant $k_i>0$ and an absolutely continuous $\ell_i\colon[0,T]\to \R$. Notice that our results hold also for nonlinear models of elasticity.
For instance, the soft link may behave like a Duffing-type nonlinear spring, i.e.
\begin{equation*}
\EE^\mathrm{link}_i(t,z_i)=\frac{k_i}{2}\,\left(z_i-\ell_i(t)\right)^2+\frac{\beta_i}{4}\left(z_i-\ell_i(t)\right)^4,
\end{equation*} 
where the quartic term produces a hardening of the spring. In such a case the assumptions \ref{hyp:E1}--\ref{hyp:E5} and \ref{hyp:E7} are all satisfied. Pay attention that \ref{hyp:E6} holds only if $\ell_i$ are continuosly differentiable; however in this specific example one can argue as in Lemma~\ref{lemma: QE}, thus \ref{hyp:E6} is not really necessary.


\subsection{A rheological model} \label{subsec:rheo}
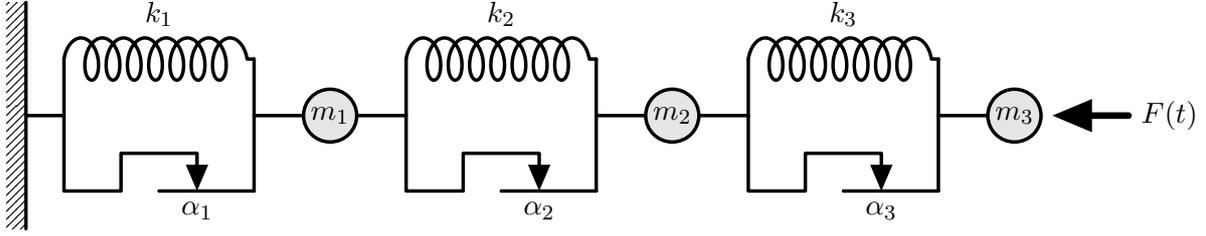
\begin{figure}
\begin{center}
    \begin{tikzpicture}[line cap=round,line join=round,>=triangle 45,x=5mm,y=5mm,line width=1.3pt]
\clip(-0.5,-4.) rectangle (31.,4.);
\draw [] (0,3)-- (0,-3);
\fill [pattern = north east lines] (-0.5,-3) rectangle (0.,3);
\draw [] (0,0)-- (1,0);
\draw [->] (1,1.5)-- (1,-2)-- (2.5,-2)-- (2.5,-1)-- (4.5,-1)-- (4.5,-2)node[anchor=north]{$\alpha_1$};
\draw [] (3.5,-2)-- (6,-2);
\draw [] (6,-2)-- (6,1.5);
\draw[decoration={aspect=0.55, segment length=2.8mm, amplitude=2.8mm,coil},decorate] (1.,1.5)-- (6.,1.5);
\draw [] (6,0)-- (7.3,0);
\draw [fill=gray!20] (8,0) circle (0.7)node[]{$m_1$};
\draw [] (8.7,0)-- (10,0);
\draw [->] (10,1.5)-- (10,-2)-- (11.5,-2)-- (11.5,-1)-- (13.5,-1)-- (13.5,-2)node[anchor=north]{$\alpha_2$};
\draw [] (12.5,-2)-- (15,-2);
\draw [] (15,1.5)-- (15,-2);
\draw[decoration={aspect=0.55, segment length=2.8mm, amplitude=2.8mm,coil},decorate] (10.,1.5)-- (15.,1.5);
\draw [] (15,0)-- (16.3,0);
\draw [fill=gray!20] (17,0) circle (0.7) node[]{$m_2$};
\draw [] (17.7,0)-- (19,0);
\draw [->] (19,1.5)-- (19,-2)-- (20.5,-2)-- (20.5,-1)-- (22.5,-1)-- (22.5,-2)node[anchor=north]{$\alpha_3$};
\draw [] (21.5,-2)-- (24,-2);
\draw [] (24,1.5)-- (24,-2);
\draw[decoration={aspect=0.55, segment length=2.8mm, amplitude=2.8mm,coil},decorate] (19.,1.5)-- (24.,1.5);
\draw [] (24,0)-- (25.3,0);
\draw [fill=gray!20] (26,0) circle (0.7)node[]{$m_3$};
\draw [line width=2.5pt,<-] (27,0)-- (29,0) node[anchor=west]{$F(t)$};
\draw (3.5,2.7) node{$k_1$};
\draw (12.5,2.7) node{$k_2$};
\draw (21.5,2.7) node{$k_3$};
\end{tikzpicture}
\end{center}
	\caption{A rheological model discussed in Subsection~\ref{subsec:rheo}, cf. also \cite[Sec.~2.2.6]{BSL00}}
	\label{fig:rheo}
\end{figure}

In order to illustrate a second example with multiple material points, we propose here, with our notation, a rheological model presented in \cite[Sec.~2.2.6]{BSL00}, and illustrated in Figure \ref{fig:rheo} for $N=3$.

The model consists on $N$ material points and $N$ $P_i$-elements connected in series. A $P_i$ element is composed of a St-Venant element with threshold $\alpha_i>0$ and a linear spring with constant $k_i>0$ connected in parallel. 
As before, we denote with $x_i$ the position on the line of the $i$-th material point, having mass $m_i>0$. The first $P_i$-element is connected to the first material point at one end, whereas the other end is fixed in the origin. Moreover, the $N$-th material point is subject to an external force $F(t)$, absolutely continuous in time. 
Hence
\begin{align*}
X=Z=K=\R^N, && \piz=\Id, && \M=\diag(m_1,\dots,m_N).
\end{align*}
The energy $\EE$ will be the sum of a potential energy $F(t)x_N$ used to describe the external force, plus the elastic energies of the $P_i$-elements, namely:
\begin{equation*}
\EE(t,x)=\EE_\mathrm{sh}(t,x)=F(t)x_N+\frac{k_1}{2}\, x_1^2+\sum_{i=2}^N \frac{k_i}{2}(x_i-x_{i-1})^2.
\end{equation*}
Similarly, the dissipation potential $\mc R$ will be the sum of the dissipation potentials associated to each St-Venant element, namely
\begin{equation*}
	\RR(t,v)=\RR(v)=\alpha_1\abs{v_1}+\sum_{i=2}^N\alpha_i\abs{v_i-v_{i-1}},
\end{equation*}
where we recall that in the first $P_i$-element one end is fixed. The assumptions \ref{hyp:E1}-\ref{hyp:E5}, \ref{hyp:E7}, \ref{hyp:RK} are easily verified, as also \ref{hyp:E6} if in addition $F$ is continuously differentiable. As before, however, \ref{hyp:E6} can however be avoided by arguing as in Lemma~\ref{lemma: QE}.

\subsection{A planar model}\label{subsec:plan}

Let us now consider the two-dimensional analogous of the simple model discussed in Subsection~\ref{subsec:play} and illustrated in Figure \ref{fig:play}. Setting for simplicity the rest length of the spring to zero, we have
\begin{align*}
X=Z=\R^2, && \piz=\Id, && \EE(t,x)=\EE_\mathrm{sh}(t,x)=\frac k2\abs{p(t)-x}^2,
\end{align*}
and \ref{hyp:QE} again holds. A point mass at $x$ can be therefore considered as a test particle (or more concretely, the point of a cantilever), probing the frictional properties of the surface. For simplicity, here we limit ourselves to autonomous dissipation. Until now we have presented only models lying on a line, so that the friction forces possibly acting on each mass are described by two parameters $\mu^+$ and $\mu^-$. If instead the test mass lies on a plane, dry friction is described by a function on the unit circle. 
Whereas the isotropic case $\RR(v)=\mu\abs{v}$ is simple, the nature of friction when the surface is anisotropic is a complicated matter.

Experimentally, friction of scaly surfaces, for instance snakes or sharks skins, is usually measured only in four orthogonal directions: forwards, backwards, and the in two transversal directions (usually showing a symmetric behaviour), cf.~e.g.~\cite{BWBG09, Man16}. We are not aware of experimental characterizations of the friction coefficients with respect to all the other intermediate directions. There is however a mathematical restriction on the scenarios that can be effectively described by the subdifferential of a function $\RR$. What we aim to show here is that, by introducing the constraint $K$, we allow to study a qualitatively different class of models, non included in the case $\RR<+\infty$.

If $X=K$, namely there is no velocity constraint, then the functional $\RR$ is continuous by convexity, and so the friction coefficient changes continuously with respect to the direction of the velocity. Moreover, we notice that convexity affects ulteriorly the structure of the friction coefficient: for instance, oscillations arbitrarily both ample and frequent of the friction coefficient as the direction varies are not allowed.

When hooks or scales introduce anisotropic friction on a plane, a scenario that can be expected, or at least desirable, is as follows:
\begin{itemize}
	\item friction is extremely high for all velocities with a nonzero backward component (i.e.~for all $v=(v_1,v_2)$ with $v_1<0$); 
	\item friction is low for all the remaining velocities ($v_1\geq 0$), in particular also for purely lateral velocities ($v_1=0$).
\end{itemize}
If $X=K$, such a case can be portrayed only approximatively, since a smooth transition is compulsory from low to high friction.
The scenario can instead be better described by setting
\begin{equation*}
	K=\{v\in\R^2 \mid v_1\geq 0\}.
\end{equation*}
Indeed, we emphasize that $\RR$ is in general lower semicontinuous, but not continuous, on the boundary of $K$. 

A situation even more radical is usually considered in the modelling of slithering locomotion, with \lq\lq snake in a tube\rq\rq\ models \cite{CicDeS}. While slithering on a plane, snakes experience a very large resistance to transversal sliding, compared to the longitudinal one, so that the whole body of the snake follows the same path covered by its head. 
Hence, according to the description in such models, a test particle on a snake skin would experience:
\begin{itemize}
	\item extremely high friction for all velocities with a nonzero lateral component ($v_2\neq0$); 
	\item high friction for a purely backward velocity ($v_1<0$ and $v_2=0$);
	\item low friction for a purely forward velocity ($v_1>0$ and $v_2=0$).
\end{itemize}
Again, the situation can be portrayed only approximatively by a finite dissipation functional $\RR$, while it is effectively described by introducing the constraint $K$ as
\begin{align*}
	K=\{v\in \R^2\mid v_2=0\}, &&\text{or} &&	K=\{v\in \R^2\mid v_1\geq 0, v_2=0\}.
\end{align*}
Notice that all the three  examples of cones $K$ in this subsection satisfy condition \ref{cond:R5char3} of Proposition~\ref{prop:R5char}.

\subsection{Interpretation as sweeping process} \label{subsec:sp}

In the 70s, Moreau noticed that several mechanical problems of the form \eqref{quasprob} with quadratic energy 
can be fruitfully transformed in the form
\begin{equation}\label{eq:sweep}
\dot y(t)\in -\NN_{C(t)}(y(t)),
\end{equation}
where $\NN_C(y)$ is the normal cone in $y$ with respect to the convex set $C$. Systems of this form are called \emph{sweeping processes}, and present the obvious advantage that the dynamics is expressed in normal form. Vanishing viscosity approximations have played a key role in the study of sweeping processes, not only for characterizing jumps \cite{KreMon}, but also for instance in the derivation of necessary conditions in optimal control  \cite{ArrCol,BK}. One may therefore wonder whether there is any strong connection between the second order sweeping process \eqref{eq:sweep2order} describing the dynamic problem and the first order sweeping process \eqref{eq:sweep} describing the quasistatic problem.
Let us thus recall, briefly, how  \eqref{eq:sweep} can be recovered by \eqref{quasprob}, in the simple case with energy
\begin{equation*}
\mc E(t,x)=\frac 12\langle x-\ell(t),x-\ell(t)\rangle,
\end{equation*}
where $X=Z=\R^N$.  In this case, equation \eqref{quasprob} reads
\begin{equation} \label{eq:sweepincl0}
	-x(t)+\ell(t)\in\partial_v\RR(t,\dot x(t)).
\end{equation}
Now we exploit the convexity of $\RR(t,\cdot)$, so that by the Legendre--Fenchel equivalence \eqref{eq:sweepincl0} is equivalent to
\begin{equation*} 
\dot x(t)\in\partial_v\RR^*(t,-x(t)+\ell(t))=\partial_v\chi_{C_0(t)}(-x(t)+\ell(t))
=\NN_{C_0(t)}(-x(t)+\ell(t)),
\end{equation*}
where $\RR^*(t,\cdot)$ denotes the Legendre transform of $\RR(t,\cdot)$. Since $\RR$ satisfies the properties \ref{prop:regR1} of Corollary~\ref{propertiesR}, then $\RR^*(t,\cdot)$ is exactly the characteristic function of the set $C_0(t):=\partial_v\RR(t,0)$. The change of coordinate $y(t):=-x(t)$ gives \eqref{eq:sweep} with $C(t):=C_0(t)-\ell(t)$.

Unfortunately, the same trick seems quite dispensable for the dynamical problem \eqref{dynprob}. Indeed, it is already in normal form, so that the Legendre transform actually hides the higher order derivative, resulting, for $\mathbb V=0$, in
\begin{equation}\label{eq:dynsweep}
\dot y(t)\in -\NN_{C(t)}(y(t)+\eps^2\M\ddot y(t)).
\end{equation}
Notice that an additional vanishing viscosity $\mathbb V$ can be incorporated with the convex function $\RR(t,\cdot)$ during the Legendre transform, resulting in a smooth approximation of the evolution problem \eqref{eq:dynsweep}.
Hence, the dynamic version \eqref{eq:dynsweep} of \eqref{eq:sweep} must not be confused with the second order sweeping process \eqref{eq:sweep2order}. Indeed, although both are equivalent formulations of the dynamic problem \eqref{dynprob}, in \eqref{eq:sweep2order} the sweeping set $K$ describes only a constraints on the velocities, whereas in \eqref{eq:dynsweep} the sweeping set $C$ accounts both for the rate-independent dissipation and for possible constraints on the velocities. Although the sweeping process therefore seems not to be the most favourable form to consider vanishing inertia approximations, we are confident that advancement in alternative formulations will still benefit the whole theory.

	\subsection{Example of $K$ not satisfying \ref{hyp:R5}.} \label{subsec:notR5}
	As we have seen, in all our mechanical examples the set $K$ satisfies \ref{hyp:R5}. Indeed, we expect this assumption to be usually true in concrete problems. In order to help the Reader understand why, however, it is not automatically satisfied, we present here a -- purely theoretical -- counterexample.  
Let us set $X=\R^3$, $Z=\R^2$, $\piz(x)=(x_2,x_3)$ and
\begin{equation*}
K:=\{(\lambda, \lambda a, \lambda b)\mid \lambda\geq 0,\, a^2+(b-1)^2\leq 1\}.
\end{equation*}
Let us pick $z=(\cos \theta, \sin \theta)$, with $\theta\in(0,\pi/2)$, so that $\abs{z}_Z=1$.
A simple computation shows that
\begin{align*}
	(\lambda,\cos \theta, \sin \theta)\in K \qquad &\text{if and only if}\quad \cos^2\theta + (\sin\theta-\lambda)^2\leq \lambda^2 \quad \text{and $\lambda>0$}\\
	&\text{if and only if}\quad \lambda\sin \theta\geq \frac{1}{2}.
\end{align*}
Hence \ref{hyp:R5} is violated by any sequence $\theta_i\to 0^+$. 

	
	\bigskip
	
	\noindent\textbf{Acknowledgements.}
	The authors are grateful to Professors Gianni Dal Maso and Giovanni Colombo for many helpful discussions on the topic. P.~G.~was partially supported by the GA\v{C}R--FWF grant 19-29646L and the M\v{S}MT\v{C}R grant 8J19AT013. F.~R. is member of the Gruppo Nazionale per l'Analisi Matematica, la Probabilit\'a e le loro Applicazioni (GNAMPA) of the Istituto Nazionale di Alta Matematica (INdAM).
	\bigskip
	
	\bibliographystyle{siam}

	{\small
		\vspace{10pt} (Paolo Gidoni) Czech Academy of Science, Institute for Information Theory and Automation (UTIA), \textsc{Pod vodárenskou veží 4, CZ-182 08, Prague 8, Czech Republic}
		\par 
		\textit{e-mail address}: \textsf{gidoni@utia.cas.cz}
		\par
		\vspace{10pt} (Filippo Riva) SISSA -- International School for Advanced Studies, \textsc{Via Bonomea, 265, 34136, Trieste, Italy}
		\par 
		\textit{e-mail address}: \textsf{firiva@sissa.it}
		\par
	}

\end{document}